\documentclass[10pt]{article}
\usepackage[english]{babel}
\usepackage{comment}
\usepackage[left=2.25cm,
            right=2.25cm,
            top=2.5cm,
            bottom=3cm
            ]{geometry}  
\usepackage{amsmath,amsthm, amssymb}
\usepackage{amsfonts}
\usepackage{bbm} 
\usepackage{graphicx, caption, subcaption}
\usepackage[title]{appendix}

\usepackage[numbers,sort&compress]{natbib}

\usepackage{xcolor}
\usepackage[overload]{empheq}
\usepackage[hidelinks]{hyperref}

\newcommand{\Comment}[1]{}

\newcommand{\lr}[1]{\left(#1\right)}
\newcommand{\slr}[1]{\left[#1\right]}
\newcommand{\glr}[1]{\left\{#1\right\}}
\newcommand{\expect}{\mathbb E }
\newcommand{\avg}[1]{\langle #1 \rangle}
\newcommand{\pder}[2][]{\ensuremath{\frac{\partial#1}{\partial#2}}} 
\newcommand{\der}[2][]{\ensuremath{\frac{d#1}{d#2}}} 
\newcommand{\bb}[1]{\boldsymbol{#1}}
\newcommand{\norm}[1]{\left\lVert#1\right\rVert}

\newtheorem{remark}{Remark}
\newtheorem{theorem}{Theorem}
\newtheorem{proposition}{Proposition}
\newtheorem{corollary}{Corollary}
\newtheorem{definition}{Definition}
\newtheorem{lemma}{Lemma}

\title{On the supra-linear storage in dense networks of grid and place  cells}
\date{}
\author{Adriano Barra\footnote{Dipartimento di Scienze di Base ed Applicate per l'Ingegneria, Sapienza Universit\`a di Roma, Italy $\&$ INFN, Sezione di Roma1 $\&$ CNR Nanotec, Salento Unit, Lecce, Italy.}, Martino S. Centonze\footnote{Dipartimento di Matematica, Universit\`a di Bologna, Italy.}, Michela Marra Solazzo \footnote{Dipartimento di Matematica e Fisica ``Ennio de Giorgi'', Unisalento, Italy.}, Daniele Tantari$^\dagger$ 
}

\begin{document}

\maketitle

\begin{abstract}
Place-cell networks, typically forced to pairwise synaptic interactions, are widely studied as models of cognitive maps: such models, however, share a severely limited storage capacity, scaling linearly with network size  and with a very small critical storage. This limitation is a challenge for navigation in three-dimensional space because, oversimplifying, if encoding motion along a one-dimensional trajectory embedded in two dimensions requires $O(K)$ patterns (interpreted as bins), extending this to a two-dimensional manifold embedded in a three dimensional space -yet preserving the same resolution- requires roughly $O(K^2)$ patterns, namely a supra-linear amount of patterns. In these regards, dense Hebbian architectures, where higher-order neural assemblies mediate memory retrieval, display much larger capacities and are increasingly recognized as biologically plausible, but have never linked to place cells so far. 
\newline
Here we propose a minimal two-layer model, with place cells building a layer and leaving the other layer populated by neural units that account for the internal representations (so to qualitatively resemble grid cells in the medial enthorinal cortex of mammals): crucially, by assuming that each place cell interacts with pairs of grid cells (the minimal  quest to capture information on position but also on direction, i.e. the one- and two-point correlation functions),  we show how such a model is formally equivalent to a dense Battaglia-Treves-like Hebbian network of grid cells only endowed with four-body interactions. By studying its emergent computational properties by means of statistical mechanics of disordered systems, we prove -analytically- that such effective higher-order assemblies (constructed under the guise of biological plausibility)  can support supra-linear storage of continuous attractors; furthermore, we prove -numerically- that the present neural network (namely the simplest dense generalization of the interplay between grid and cells) is, thus, already capable of recognition and navigation on general surfaces embedded in a three-dimensional space. 
\end{abstract}

\section{Introduction}\label{sec:intro}
The hippocampus, particularly the CA1 region, hosts place cells that fire selectively when an animal occupies specific locations, thereby forming the building blocks of cognitive maps—internal representations of the external, physical, space \cite{OKeefe1971,Rich2014,Fenton2008,Leutgeb,PlaceCellOriginal}.
These networks have been extensively modeled as continuous attractor neural networks (CANNs), which support localized bumps of activity that smoothly track stimuli along continuous manifolds (see e.g. \cite{Ale0,MonassonTreves2,MonassonTreves19}). Among such models, the Battaglia–Treves formulation, with $N$ McCulloch–Pitts neurons storing $K$ spatial maps, has served as a canonical reference \cite{BattagliaTreves1998,Treves-Noi}. Its statistical-mechanical analysis has yielded exact phase diagrams and clarified the roles of noise and inhibition, yet revealed a severe limitation: storage scales only linearly with system size, i.e., $K_{\max} = \alpha_c N$, with $\alpha_c \lesssim  10^{-2}$, far below the Hopfield benchmark ($\alpha_c \sim 10^{-1}$). Even improved models deepened in more recent times, see e.g. \cite{MonassonPlaceCellsLong,MonassonPlaceCellsPRL}, still face roughly the same low capacity. 
\newline
This shortcoming is indeed not unique to the Battaglia–Treves model but, rather, stems from the common assumption of pairwise synaptic couplings ($p=2$), inherited from classical Hebbian learning and shared by most attractor frameworks for spatial memory. In contrast, recent work on dense Hopfield models \cite{KrotovNew1} has demonstrated that many-body generalizations retain biological plausibility while achieving supra-linear storage \cite{LindaSuper,LindaRSB,DenseCapacity,DanielinoDenso,Krotov1,Krotov2,theriault2025saddle}. Extending this perspective to  networks that try to capture spatial correlations  can thus be relevant, especially if we think that whereas encoding locomotion along a one-dimensional manifold embedded in $d=2$ dimensions requires $O(K)$ patterns, representing motion on a two-dimensional manifold embedded in $d=3$ dimensions demands roughly $O(K^2)$ patterns (if we want to preserve spatial resolution), resulting in an unattainable scenario if tackled by neural networks supporting solely linear capacity storage.

On top of that, the discovery of grid cells in the medial entohrinal cortex (MEC) of mammals \cite{hafting2005microstructure} has led to the idea that the spatial selectively shown by place cells is not entirely encoded in the hippocampus (where the CA1 and CA3 regions populated by place cells lie), but it is rather a byproduct of internal activity in the MEC and its connection with the hippocampus \cite{place-to-grid1, place-to-grid2, place-to-grid3}. In fact, as the animal crosses specific positions in space, grid cells activate coherently in periodic hexagonal-grid patterns (hence showing spatial selectivity) and their activity is fed to the hippocampus producing the aperiodic spatial selectivity shown by place cells \cite{place-to-grid3}. 
Grid cells are known for maintaining their characteristics (\emph{i.e.} scale, phase and orientation) across different environments \cite{Treves-Moser}, which suggests that grid cells work as universal maps, which is compatible with the idea that grid-cells offer a universal metric for space-representation and space-navigation. The latter constitutes a key difference with place cells, whose configurations change at different environments, a property that is called remapping \cite{Muller-Kubie}, which is essential for recalling past memories associated with different space environments, \cite{Moser-yasser}.
\newline
\newline
In order to investigate the interplay between place and grid cells in mammals' navigation system, we propose here a suitably simplified  associative memory model that tries to capture some of the main properties of the biological counterpart, while attaining the possibility of studying its computational properties with techniques inherited from statistical mechanics of spin glasses, namely interpolation technique and replica trick. Nevertheless, despite the drastic simplifications carried out in keeping its architecture minimal, which is needed to perform exact computations (at the replica symmetry level of description), the model is able to work as a navigation system on rather general manifolds, capturing spatial correlations within the environment and enjoying a supra-linear storage of patterns coding for its navigation.
\newline
Building on analogies with $p$-spin models in spin-glass theory \cite{Gardner,Gardner2,Baldi,Burioni} and relying upon the duality between (higher-order) Boltzmann machines and (generalized) Hopfield neural networks \cite{FraDenso,LindaSuper,LindaUnsup,DaniPRE2018,Mezard2017,Monasson2017,manzan2025effect,theriault2025modelling,alemanno2023hopfield,decelle2021inverse}, we develop a minimal two-layer architecture as a core-model for spatial navigation in mammals:  the hidden (or {\em more internal}) layer is built off by neurons whose function is to represent the spatially coherent states that qualitatively resemble grid cells activity that, in turn, underlie the firing of place cells, the latter being all allocated in the visible  (or {\em more external}) layer. We stress the fact that, in our model, the purpose of the hidden layer of neurons is to produce internal representations that are localized in the space coded by a given manifold $\mathcal M_{hidden}$. The coherent activity produced in the hidden layer is responsible for the emergence of the activation of place cell neurons in the visible layer at specific locations in the visible space $\mathcal M_{visible}$. The latter is a binned representation of the environment, where each place cell is attached to a given anchor point (within its surrounding region, i.e. the {\em place field}): this way, hidden neurons work qualitatively  as grid cell units. 
\newline
Up to this point, the model is general as its actual representation depends on the particular choice of $\mathcal M_{hidden}$, which is not fixed: in our simulations and computations, however,  we minimally diverge from biological plausibility by choosing $\mathcal M_{hidden}=\mathcal S_D$ to be the $D-$dimensional (hyper-)sphere of the same dimension of $\mathcal M_{visible}$ (while in biological circuits of grid cells $\mathcal M_{hidden}$ is rather a torus \cite{di2025role})\footnote{To be sharp, in our simulations, coordination by place cells for toroidal navigation will be taken into account and solely grid cells will be a tessellation of a regular -Euclidean- space for the sake of simplicity.},  and we focus on the study of aperiodic (rather than periodic) solutions of the MC dynamics, as this considerably simplifies the calculations and numerical subtleties, yet letting the model still able to capture key aspects of the general qualitative behavior of its biological counterpart.
\newline
Crucially, if we force each place cell in the visible layer to interact with (at least) couples of grid cells -the minimal quest to capture both information on orientation  but also for navigation (namely the one- and two-point correlation functions), once the visible layer is integrated out (thus, by relying on the above mentioned duality, we focus  on the marginal distribution of solely the hidden neurons), this construction is then shown to be formally equivalent to a dense Battaglia–Treves network \cite{BattagliaTreves1998}  with many-body (i.e. four) interactions that allow to code higher-order spatial correlations needed to bin a $D\geq 3$ $\mathcal M_{visible}$ space.
In this dense formulation, the maximal storage  of $K$ patterns naturally scales as $K_{\max} = \alpha_c N^{p-1}$, where the supra-linear factor $N^{p-1}$ (rather than the small pre-factor $\alpha_c$) drives the capacity enhancement: as a result, effective fully connected higher-order neural networks, involving quadruplets (or more) of neurons but actually representing  lower-order biologically-driven layered networks, can thus constitute a natural route to overcome the storage bottleneck and, in a cascade fashion, easily allow for spatial navigation in dimensions higher than two. 
\newline
In practice, for a $d=3$ dimensional embedding space, it suffices to work with $p=4$-order interactions: we study this case in detail. Analytically, by inspecting its supra-linear storage capacity and checking the stability of coherent attractor states, numerically, facing  challenging navigation tasks, concretely showing how spatial navigation on a bi-dimensional manifold embedded in a three dimensional space becomes affordable by such a neural network. 
\newline
In doing so, we provide a theoretical framework that highlights the computational and dynamical advantages of many-body interactions in spatial memory, offering a step toward more biologically realistic models of information processing neural networks within mammals' brain.

\section{The model: from definitions to computational capabilities}\label{sec:model}

In this Section, trying to preserve the most biological plausibility,  we introduce the core-mechanisms that we identifies as mandatory for a neural network  in order to  accomplish spatial orientation and navigation on manifolds embedded in generic dimensions (i.e. not confined to planar motion).  
\newline
In particular, in Sec. \ref{SpatialRecognition} we introduce the simplest bipartite structure where one layer  -built off by place cells- interact in a mean field manner with another layer -built off by grid cells- such that, each place cell senses couples of grid cells (i.e. the interactions are ternary and not pairwise): this is the minimal quest to capture one- and two-point correlation functions among grid cells for a given chart to be recognized.
\newline
This assumption has two fundamental -despite elementary- consequences: the former is that, the dual representation of this bipartite network (achievable by marginalizing over the place cells), is a generalized dense Battaglia-Treves model equipped with four-wise interactions among grid cells only and this network is able to accomplish supra-linear storage of patterns and thus can play as a working model for spatial orientation also in the challenging case of motion in a  three-dimensional environment. 
\newline
The latter is that the field acting on each place cell contains Hebbian pairwise interactions among grid cells, hence -as grid cells correlate (due to their interactions) while they recognize the underlying chart- this forces a unique place cell to fire (or just a few of them), letting to this cell the freedom to operate in a quasi-grandmother way and this is pivotal to extend elementary the model from solely spatial recognition to account also  for spatial navigation. 
\newline
Indeed, in Sec. \ref{SpatialNavigation} we extend this core-model by providing also information on consecutive maps in order to turn the network into a true behavioral model able to cope with spatial navigation too. Crucially, as place cells can play like grandmother cells (namely they activate in a rather specific way, that is when the animal crossed their related place field), this extension can be achieved trivially, simply by adding to the Cost function defining the core-model an extra {\em navigation term} where a coupling among two consecutive place cells suffices to drive the animal within the manifold under exploration as it gives rise to a stochastic process in space à la Markov: we stress that, without a quasi-grandmother cell-like behavior of the visible layer, modeling such a spatial drive would be rather cumbersome.   
\newline
The whole result in a minimal neural network's architecture that preserves the Hebbian structure of the synaptic tensors and allows locomotion on manifolds embedded in $\mathbb{R}^3$, namely the challenging scenario (from a modeling perspective) of actual interest.    

\begin{figure}
    \centering
    \includegraphics[width=0.35\linewidth]{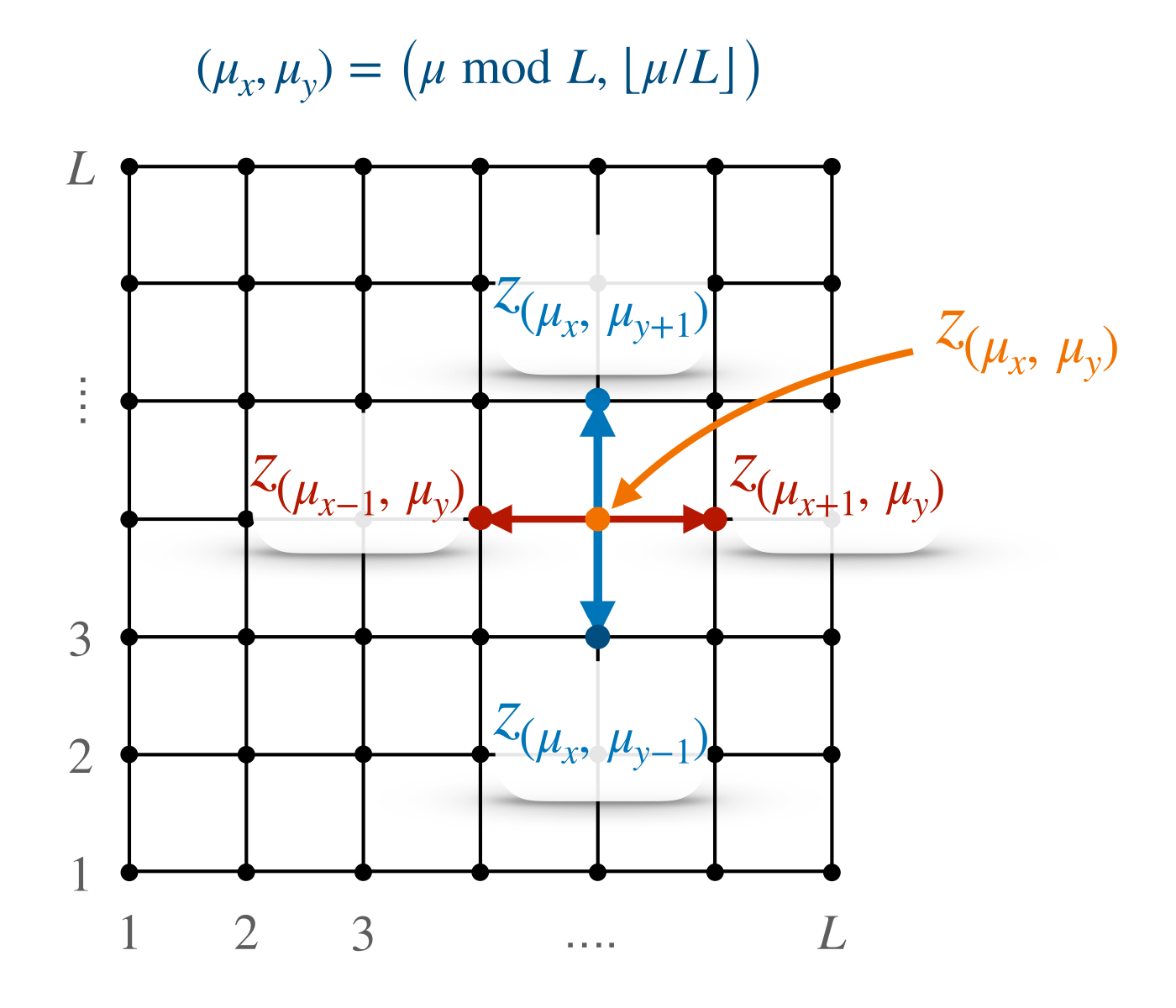}
    \includegraphics[width=0.6\linewidth]{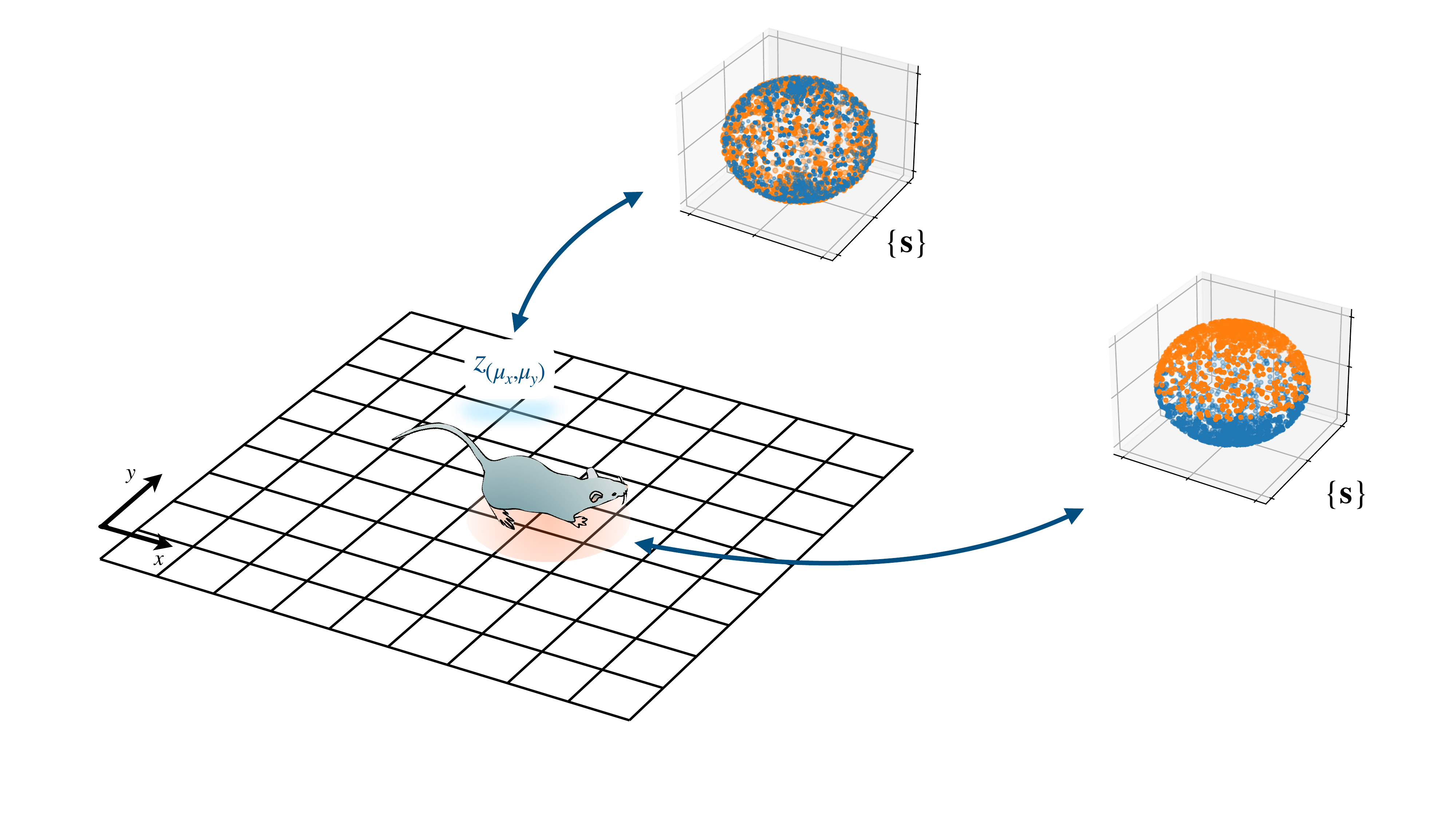}
    \caption{A sketch of the model. \emph{Left}: the place cells $\{\mathbf z\}$ are disposed on the vertices of the grid following the mapping that associates each index $\mu$ to the coordinate $(\mu_x,\mu_y)=(\mu \bmod L, \lfloor \mu/L\rfloor)$. \emph{Right}: as the animal crosses specific points in the environment, place cells activate accordingly, producing a coherent state in the space of grid cells $\{\mathbf s\}$ and relative map $\mu$.}
    \label{fig:graphic}
\end{figure}
\subsection{The simplest representation: one layer of grid cells and one layer of place cells}\label{SpatialRecognition}
Nowadays, the interplay of grid and place cells is understood to be essential for space orientation and navigation in mammals. However, grid and place cells are placed in two distinct areas of the brain, the hyppocampus and the MEC respectively, which causes some difficulties in a proper understanding of how these two neural circuits are wired together in order to produce the observed cognitive behavior related to space orientation. We propose a simplified model that combines place and grid cell-like neurons in a bipartite architecture, where a recurrent continuous attractor network of $N$ hidden neurons $\{s_i\}_{i=1,..,N}$ that play the role of grid cells, is (recurrently) connected to the visible layer built off by $K$ neurons $\{z_\mu\}_{\mu=1,..,K}$, which play the role of place cells. The bipartite architecture allows to bridge the internal space $\mathcal M_{hidden}$ to the visible space $\mathcal M_{visible}$, which represents the external environment navigated by the animal. Let us assume that $\mathcal M_{hidden}$ and $\mathcal M_{visible}$ have the same intrinsic dimension $D$, which is smaller than the dimension $d$ of the embedding space where these manifolds live, \emph{i.e.} $D=d-1$. Concretely, we shall focus on the bi-dimensional navigation embedded in our three-dimensional Euclidean space, so that $D=2$ and $d=3$, but our analytical results will be valid for any $d\geq 1$ in general\footnote{Note that $\mathcal M_{visible}$ has periodic boundary conditions, namely a toroidal topology.}.
\newline
Let us assume that each hidden neuron $s_i$ is mapped to the hidden space $\mathcal M_{hidden}$ via a multi-chart $\eta^\mu_i: \mathcal M_{hidden} \to \mathbb R^D$, one chart per each representation $\mu=1,..,K$ of the hidden space. The assumption of the existence of a multi-chart representation, rather than a single one, is essential as each such representation can be connected to a bijective map $\phi_\mu$ that bridges $\mathcal M_{hidden}$ to a point $r^\mu \in \mathcal M_{visible}$, \emph{i.e.} $\phi_\mu: \mathbb{\eta}^\mu \to r^\mu \in \mathcal M_{visible}$. In other words, we assume that the visible space $\mathcal M_{visible}$ is binned in $K$ bins, and the center $r^\mu$ of each bin is the anchor point of one visible neuron $z_\mu$, which in turn is connected to (couples of) hidden neurons $\mathbf s$ in a given fixed chart $\mathbf{\eta}^\mu$. Notice that we need a binning procedure that allows to bin this space with a number of bins $K$ of order $K\sim L^{d-1}$, where $L$ is the typical linear size of the visible space. This means that, for a generic embedding dimension $d$, we need a dense model whose order of interactions $p$ scales at least as $p=d$, hence allowing to extensively bin the external space with the size of the hidden layer: $N\sim L$. In our model, the quest that each place cells communicates with couples of grid cells automatically forces the lower (even) value of $p$  such that $p\geq d$, which for $d=3$ is $p=4$ and this suffices to allow for three-dimensional orientation and navigation as we deepen in the rest of the manuscript.\\
Let us now introduce the equations that govern the stochastic dynamics of our model:
\begin{align}
    &\tau_h \der[u_i(t)]{t} = -u_i(t) + \sum_{j=1}^N \sum_{\mu=1}^K J^\mu_{ij} \ z_\mu(t) s_j(t) + \epsilon_s(t)\\
    &\tau_v \der[z_\mu(t)]{t} = -z_\mu(t) + \frac{1}{2}\sum_{j=1}^N \sum_{i=1}^N J^\mu_{ij} \ s_i(t) s_j(t) + \epsilon_v(t)\\
    &s_i = \sigma(\gamma u_i),\label{eq:pre-post}\\
    &\avg{\epsilon(t)}=0, \:\: \avg{\epsilon(t)\epsilon(t')}=2\tau\beta^{-1}\delta(t-t')\label{eq:noisecorr}
\end{align}
where $\tau_h, \tau_v$ are the timescales of the hidden and visible layer respectively, $u_i$ is the pre-synaptic potential of the hidden neuron $i$, and it is related to the post-synaptic potential $s_i$ with the relation provided in eq. \ref{eq:pre-post}, where $\sigma(\gamma x)=\frac{1}{1+e^{-\gamma x}}$ is the sigmoid function with gain $\gamma >0$ and $\mathbf J^\mu=\{J^\mu_{ij}\}_{i,j=1,..,N}$ is the synaptic tensor that connects the hidden neurons in each chart $\mu$ to the corresponding visible place cells. Finally, the synaptic noise in each layer  $\epsilon(t)$ (with the subscripts $v$ and $s$) follows the fairly standard one- and two-points correlation relations, as coded in eq. \ref{eq:noisecorr} with fast noise (or 'temperature') $T=\beta^{-1}$ ruled by its (fastest) timescale $\tau$. 
\newline
For simplicity we assume that the visible neurons $\mathbf z$ activate via the identity input-output relation, but other choices (such as a ReLu activation) can be used. In our model, we do not study the problem of inferring the synaptic matrix $\mathbf J^\mu$ from the data, but rather assume that the maps $\bf{\eta}^\mu$ are known  (namely their entries are independently sampled accordingly to a uniform distribution as explained in Appendix \ref{AppendiceZero}) and we write directly the synaptic tensor $\mathbf J^\mu$ in the Hebbian form, which reads
\begin{align}
    J^{\mu}_{ij} = \sqrt{\frac{8}{N^{3}}} \ \eta^\mu_i \cdot \eta^\mu_j
\end{align}
where $\eta^\mu_i \cdot \eta^\mu_j$ is the usual dot product in $R^d$ and the pre-factor ensures the linear extensivity with the hidden layer's size $N$ in the thermodynamic limit: we refer to the Appendix \ref{AppendiceZero} to check the details that allow to write the standard Battaglia-Treves synaptic tensor in terms of this Hebb-like prescription.\\ 
In the zero noise limit $\beta\to\infty$ the dynamics becomes deterministic and admits the following Lyapunov function $\mathcal{H}(\boldsymbol  s, \boldsymbol z|\boldsymbol \eta)$ (that will also play  as the {\em Hamiltonian} in the analytical investigations and as the {\em Cost Function} in the numerical inspections that follow):
\begin{align}\label{eq:lyap}
    \mathcal{H}(\boldsymbol  s, \boldsymbol z|\boldsymbol \eta)=- \frac{1}{2}\sum_{\mu=1}^K \sum_{i,j=1}^N J^\mu_{ij} s_i s_j z_\mu +\frac{1}{2}\sum_{\mu=1}^K z^2_\mu + c(\boldsymbol s) 
\end{align}
where $c(s)=\sum_{i=1}^N \int^{s_i} ds_i'\ \sigma^{-1}(s_i')$ is a term arising from the input-output relation provided in eq. \eqref{eq:pre-post}. We note that  this and the other term $\propto z_{\mu}^2$ at the r.h.s. of eq. \eqref{eq:lyap} play the role the (negative log) of the prior over the hidden and visible neurons once one introduces the likelihood distribution  \cite{Coolen}, as it will become clear soon\footnote{Furthermore, in the statistical mechanical treatment that follows, these terms will be reabsorbed in the prior directly within  the partition function ({\em vide infra}).}.
\newline
It is indeed a simple exercise to show that the Lyapunov function  \eqref{eq:lyap} decreases along any dynamical trajectory $\{\boldsymbol u(t), \boldsymbol z(t)\}$, that is:
\begin{align}
    \der[\mathcal H]{t} = -\tau_v \sum_\mu \lr{\der[z_\mu(t)]{t}}^2 - \tau_h \sum_i \sigma'(\gamma u_i(t)) \lr{\der[u_i(t)]{t}}^2 \leq 0.
\end{align}
because $\sigma$ is an increasing function of its argument (such that $\sigma'(\gamma u_i)>0$) and it eventually reaches equilibrium at long times $t \to \infty$. For a given finite value of the synaptic noise $\beta<\infty$, the dynamics of the network is no longer deterministic, rather it becomes intrinsically stochastic and it can be studied by introducing the likelihood at time $t$, $p_t(\boldsymbol s, \boldsymbol z | \boldsymbol \eta)$, that -thanks to Detailed Balance granted by the symmetry of the Hebbian couplings in the Cost function-  converges for $t\to \infty$ to the following Boltzmann-Gibbs measure:
\begin{align}\label{eq:likelihood}
    \lim_{t \to \infty}p_t(\boldsymbol s, \boldsymbol z | \boldsymbol \eta) = p_{\infty}(\boldsymbol s, \boldsymbol z | \boldsymbol \eta)=Z^{-1}(\boldsymbol{\eta}) \ e^{-\beta \mathcal{H}(\boldsymbol  s, \boldsymbol z|\boldsymbol \eta)}
\end{align}
where $Z_N(\beta, \boldsymbol \xi)$, i.e., the normalization factor, is also referred to as the {\em partition function}. 
\newline
In the following, we take the infinite gain limit, \emph{i.e.} $\gamma \to \infty$, which results in boolean  variables for the hidden neurons, \emph{i.e.} $\boldsymbol s=\{0,1\}^N$ (namely, driven by simplicity, we keep the $s$ variables to be $N$ McCulloch $\&$ Pitts neurons as in the original Battaglia-Treves model), as this allows us to further simplify the sampling procedure without loosing much information. The z variables are instead real-valued neurons, equipped -as stated- with a Gaussian prior (i.e. the term  $\propto z_{\mu}^2$ in  the Cost function \eqref{eq:lyap}, whose -fairly standard- role is to prevent them to activate toward too high values). \\
In general, sampling from the likelihood \eqref{eq:likelihood} is difficult, if not intractable, since computing the partition function $Z$ is hard. In order to circumvent this difficulty, we use the pseudo-likelihood  \cite{pseudolikelihood1,pseudolikelihood2} in place of the likelihood, where we isolate the hidden neuron at site $i$, $s_i$, conditioned to all other neurons except it: $\{\boldsymbol s_{\setminus i}, \boldsymbol z\}$, and similarly for the visible layer, \emph{i.e.} we isolate $z_\mu$ conditioned to all other neurons $\{\boldsymbol s, \boldsymbol z_{\setminus \mu}\}$. Hence we define two pseudo-likelihoods, one per each layer, that read
\begin{align}
    &p(s_i|\boldsymbol s_{\setminus i}, \boldsymbol z, \boldsymbol \eta)=Z(\boldsymbol s_{\setminus i},\boldsymbol \eta)^{-1} e^{-\beta \mathcal H(s_i|\boldsymbol s_{\setminus i}, \boldsymbol z, \boldsymbol \eta)},\\
    &p(z_\mu|\boldsymbol s, \boldsymbol z_{\setminus \mu}, \boldsymbol \eta)=Z(\boldsymbol z_{\setminus \mu},\boldsymbol \eta)^{-1} e^{-\beta \mathcal H(z_\mu|\boldsymbol s, \boldsymbol z_{\setminus \mu}, \boldsymbol \eta)}.
\end{align}
The two pseudo-likelihoods defined above allow us to perform alternate Gibbs sampling as an algorithm for Monte Carlo dynamics. The advantage of this procedure is that now we are able to easily sample from the partition functions $Z(\boldsymbol s_{\setminus i},\boldsymbol \eta)$ and $Z(\boldsymbol z_{\setminus \mu},\boldsymbol \eta)$, allowing us to finally write the effective updating rules for the hidden and visible layers, which read
\begin{align}
    &P(s^{t+1}_i=1|\boldsymbol s^t_{\setminus i},\boldsymbol z^t,\boldsymbol \eta)=\sigma\lr{\beta h_i(\boldsymbol s^t_{\setminus i},\boldsymbol z^t,\boldsymbol \eta)},\\
    &P(z^{t+1}_\mu|\boldsymbol s^t,\boldsymbol z^t_{\setminus \mu},\boldsymbol \eta)=\mathcal N \lr{\zeta_\mu(\boldsymbol s^t,\boldsymbol z^t_{\setminus \mu},\boldsymbol \eta), \beta^{-1}},
\end{align}
where we introduced the cavity fields $h_i$ and $\zeta_\mu$ as follows
\begin{align}
    &h_i(\boldsymbol s_{\setminus i},\boldsymbol z,\boldsymbol \eta)=\sqrt{\frac{8}{N^{3}}} \ \sum_\mu \sum_{j\neq i} \eta^\mu_i \cdot \eta^\mu_j \ s_j z_\mu,\\
    &\zeta_\mu(\boldsymbol s,\boldsymbol z_{\setminus \mu},\boldsymbol \eta) = \zeta_\mu(\boldsymbol s,\boldsymbol \eta) = \sqrt{\frac{2}{N^{3}}} \ \sum_{i,j} \eta^\mu_i \cdot \eta^\mu_j \ s_i s_j.\label{eq:zcav}
\end{align}
Notice that the cavity field $\zeta_\mu(\boldsymbol s,\boldsymbol \eta)$ does not depend on $\boldsymbol z$ anymore, hence it plays the role of a magnetic field in the visible layer. This reveals the simplicity but also effectiveness of the present model: as a magnetic field is polarized in a given direction $\mu$ by virtue of its pairwise internal correlations among the grid neurons $\boldsymbol{s}$, the related place cell $z_\mu$ activates accordingly, producing a spike that is localized at position $r^\mu$ in the visible space $\mathcal M_{visible}$.  Furthermore, as maps are uncorrelated, in the large $N$ limit, once a place cell is firing (highlighting that the animal entered its place field), all the others stay silent (much as in the Hopfield benchmark, where once a Mattis magnetization has raised because its related pattern has been retrieved, all the other remain quiescent), thus -in the present model-  place cells spontaneously behaves as grandmother cells, acquiring the required selectivity that, empirically, typically these cells enjoy\footnote{For the sake of clearness, still bridging with the Hopfield reference, the possible presence of spurious attractor states implies that, still confined within the retrieval region whose existence we still must prove, not just a unique magnetization may rise from zero but, at worst, a few of them: this does not alter however the high specificity these cells acquire by working in the present architecture.}: we will prove, in the second part of the paper, that this grandmother-like behavior results to be pivotal in order to turn such a recognition model into a navigation model.

\bigskip

In the thermodynamic limit (and confined to the low noise and an affordable storage of charts), we expect the sampling procedure outlined above to converge towards global minima of the cost function \eqref{eq:lyap} that are continuously connected to form continuous attractors for the neural dynamics, which -in the present setting- carries the spatial information about the location of the animal in the external space. To inspect and quantify such a phenomenon, namely the ability of the network to orientate itself in the external environment and, consequently, navigate within it,  we must at first derive  its phase diagram  and then prove the existence of a not-empty retrieval region within it, \emph{i.e.} a phase where the model is able to produce spatially coherent states in the hidden manifold that are directly connected to localized activity in the visible space (the latter, in turn, are correlated to the animal's position in the physical space). 
\newline
To reach this goal, we need to introduce a set of {\em control parameters} (that, in turn, play as the axes of the phase diagram) and a set of {\em order parameters} (that are simple observables able to capture the macroscopic behavior of the network). We introduce three control parameters:
\[
\lambda, \qquad \beta = \frac{1}{T}, \qquad \alpha = \frac{p!}{2 d^{p/2}}\lim_{N \to \infty} \frac{K}{N^{p-1}},
\]
where   \(\lambda \in \mathbb{R}^+\) tunes the global inhibition strength in the network\footnote{Note that, in general, inhibition is mandatory to prevent the network from globally activating as, once we integrate out the place cells,  we are left with a dense network built off solely by Boolean variables $[0,+1]$ rather than Ising spins $[-1,+1]$.},  \(\beta \in \mathbb{R}^+\) is the so-called {\em inverse temperature}, ruling the level of fast noise in the dynamics\footnote{Note that, for \(\beta \to 0\), network dynamics is dominated by noise and resemble an unstructured random walk in configuration space. Conversely, in the zero-temperature limit \(\beta \to \infty\), the dynamics steepest descends accordingly to a  deterministic energy minimization, leading the system toward stable attractors that correspond to stored spatial maps.} and \(\alpha\) accounts for the load of patterns in  the network within the high-storage prescription, namely working at the maximal storage before blackout catastrophes may emerge. 
\newline
From now on, for the sake of simplicity, we fix one value of the control parameters, namely we work at $\lambda=1$: this simplifies considerably the calculations and, as we prove along the paper ({\em vide infra}, in particular Figure \ref{fig:P4P6P8} and its caption), if the network is able to work for unitary values of the inhibition strength, it can certainly work (even better) for other (close by) values\footnote{Indeed in the insets of   Figure \ref{fig:P4P6P8} we show $\alpha_c$ vs $\lambda$ where it shines that the case $\lambda=1$ plays as an effective lower bound for the critical storage (namely, slightly higher values of $\lambda$ improve the network performances).}.
\newline
Once the control parameters have been introduced, the macroscopic behavior of the system is naturally described by the following order parameters
\begin{align}
    \bb x_\mu &= \frac{1}{N} \sum_{i=1}^N \bb \eta^\mu_i s_i \quad \text{(population vector)} \label{x_mu} \\
    q_{ab} &= \frac{1}{N} \sum_{i=1}^N s_i^a s_i^b \quad \text{(grid cell's replica overlap)} \label{qab} \\
    p_{ab} &= \frac{1}{K} \sum_{\mu=1}^K z_{\mu}^a z_{\mu}^b \quad \text{(place cell's replica overlap)} \label{pab} \\
    m &= \frac{1}{N} \sum_{i=1}^N s_i \quad \text{(mean firing activity)} \label{m}
\end{align}
where \(a, b = 1, \ldots, n\) denote replica indices.
\newline
Next, as standard in the statistical mechanics of disordered systems, we introduce and study the (quenched) free energy of the model $\mathcal{A}(\alpha,\beta,\lambda)$, namely
\begin{equation}\label{EnergiaLiberaDefinizione}
    \mathcal{A}(\alpha,\beta,\lambda) := \lim_{N \to \infty} \frac{1}{N} \mathbb{E}\ln Z_{N,K}(\beta, \lambda, \bb \eta),    
\end{equation}
where the expectation $\mathbb{E}$ averages over the randomness in the quenched charts: as standard in the theoretical investigations, these are entirely random objects, namely their entries are Rademacher variables. 
\newline
Once reached an expression for the quenched free energy in terms of the control and order parameters of the theory, its extremization w.r.t. the order parameters returns to a set of self-consistency equations that trace their evolution in the space of the control parameters, whose inspection allows to paint the phase diagram of the model, namely to obtain the explicit evolution of the order parameters in the space of the control parameters\footnote{For the sake of clearness,  to be sharp, due to historical reasons we are using the statistical pressure $\mathcal{A}$ and not the free energy $F$ with no loss of generality as $\mathcal{A} = -\beta F$.}.
\newline
Given that the dynamics of the visible neurons is driven by the internal correlations among the hidden neurons $\boldsymbol{s}$, we can safely integrate out the formers over the factorized Gaussian measure $D\boldsymbol z=\prod_\mu \frac{dz_\mu}{\sqrt{2\pi \beta^{-1}}} \ \exp\lr{-\frac{\beta}{2}z_\mu^2}$, and write the partition function as follows
\begin{align}\label{DualityEq}
Z_{N,K}(\beta, \lambda=1, \bb \eta) &= \sum_{\bb s=\{0,1\}^N }\int D\boldsymbol z\: \exp\lr{ \beta\sqrt{\frac{2}{N^{3}}}\sum_{\mu=1}^K \sum_{i,j=1}^N \ \eta^\mu_i \cdot \eta^\mu_j s_i s_j z_\mu }\\\label{eq:gridZ}
&=\sum_{\bb s=\{0,1\}^N } \exp\left( \frac{\beta}{N^3}\sum_{i_1,i_2,i_3,i_4=1}^N\sum_{\mu=1}^K \eta^\mu_{i_1} \cdot \eta^\mu_{i_2}\ \eta^\mu_{i_3} \cdot \eta^\mu_{i_4}\ s_{i_1}s_{i_2} s_{i_3}s_{i_4}\right)
\end{align}
\begin{figure}[!h]
    \centering
    \includegraphics[width=0.5\linewidth]{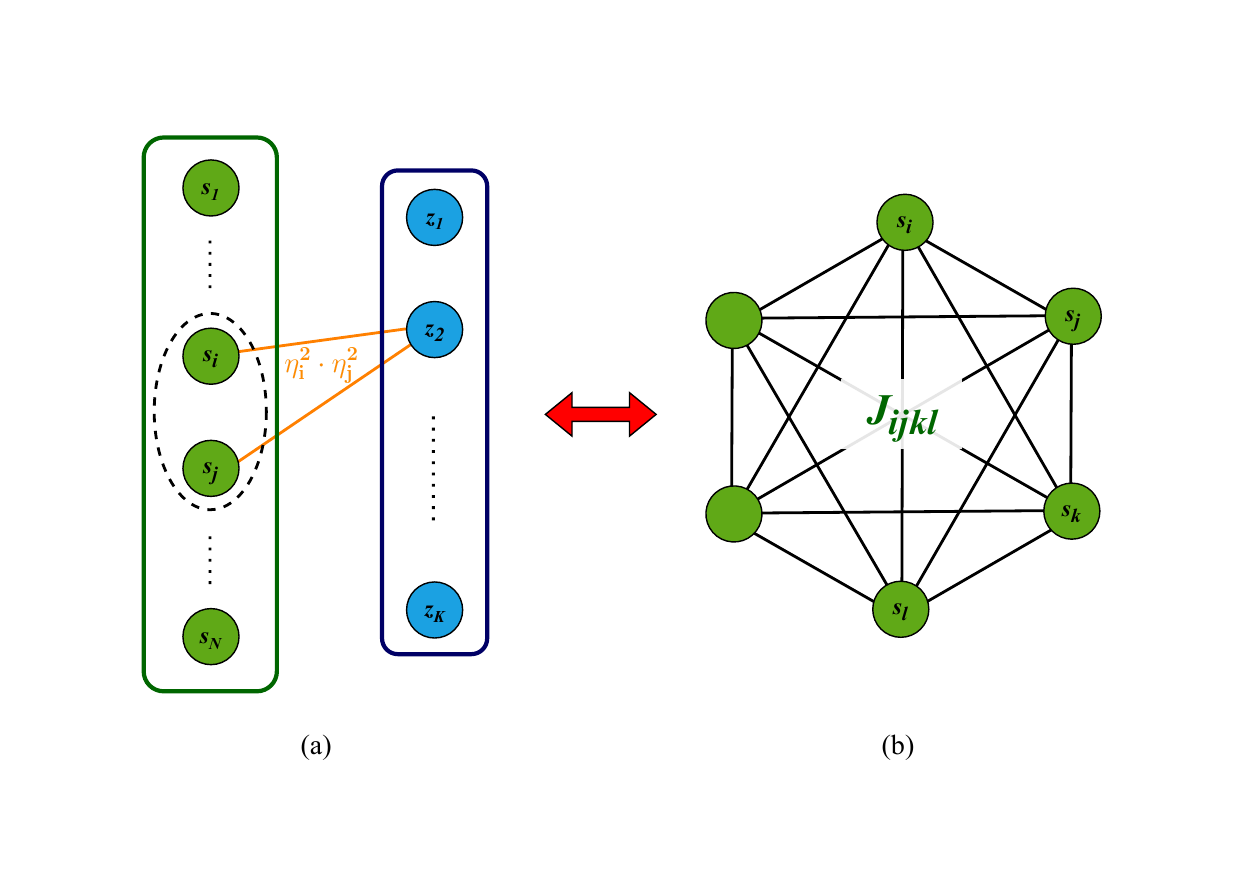}
    \caption{Duality of representation. \emph{Left}: the two-layer neural network where couples of grid cells (green circles in the left layer) are coupled to a place cell (blue circle in the right layer). For the sake of simplicity only one triplet -\emph{i.e.} one coupling- is shown. \emph{Right}: the equivalent representation (obtained by marginalizing out the place cells, see eq. \eqref{DualityEq}) in terms of a dense Battaglia-Treves-like neural network of grid cells only.}
    \label{fig:dualita}
\end{figure}
Namely we reached the equivalent dual representation of this model naturally in terms of a dense Hebbian network built off by solely grid cells: see Figure \ref{fig:dualita}.
While we refer to the supplementary material for the mathematical details that allow to express  the free energy of this class of models in terms of control and order parameters  (achieved by two independent approaches, namely interpolation technique -see Appendix \ref{Appendic:interpolazione}- and replica trick -see Appendix \ref{sec:replica}), hereafter we report directly the results that stem from this investigation and its  extremization, namely the explicit expression of the free energy as well as the  self-consistent equations for its order parameters.
\newline
Under the replica symmetry ansatz (namely assuming that these stochastic variables do not fluctuate in the thermodynamic limit, rather they concentrate around their unique averages $\overline{m}$, $\overline{x}$, $\overline{q}$), the free energy of the dense Battaglia-Treves model reads as
\begin{align}\label{eq:freeen}
    &\mathcal A\lr{\alpha,\beta, \lambda} = \lr{1-p} \beta \| \bb{\overline{x}} \|^p - \beta \lr{\lambda-1} \lr{1-p} \overline{m}^p + \lr{1-p} \alpha \beta^2 \lr{ \overline{q}_1^p -\overline{q}_2^p } + \nonumber\\
    & \quad +\expect_{\bb{\eta}} \int Dz \ln \bigg[1 + \exp \bigg(\beta p \| \bb{\overline{x}} \|^{p-2} \lr{\bb{\overline{x}} \cdot \bb{\eta}} - \beta p \lr{\lambda-1}\overline{m}^{p-1} + \alpha \beta^2 p \lr{ \overline{q}_1^{p-1} - \overline{q}_2^{p-1}} + \beta \sqrt{2 \alpha p \overline{q}_2^{p-1}} z \bigg) \bigg].
\end{align}
\begin{figure}[!h]
    \centering
    \raisebox{0.35em}{\includegraphics[width=0.50\linewidth]{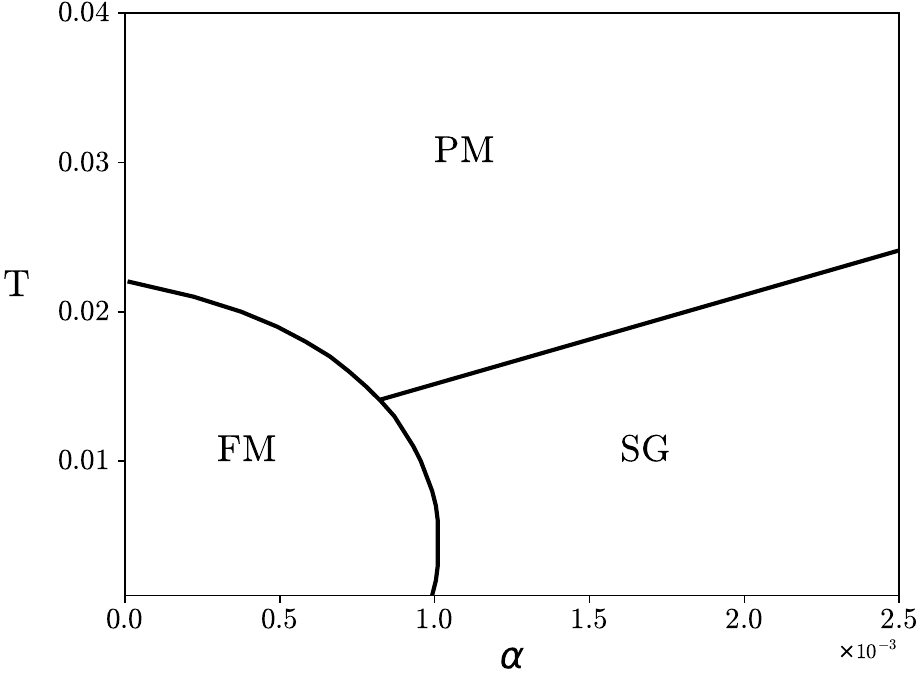}}%
    \hspace{0.5em}%
    \caption{Phase diagram of the dense Battaglia-Treves model for $p=4$. 
    Notice that the free energy eq. \ref{eq:freeen} is formally invariant under $d$ (since $d$ always only appears inside the definition of $\alpha$), hence the phase diagram is the same for every value of $d$.
    As expected, we note the presence of three regions, namely the high noise limit captured by the {\em paramagnetic phase} $(PM)$  -where nor computational capabilities neither spin glass features appear- the low noise but too much load regime captured by  the spin glass region $(SG)$ (where glassy features are shown but the model is handling too much information and it fails in performing chart recognition) and the {\em ferromagnetic phase}  $(FM)$  -where retrieval of maps is effectively achieved by the network in this challenging high storage regime where $K \propto N^{3}$.}
    \label{fig:DiagrammiFaseMain}
\end{figure}

Remarkably, as deepened in the Appendix, in reaching this expression we ensured the maximal scaling $K = \frac{2 \alpha d^{p/2}}{p!} N^{p-1}$, namely the expected supra-linear scaling $K \propto N^{p-1}$ that, in this particular setting with $p=4$, reads as $K= \alpha N^3$: if we now prove the existence of a not-empty retrieval region in the phase diagrams of this network (as it is indeed shown by the plots provided in Figure \ref{fig:DiagrammiFaseMain} for both two and three dimensions) we reached the first part of our thesis: once the analytical inspections grant the existence of such a retrieval region, we then must computationally verify that, actually, confining the network to that region, it is indeed able to reconstruct the spatial charts and, eventually, use them for navigation.  
\newline
The set of self-consistent equations that trace the evolution of the order parameters in the space of the control parameters - stemmed by the quest $\nabla_{\bb{\overline{x}},\overline{q_1},\overline{q_2}}\mathcal{A}\lr{\alpha,\beta,\lambda}=0$ - is reported hereafter and allows us to  draw the phase diagrams reported in Figure \ref{fig:DiagrammiFaseMain}.
\newline
\begin{align} \label{eq:self1}
&\| \bb{\overline{x}} \| ^2 = \int D\boldsymbol z\  \langle \ \lr{\bb{\overline{x}} \cdot \bb{\eta}} \sigma \lr{ \beta h(\bb{\overline{x}}, \overline{q}_1,\overline{q}_2; z)}\ \rangle_{\bb{\eta}},\\ \label{eq:self2}
&\overline{q}_1 = \int D\boldsymbol z\  \langle \ \sigma(\beta h(\bb{\overline{x}}, \overline{q}_1,\overline{q}_2; z)\ \rangle_{\bb{\eta}},\\ \label{eq:self3}
&\overline{q}_2 = \int D\boldsymbol z\  \langle \ \sigma^2(\beta h(\bb{\overline{x}}, \overline{q}_1,\overline{q}_2; z)\ \rangle_{\bb{\eta}},\\ \label{eq:hfield}
&h(\bb{\overline{x}}, \overline{q}_1,\overline{q}_2; z) = p \| \bb{\overline{x}} \|^{p-2} \lr{\bb{\overline{x}} \cdot \bb{\eta}} - p \lr{\lambda-1}\overline{m}^{p-1} + \alpha \beta p \lr{ \overline{q}_1^{p-1} - \overline{q}_2^{p-1}} + \sqrt{2 \alpha p \overline{q}_2^{p-1}} z.
\end{align}
where $(\bb{\overline{x}}, \overline{q}_1,\overline{q}_2)$ are the expected values of the population vector and the diagonal and off-diagonal part of the overlap respectively and $h(\bb{\overline{x}}, \overline{q}_1,\overline{q}_2; z)$ is the internal effective field.

\begin{figure}[!h]
    \centering
    \includegraphics[width=0.32\linewidth]{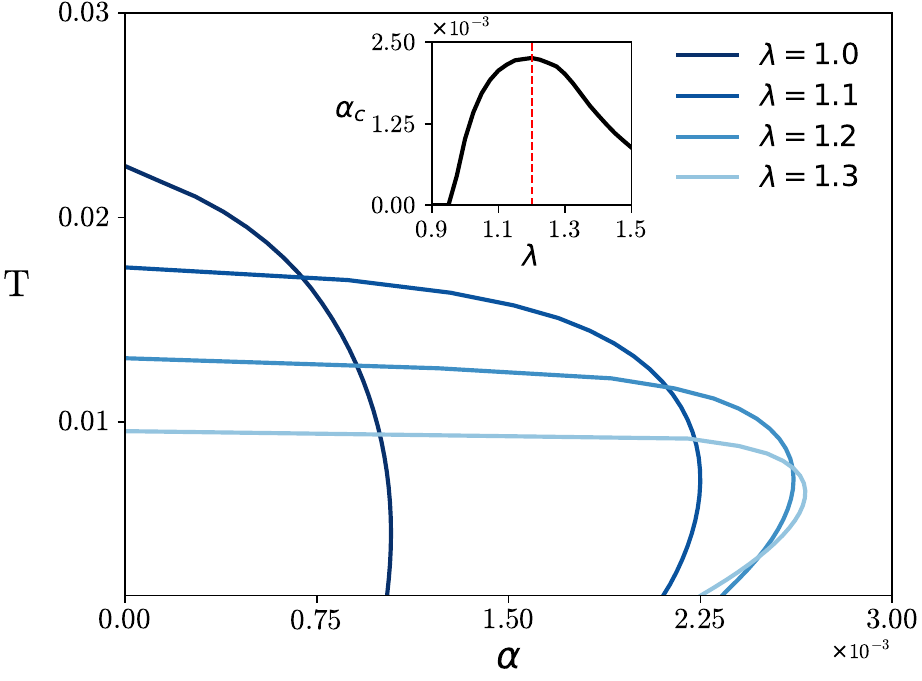}
    \includegraphics[width=0.32\linewidth]{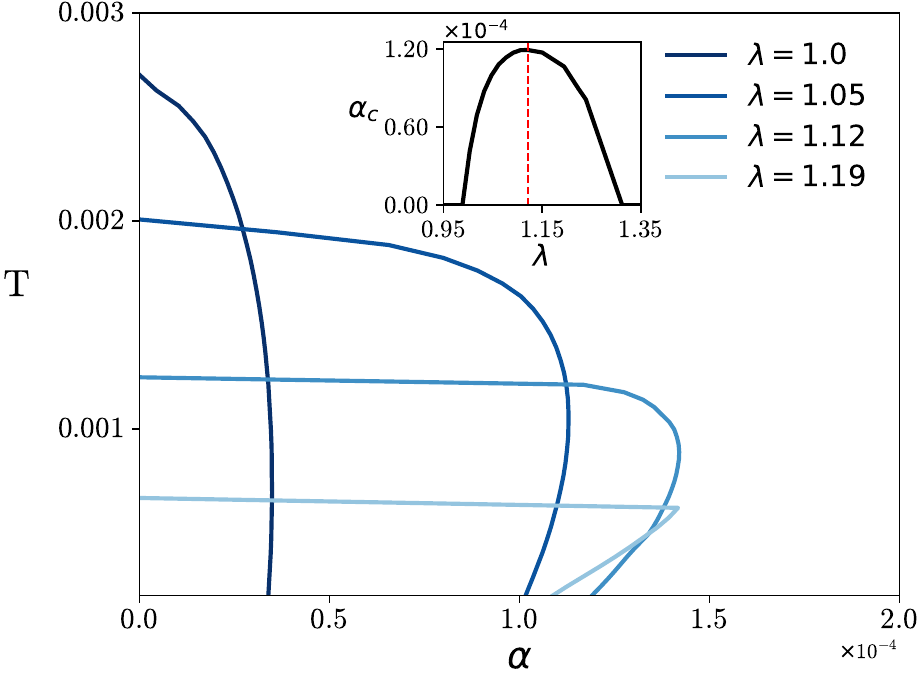}
    \includegraphics[width=0.32\linewidth]{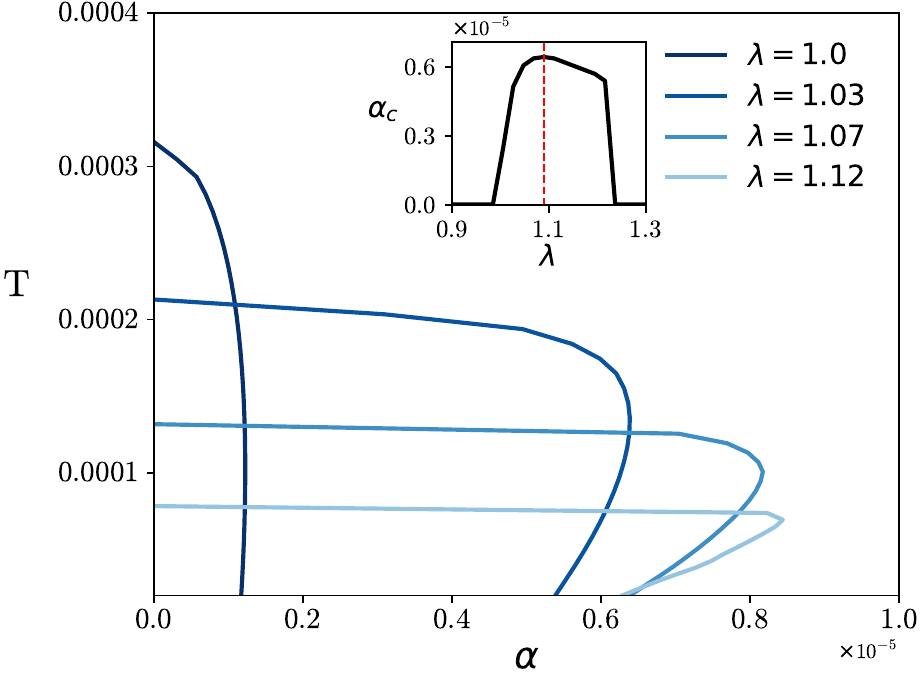}
    \caption{
    Phase diagrams of dense Battaglia-Treves like neural networks of grid cells with (left) $d=2, p=4$,  (center) $d=2, p=6$ and (right) $d=2, p=8$ at different values of the inhibition strength $\lambda$ confined to planar motion. We remark two points: The former is that, qualitatively,  the existence of a retrieval region is robust with respect to the density of the network. The latter is that, while $\lambda$ actually does not appear in the main text (due to our choice of using  $\lambda=1$ in order to simplify the mathematical treatment, as explained in Appendix \ref{AppendiceZero}), such working assumption (i.e. $\lambda=1$) is roughly a worse case scenario as proved by the insets of these  diagrams, where the maximal storage $\alpha_c$ is shown versus $\lambda$.}
    \label{fig:P4P6P8}
\end{figure}
Furthermore, in Figure \ref{fig:P4P6P8} we also provide and compare the phase diagrams that networks equipped with, respectively, $p=4$, $p=6$ and $p=8$ couplings, would give rise to: the ultimate purpose of these plots is to prove robustness of the theory also w.r.t. the network's density because, while  we analyzed the case $p=4$ as a special test case, the theory is completely general. We stress that we have conducted MC simulations in order to test numerically our theoretical results of the $p=4$ model and confirmed that the scaling $K\propto N^{3}$ is indeed correct, with a critical load $\alpha_c$ at zero temperature that is compatible with the theoretical value of $\alpha_c = 10^{-3}$ within the numerical errors, for either $d=2,3$ dimensions.

\subsection{Extending the model from  {\em spatial recognition} to  {\em spatial navigation}}\label{SpatialNavigation}
The model described so far is a reconstruction model, not yet a behavioral model, namely it is able to store and recall several attractor states corresponding to discretized positions in the visible space but there is no dynamics underlying (i.e. it accounts for chart reconstruction by a static animal). At contrary, biological neural circuits that perform spatial navigation are typically able to reconstruct the animal's position dynamically, that is, while the animal is moving within its environment. 
\newline
In other words, the animal's position changes according to the animal's motion, with mechanisms that can vary in accordance with biological complexity but that typically share the main functioning: external stimuli representing angular or linear velocity are reconstructed according to visual clues about the external environment, and are used to modulate the synaptic interactions in the neural layers where the bump of activity is correlated with the position\footnote{Note that this modulation drives the system out of equilibrium -as Detailed Balance is no longer granted- and typically breaks the symmetry of interactions, leading to the emergence of new dynamical effects that can, in some cases, break the stability of the attractors of the dynamics leading to chaotic regimes \cite{Crisanti}: we will not deepen these chaotic aspects in the present paper.}. 

In our model, we assume that the external velocity of the animal in the visible space $\mathcal M_{visible}$ is well reconstructed (due to the fact that place cells interact with couples of grid cells in the core-model defined in eq. \eqref{eq:lyap}) and available to the layer of place cells, where it is used to modulate new interactions in the $\mathbf z$ layer\footnote{Note that, tacitely, we assume that these interactions take place on a lower timescale compared to $\tau_z$, allowing the core-network presented so far (namely the cost function \eqref{eq:lyap}) to work effectively in a quasi-equilibrium regime.}. The role of these new interactions is to drive the network -and thus the moving animal- towards the next basin of attraction, which represents a new position in the visible space that is correlated with the animal's true position as the latter explores the environment\footnote{The strength of the new interaction is represented by the firing rate of a new layer of neurons, in analogy to conjunctive neurons that receive information about linear velocity and current position from the cells immediately above them in the visible layer \cite{PlaceCellOriginal, battaglia-moser}, however --for the sake of simplicity-- we omit this third layer and directly simulate the activity of such conjunctive neurons via effective fields to be added core Cost function provided by \eqref{eq:lyap}.}: see the introductory Figure \ref{eq:lyap}. 
\newline
Operatively, thanks to the grandmother like behavior of the place cells in this model, adding a new effective term  to generalize $\mathcal{H}(\bb{\sigma},\bb{z}|\bb{\eta}) \to \mathcal{H}(\bb{\sigma},\bb{z}|\bb{\eta}) + \mathcal H_{nav}$ in order to turn the model into a navigation model is a trivial task: indeed, again by a glance at Figure  \ref{eq:lyap}, it is immediate to realize that the new {\em navigation term} must read as
\begin{align}
    \mathcal H_{nav} =\  &J^+_{x} \sum_{{\mu_x,\mu_y}=1}^L z_{(\mu_x,\mu_y)} z_{(\mu_x+1,\mu_y)} +J^-_{x} \sum_{{\mu_x,\mu_y}=1}^L z_{(\mu_x,\mu_y)} z_{(\mu_x-1,\mu_y)} +\nonumber\\
    &+J^+_{y} \sum_{{\mu_x,\mu_y}=1}^L z_{(\mu_x,\mu_y)} z_{(\mu_x,\mu_y+1)} +J^-_{y} \sum_{{\mu_x,\mu_y}=1}^L z_{(\mu_x,\mu_y)} z_{(\mu_x,\mu_y-1)}.
\end{align}
Notice that the indices $(\mu_x, \mu_y)$ represent the coordinate of each neuron $z_\mu$ in the visible $D=2$-dimensional space via the following transformation:
\begin{align}
    \mu_x=\mu \text{ mod } L, \ \ \mu_y=\lfloor \mu/L\rfloor.
\end{align}
In this way we are able to cover the space $\mathcal M_{visible}=L^2$ by placing the place cell $z_\mu$ in the position given by the coordinates $(\mu_x,\mu_y)\in M_{visible}$. Notice also that the covering is periodic along both directions, hence realizing a toroidal topology, as observed experimentally in biological place cells \cite{hafting2005microstructure}\footnote{We stress once more that modeling such an extension from spatial recognition to spatial  navigation  would be, in principle, rather complicated, while here -ultimately due to the highly selective firing of place cells that allows them to behave in a quasi-grandmother manner-  it can be taken into account by simply coupling in a pairwise manner two  {\em consecutive} place cells. Furthermore, we also stress that -despite the high selectivity of the place cells is empirical well established (and related to the amplitude of their place fields), here we did not assume their behavior, rather we obtained it as an emergent property of the collective action of all the cells.}.
\newline
The quantities $(J^\pm _{x}, J^\pm_{y})$ are functions of time accounting for the firing activity of the conjunctive neurons \cite{battaglia-moser}, which are modulated by the external velocity $v_{ext}$ of the animal. In particular, suppose that a conjunctive neuron responsible for a shift in the $x$ direction, fires $f_x\tau$ times (where $\tau$ is the conjunctive neurons timescale): each time it fires, it drives the activity of $z_\mu$ towards the right (or left) direction by one unit of distance. Hence, if the $x$ component of the velocity of the animal is $v_x$, we assume the simplest proportionality rule: $f_x\tau \sim |v_x|$, where the sign of $v_x$ selects which conjunctive neuron has to operate (there are left- and right- neurons responsible for the motion along $+x$ and $-x$), and similarly for the $y$ direction. The ratio $f_x\tau/|v_x|$ gives the relative strength of the new interaction with respect to the intensity of the cavity fields defined in  \eqref{eq:zcav}. 
\newline
Concretely, the simplest modeling assumption is to chose $(J^\pm _{x}, J^\pm_{y})$ to be proportionally representing the firing activity of the conjunctive neurons, such that along the $\pm x$ directions we have
\begin{align}\nonumber
    \frac{1}{T}\int_0^T J^\pm_{\mu_x}(t)dt \sim f^\pm_x\sim |v_x|
\end{align}
where $f^\pm_x$ is the firing rate of the $\pm x$ conjunctive neurons, and similarly along $y$. 
\begin{figure}
    \centering
    \includegraphics[width=0.35\linewidth]{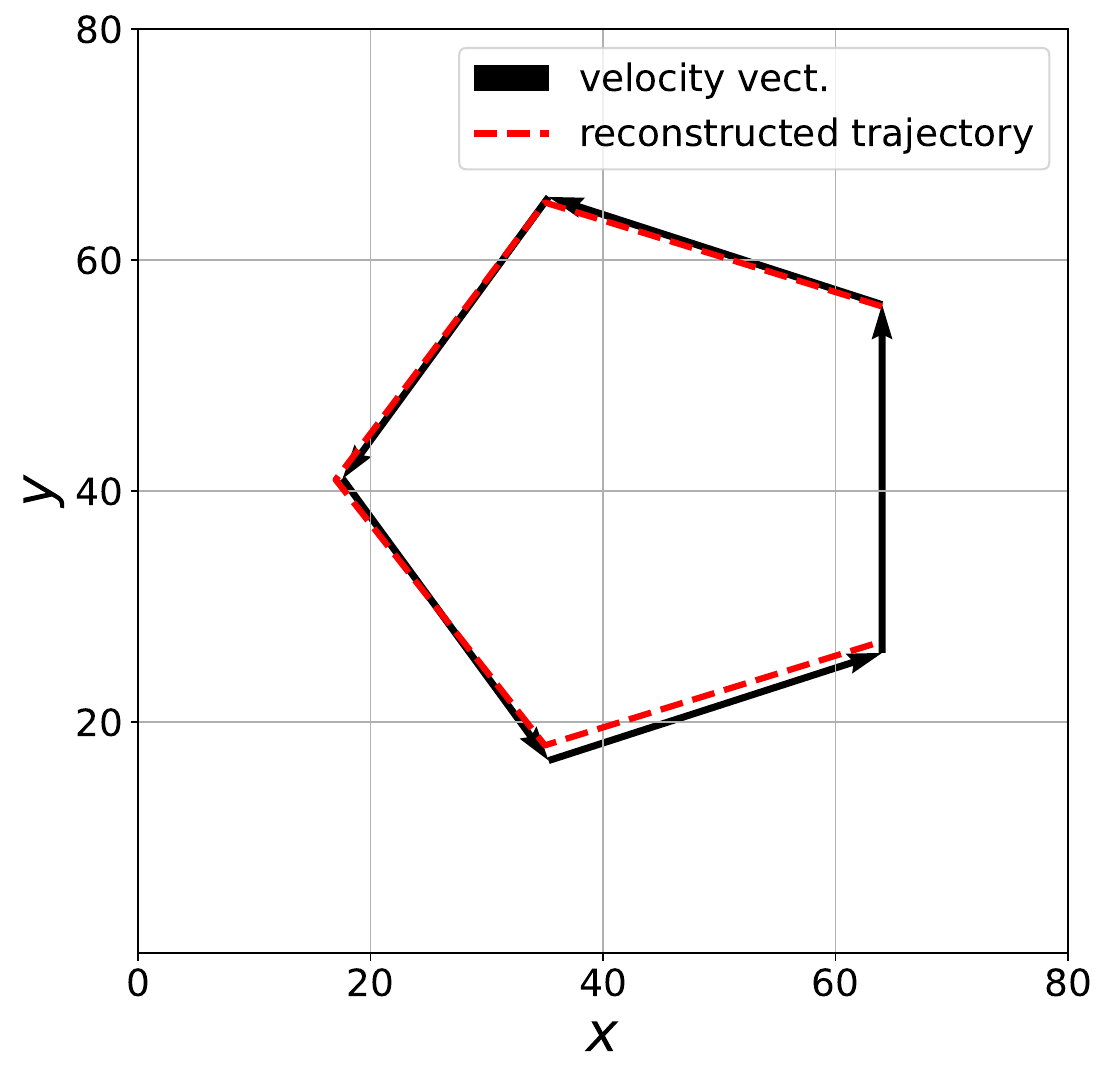}
    \includegraphics[width=0.35\linewidth]{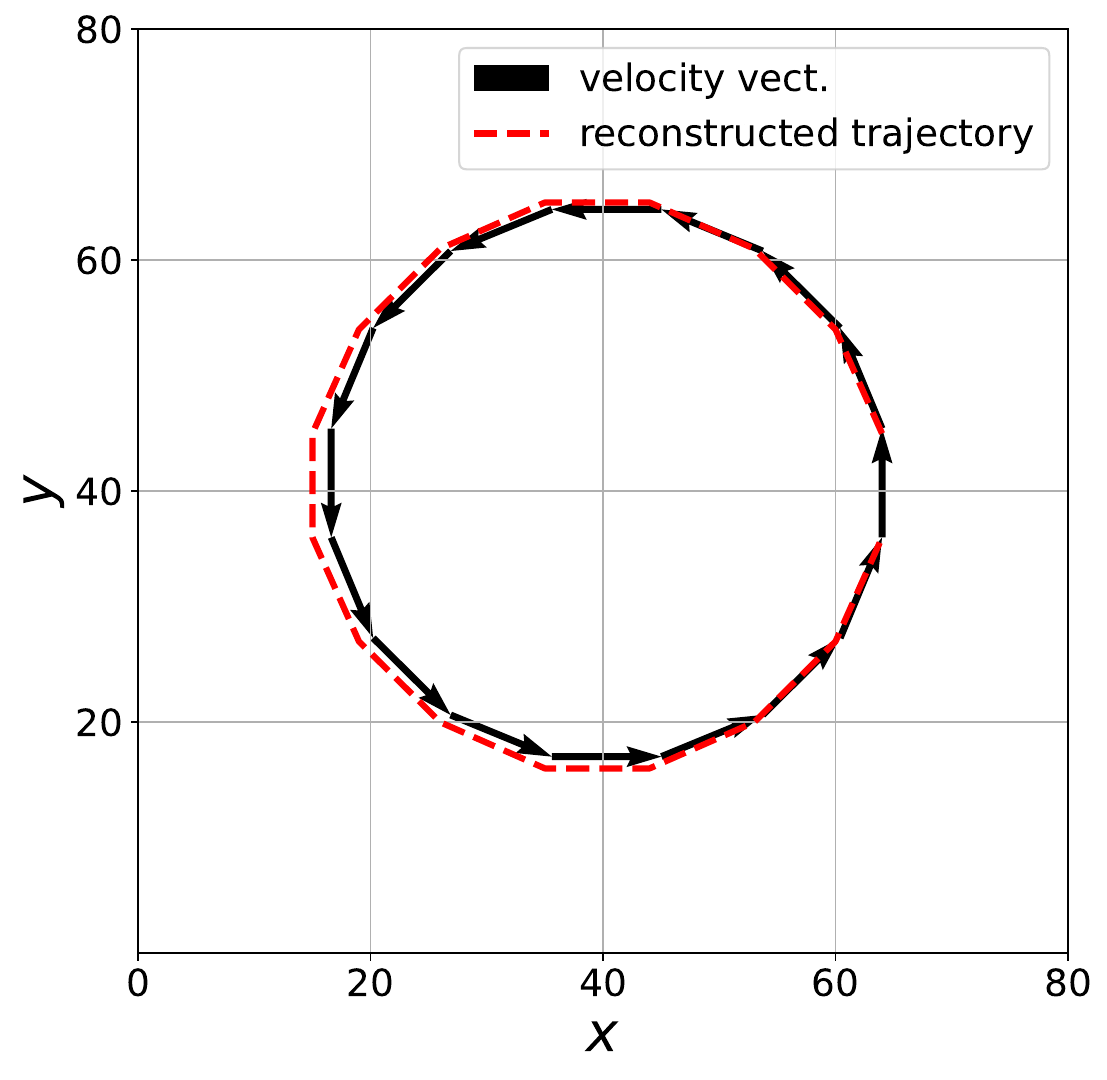}
    \caption{Simplest tests of the dynamical reconstruction of two trajectories, the first one in the shape of a pentagon (left) and the second one a circular trajectory (right) of the 'animal' by the dense network. Black: true velocity of the animal, represented by arrows starting at the prescribed  anchor points,  for all the involved tiles within the plane where the motion happens. Overall, all these arrows roughly form a circle in the visible space $\mathcal M_{visible}$. Red: reconstructed trajectory in the $\mathbf z$ layer of place cells: just by visual inspection, we can appreciate how the reconstructed trajectory resembles the real motion, nevertheless, it also shines that there are errors in the reconstruction (i.e. the two circles do not perfectly overlap as the red arrows sometimes lie inside the black circle some other times outside, preserving zero mean, but not-zero standard deviations).}
    \label{fig:circle}
\end{figure}
Such a dynamical model is able to reconstruct the trajectory of the animal rather well, despite not perfectly, as shown -as a test case- at first in a particularly simple circular motion presented in Fig. \ref{fig:circle}: a crucial point is that, unavoidably, the reconstruction error is inversely related to the resolution by which the visible space is tessellated (hence the reason for a sufficiently fine-grained grid and, thus, a dense network) and it is accumulated during the integration of the velocity along the trajectory as we now deepen focusing on more classical experiments.

Now we try and reproduce computationally, by our navigation model, the celebrated firing fields presented by the Moser's and their collaborators in their famous work \cite{Leutgeb}: we confine the numerical animal in a squared box of side $L=100$ and we force it to perform a standard random walk. In Fig. \ref{fig:randomtraj} (left panel), the activity of place cells neurons $\textbf{z}$ is shown for a random trajectory of the animal. Notice that, even if each place cell has a fixed place field of area $1$ around it by construction (see Appendix \ref{AppendiceZero}), it can actually fire even outside this region: this is due to  the error in reconstruction that gets accumulated by integrating the external velocity $v$ of the animal as times goes by,  as reported in Fig. \ref{fig:errore}. 
\begin{figure}[!h]
    \centering
    \includegraphics[width=0.45\linewidth]{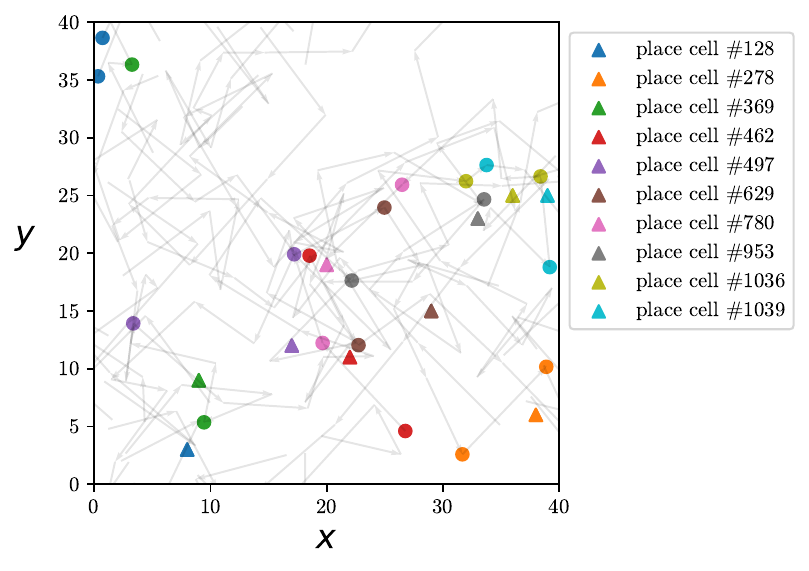}
    \includegraphics[width=0.5\linewidth]{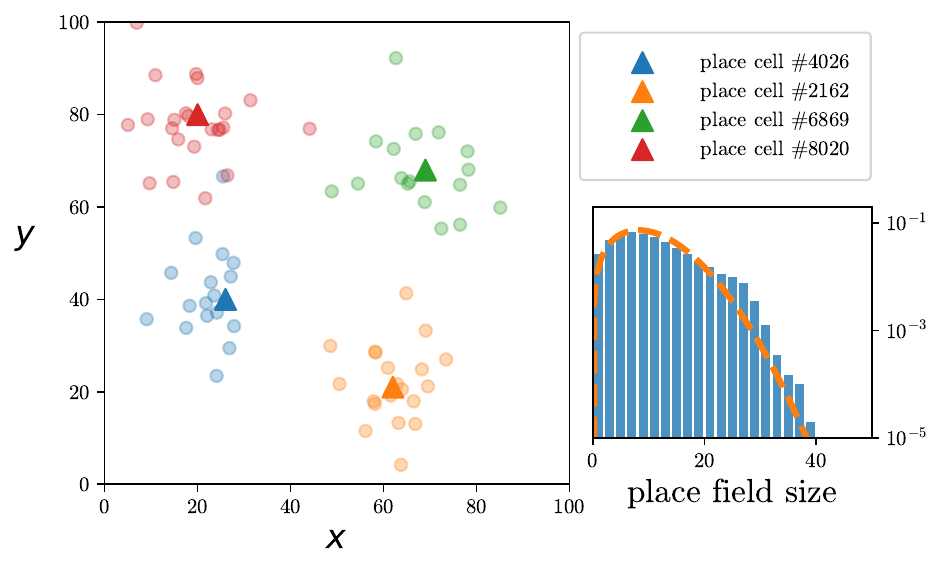}
    \caption{In this picture we reproduce theoretically, by relying upon our model with $p=4$, the same kind of plots produced empirically in the celebrated experiments reported in \cite{Leutgeb}. (Left) Simulation of a single random walk by the numerical animal is shown on a $L\times L=30\times30$ grid: while motion takes place, a subset of place cells (i.e. those indicated in the legend) activate (and are visually represented by circle points) as the animal passes closely enough to their anchor point (that are instead indicated as a triangle). (Center) After simulating several random trajectories on a larger (i.e. $L\times L=100\times 100$) grid, we have enough statistics to show the place fields of four selected place cells, which are defined as the effective area where a given cell fires. (Right) Resulting distribution of the amplitudes of the place fields: we highlight that such a histogram is fairly well compatible with the hypothesis that it follows the Gamma distribution (i.e. the $\rho(r)$ reported in eq. \eqref{eq:gamma}, shown as a dashed orange curve in the lin-log plot). 
    Note that these plots show planar navigation embedded in a three-dimensional space as, due to the periodic boundary conditions ($x+L \to x$ and $y+L \to y$), these plane by a glance are actually tori in three dimensions.}
    \label{fig:randomtraj}
\end{figure}
For each place cell, the one-dimensional errors $\delta\vec r=(\delta x,\delta y)$ are defined as the displacement between the place field centers and their firing locations: crucially, if we do a histogram of their Euclidean distances --defined as $ r=\norm{\delta\vec r}_2$, as shown in Fig. \ref{fig:randomtraj} (top right panel)-- these distribute according to a Gamma distribution
\begin{align}\label{eq:gamma}
    \rho(r)=\frac{r}{\sigma^2}e^{-\frac{r^2}{2\sigma^2}}
\end{align}
\begin{figure}[!h]
    \centering
    \includegraphics[width=0.9\linewidth]{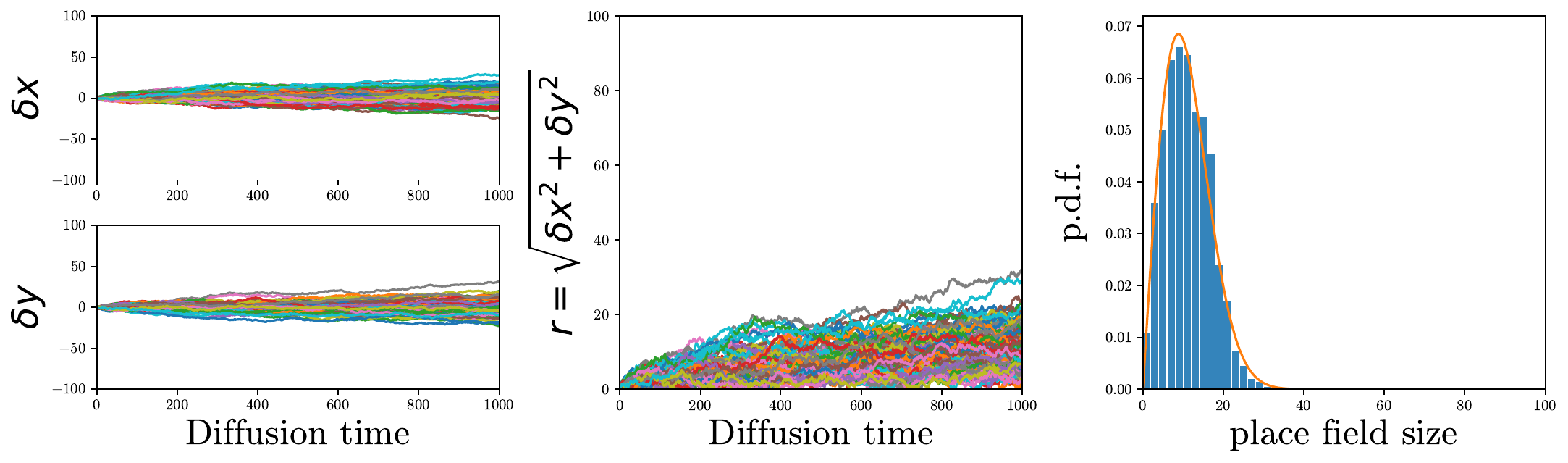}
    \caption{Errors in trajectory reconstruction by the present network as the random walk of the animal takes place. (Left)  Errors accumulate as time elapses, giving rise to a diffusion-like process.  (Middle) As a result, while the mean is kept zero (errors $\delta x,\ \delta y$ are symmetrically distributed), the distance $\sqrt{(\delta x)^2+(\delta y)^2}$ between the anchor point and the real position of the animal (when the corresponding place cell fires) gets biased (i.e., it is no longer unitary, as assumed a priori) giving rise to the Gamma distribution shown in the (Right) panel on a lin-lin scale (relative to a snapshot of the error taken at the final time of the simulated dynamics $T=1000$).}
    \label{fig:errore}
\end{figure}
where $\sigma$ is the standard deviation of $r$ that --given the diffusive nature of the process--  scales as $\sigma \sim \sqrt t$ (where $t$ is the time of the random walk): we speculate that this intrinsic error in reconstruction is the origin of the Gamma distribution (shown on a lin-log scale in Fig. \ref{fig:randomtraj}  but also on a lin-lin scale in the third panel of Fig. \ref{fig:errore}, dedicated to deepening the accumulation of noise during motion) that is empirically observed both in mice performing random walks in small cages \cite{TopiProva} as well as in  bats  performing Levy flights in long tunnels \cite{BatsProva} (see also \cite{Prova1} for a quantitative statistical analysis). 
\newline
We remark that this feature is not intrinsic to the present model, rather it happens regardless the selected tradeoff between density of the network and resolution of the tessellation as we now explain: indeed, even if we would work out a better model (eventually departing from biological evidence) and we assume that each place cells $z_{\mu}$ interact with three grid cells per time (so to eventually capture information also on the acceleration, namely collecting one- two- and three-point correlation functions), then the dual dense network would be a  generalized dense Battaglia-Treves model with couplings involving six-neurons: this would simply result in a slower time for accumulating noise during integration of the trajectory that, ultimately, still gives rise to the Gamma distribution: the impasse is intrinsic to the balance between resolution of tesselation and effective density of the network. 
\newline
As a result, we speculate that such an intrinsic tradeoff between resolution in tessellation and density of the effective dense network that must handle such a grid could be the underlying reason  of the (almost universal) Gamma distribution of place fields typically observed in the experiments\footnote{The main source of errors lies in the integration of the velocity itself given the resolution of the visible space: this gives rise to a tradeoff between storage capacity and resolution that dense network can cope with, as explained in  \cite{PRL-NOI}, see also \cite{MarulloRev2,MarulloReview}.}. We also point out that navigation flows easily if handled by the present network and this is also due the fact that dense networks have better shaped minima (if compared to shallow counterparts) that allows easily the place cells to drive the motion. At contrary, with the same resolution, a standard pairwise Battaglia-Treves model would face a storage by far above the critical value making the network stuck in spin glass minima, hence resulting in several place cells active at once but all poorly firing.

\section{Conclusions}

Driven by the observation that, keeping the resolution fixed, the larger the dimension of the environment explored by the animal, the larger the number of patterns required to tessellate it with anchor points for spatial orientation, in this work we have extended the classical Battaglia–Treves neural network model beyond the conventional pairwise ($p=2$) interaction scheme, demonstrating that a dense formulation (already with the minimal choice of $p=4$ couplings) can sustain the required supra-linear storage of spatial maps and thus allow the animal to easily explore surfaces embedded in a three-dimensional Euclidean space. 
\newline
However, rather than assuming this dense model as the starting point, we investigated a biologically-driven network with a two-layer architecture in which grid cells form the visible layer and place cells the hidden layer: crucially, in order to let the place cells detect (at least) pairwise correlations among grid-cell activities  (that is, in order to capture  one-point and two-point correlation functions, mandatory to extract information on both position as well as direction of the numerical animal exploring its surrounding), the interactions in this bipartite network are between a place cell and couples of grid cells, as coded by the cost function \eqref{eq:lyap}. By marginalizing out the place cells within a statistical mechanical treatment of this network, such a minimal network is shown to be equivalent to an effective dense Battaglia-Treves model in the grid cells only (as shown by the equivalence \eqref{DualityEq}): this duality of representation highlights how effective higher-order Hebbian assemblies can emerge from biologically plausible network structures, thus bridging the gap between theoretical dense models and hippocampal circuitry. 
\newline
Our analysis, grounded in statistical mechanics of disordered systems, reveals that the inclusion of quadruplet interactions in the Battaglia-Treves model fundamentally reshapes the storage properties of the network: unlike classical pairwise models, where the number of storable maps scales linearly with system size and  it is further severely constrained by a small critical pre-factor $\alpha_c$, the dense Battaglia–Treves model achieves supra-linear scaling, $K_{\max} = \alpha_c N^{p-1}$. Even if retaining a small $\alpha_c$, the $N^{p-1}$ factor ensures a dramatic increase in its storage capacity and this property is particularly relevant for representing navigation in higher-dimensional spaces, where -in order to preserve resolution- the number of required patterns must grow with the manifold dimensionality.
\newline
Interestingly, a not trivial result stemming from the analytical inspection of this model  at work with orientation within a given environment is that the high-selectivity of place cells (that allows them to fire solely when the animal enters their place fields) is not assumed here, rather it emerges as a consequence of the place-grid cell's interactions: as the animal crosses various place fields one after another, grid cells orchestrate time to time  so to trigger the specific response of one place cell per time  allowing these place cells to behave in a quasi grandmother way, in accordance with empirical findings. Crucially, due to this highly specialized behavior, it is thus trivial to correlate  place cells together so to turn recognition into navigation.  
\newline
Interestingly, a not trivial result stemming from the numerical inspection of this model  at work with navigation within a given environment is that, while we assumed each place field to have the same unitary amplitude,  as the motion takes place (e.g. the animal is forced to random walk in a squared cage), the place field distribution gets deformed, collapsing on a Gamma distribution that is extensively experimentally revealed in the pertinent literature (see e.g. \cite{Prova1,TopiProva,BatsProva}): this is because, despite the network is fairly able to reproduce the trajectory of the animal, yet small errors (e.g. given by the finite resolution) in its detection sum up as the motion keeps going and -while preserving zero mean (allowing for bonafide reconstruction)- they drift away the anchor point and the real position crossed by the animal when the corresponding place cell spikes. 
\newline
Remarkably, this feature is robust against model's improvements: even assuming that place cells interact with e.g. triples of grid cells (so to collect higher order information on the correlation functions), nevertheless, this would result in a more dense Battaglia-Treves network (with six interacting neurons per time), that would however preserve the same pathology, the solely difference being the slower accumulation timescale for the errors related to reconstruction.
\newline
\newline
A comment on the underlying techniques (beyond the above results of potential interest for the Neuroscience Community) that can be of interest for the Statistical Mechanical Community is that we enriched the present study with two appendices (see Appendix \ref{Appendic:interpolazione} and Appendix \ref{sec:replica}) entirely dedicated to explain how to adapt two celebrated mathematical methods in order to cope with information processing capabilities of these networks: the former is Guerra interpolation, a rigorous approach eventually more diffused within the Mathematical Physics division of our Community, the latter is the Replica Trick, that is a powerful tool largely diffused within the Theoretical Physics division of our Community. In both these approaches we assumed that, in the large network size limit, the order parameters  capturing the network's property self-average around their means  and that these are unique, namely we assumed {\em replica symmetry}, the fairly standard level of description in the bulk of neural network's Literature. Yet, some characteristics of the phase diagrams that we obtained (as e.g. the re-entrance of the retrieval region in the $\beta \to \infty$ limit at the critical storage values), suggests that replica symmetry could be broken by the true representation of the stored continuous attractors in the cost function landscape thus, in a near future, efforts will be spent to inspect the role of replica symmetry breaking in layered networks of grid and place cells.
\newline
Beyond these mathematical challenges to overcome, future inspections should also include quantitative comparisons with experimental data on hippocampal–entorhinal circuits in a systematic and exhaustive way.

\newpage

\appendix 
\setcounter{section}{-1}   
\renewcommand\thesection{\arabic{section}}

\refstepcounter{section}

\section*{Appendix Zero: Definition of maps and their statistics}\label{AppendiceZero}
Let us describe more explicitly our framework, particularly the definition of the multi-charts $\{\bb \eta^\mu_i\}^{\mu=1,..,K}_{i=1,..,N}$ and the relation between the original bipartite network and its integral representation in terms of a  dense  model, starting from the basic Battaglia-Treves reference and then generalizing to the present case.

\bigskip

Let us start by the charts: for every neuron $s_i$, these are defined as the mappings between the hidden manifold $\mathcal M_{hidden}$ and the coordinate space $\mathbb R^d$. Since we have $K$ different copies of the same hidden space $\mathcal M_{hidden}$, this mapping is responsible of the localization of the hidden neurons $\mathbf s$ in the whole space $\mathcal M_{hidden}^{\bigotimes K}$. In the following we fix $\mathcal M_{hidden}=S_D$, the hyper-sphere which we assume of unitary radius for the sake of simplicity. Hence, we can define 
\begin{definition}[Multi charts]
The multi-charts $\bb\eta^\mu_i$ are $d-$dimensional functions such that
\begin{align}
    \bb\eta^\mu_i : S_D \to \mathbb R^d
\end{align}
Recall that $d=D+1$ is the dimension of the embedding space while $D$ is the manifold where the real motion takes place. Given the manifest spherical symmetry of the system (see Fig. \ref{fig:placemaps} for $d=2$ and \ref{fig:bump}  for $d=3$), in parameterizing the multi-charts we use spherical coordinates, such that each point on the hyper-sphere is determined by the angles \(\omega = (\phi_1, \phi_2, .., \phi_D)\), with 
\begin{align}
(\phi_1,..,\phi_{D-1})\in[0,\pi]^{D-1},\quad  \phi_D\in[0,2\pi]\label{eq:dom}
\end{align}
and consequently, the multi-charts are functions of these angles \(\bb{\eta}_i^\mu = \bb{\eta}(\omega_i^\mu)\), with the condition $\eta_i^\mu\cdot\eta^\mu_i=1, \forall \mu,i$.
\end{definition}
Notice that each index $\mu=1,..,K$ can be viewed as denoting a copy of the same space $S_D$, such that each hidden neuron $s_i$ has a different position on each hyper-sphere at the same time.
\begin{remark}
Once the coordinates $\{\omega_i^\mu\}^{\mu=1,..,K}_{i=1,..,N}$ for each neuron $s_i$ and map $\mu$ are assigned, the multi-charts $\bb \eta^\mu_i=\bb{\eta}(\omega_i^\mu)$ can be interpreted as (unit) vectors in $\mathbb R^d$, where the scalar product can be defined. The scalar product of two distinct multi-charts in the same map $\mu$, $\eta^\mu_i \cdot \eta^\mu_j$ only depends on the relative angle $\phi_{ij}$ by virtue of the spherical law of cosines, namely:
\begin{align}
    \eta^\mu_i \cdot \eta^\mu_j = \cos\phi_{ij}
\end{align}
\end{remark}
\begin{figure}[!h]
    \centering
    \includegraphics[width=0.5\linewidth]{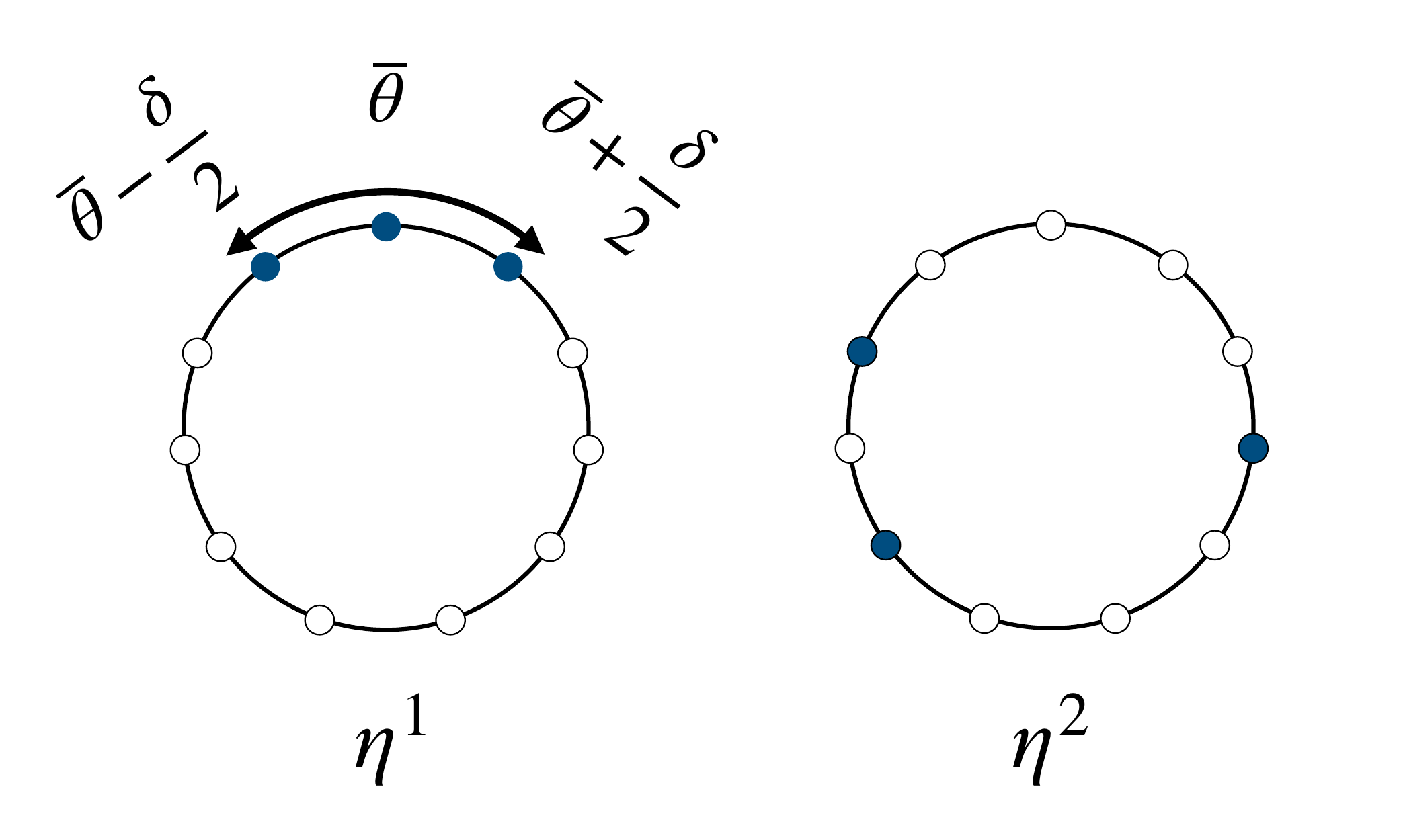} 
    \caption{Two examples of maps $\eta^1$ and $\eta^2$ shown for the case $d=2$, where the topology of the maps $\bb\eta^\mu_i$ is the circle $S_1$, characterized by the set of angles $\theta^\mu_i\in[0,2\pi]$. A retrieval state is shown in $\bb \eta^1$ as all neurons that lie within the $\overline{\psi} \in [\overline \theta - \delta/2,\overline \theta + \delta/2]$ interval are activated (here displayed as black dots), while the others stays quiescent (in white dots). The same firing pattern of neurons, that looks coherent in the first map $\bb \eta^1$, looks disordered in the other map $\bb \eta^2$. Note that the centers of the place fields are scattered roughly uniformly along the unitary circle $S_1$ and that the width of all the place fields is roughly the same (and equal to one).}
    \label{fig:placemaps}  
\end{figure}
The process by which the coordinates are assigned is very important for our goals, since it defines how the hidden neurons $s_i$ tessellate the global hidden space $(S_D)^{\bigotimes K}$. We assume that the assignment process is random, such that, independently for each map $\mu$, the coordinates of each hidden neuron $s_i$ are (also independently) extracted at random with a uniform prior over the space of angles $\{\omega_i^\mu\}^{\mu=1,..,K}_{i=1,..,N}$.

\begin{figure}
    \centering
    \includegraphics[width=0.45\linewidth]{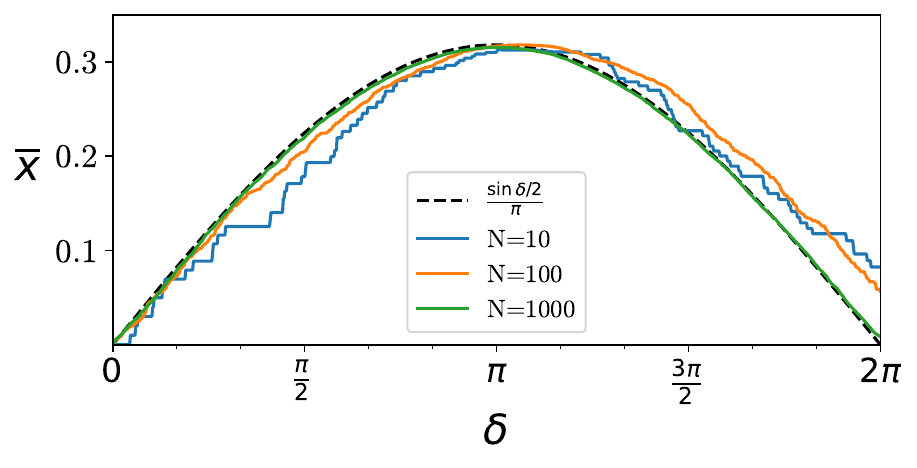}
    \includegraphics[width=0.45\linewidth]{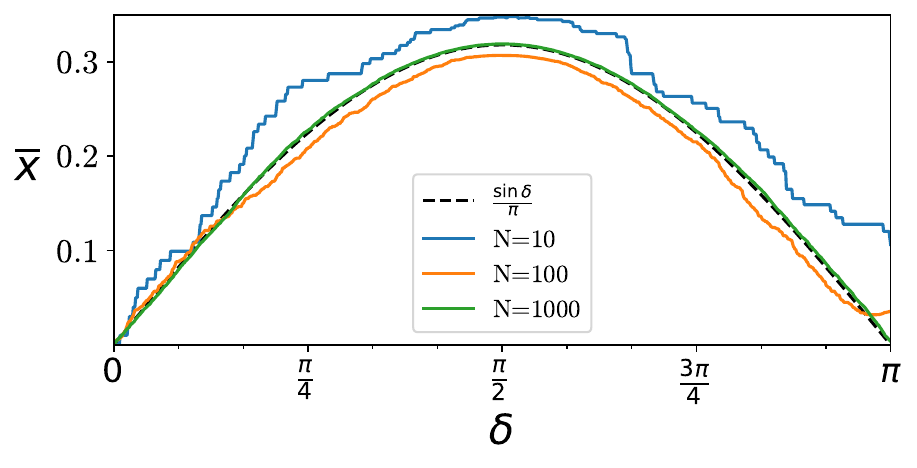}
    \caption{The norm of the population vector is shown as a function of $\delta$ in the case of $d=2$ (\emph{left}) and $d=3$ (\emph{right}), and confronted with numerical simulations at finite size $N=10,100,1000$.}
    \label{fig:xbar}
\end{figure}
\begin{definition}[Coherent States]
For any given \emph{i.i.d.} assignment of the angles $\{\omega^\mu_i\}$, a perfect coherent state in any fixed map $\nu= 1,..,K$, is a neuron state $\mathbf s$ such that the module of the order parameter $||\mathbf{\overline x}^\nu|| \equiv \overline x^\nu=\pi^{-1}$ in the thermodynamic limit $N\to \infty$.
\end{definition}
The result holds in particular in dimension $d=2$ and $d=3$. Consider first $d=2$. The population vector is $\mathbf x^\nu=\frac{1}{N}\sum_{i|s_i=1}(\cos\theta_i^\nu,\sin\theta_i^\nu)$. In the thermodynamic limit we can approximate the sum with the following integral:
$$
x^\nu=\frac{1}{2\pi} \int_{-\delta/2}^{\delta/2}d\theta^\nu (\cos\theta^\nu,\sin\theta^\nu)=\lr{0,\frac{1}{\pi}\sin\frac{\delta}{2}}
$$
where we have introduced the parameter $\delta$ which defines the width of the coherent region, \emph{i.e.} the angle window upon which the neurons are simultaneously active (while the others are quiescent), see Fig. \ref{fig:placemaps}. Hence, the norm of the population vector is $x^\nu=\frac{1}{\pi}\sin\frac{\delta}{2}$, which is equal to $\pi^{-1}$ for $\delta=\pi$.\\
For $d=3$, a similar reasoning applies, leading to
$$
x^\nu=\int_{0}^{2\pi}\frac{d\phi^\nu}{2\pi}\int_{0}^{\delta}\frac{d\theta^\nu}{\pi}(\sin\theta^\nu \cos\phi^\nu,\sin\theta^\nu \sin\phi^\nu, \cos\theta^\nu)=\lr{0,0,\frac{1}{\pi}\sin\delta}
$$
which has norm $x^\nu=\frac{1}{\pi}\sin\delta$. The results are confirmed numerically in Fig. \ref{fig:xbar}.

Now we are able to define the

\begin{definition}[Quenched Average]
For any given function \(g(\bb{\eta})\) that depends on the realization of the \(K\) maps $\{\bb \eta^\mu_i\}^{\mu=1,..,K}_{i=1,..,N}$, the quenched average is denoted as \(\expect_\eta[g(\bb{\eta})]\) or \(\avg{g(\bb{\eta})}_{\bb{\eta}}\) depending on the context, and is defined as:
\begin{equation}
\avg{g\lr{\bb{\eta}}}_{\bb{\eta}} = \int \prod_{i, \mu=1}^{N, K} d^D\omega_i^\mu \ g(\bb{\eta}(\bb \omega))=\int \prod_{i, \mu=1}^{N, K} \prod_{q=1}^Dd(\phi_q)_i^\mu \ g(\bb{\eta}(\phi_1,..,\phi_D)),
\end{equation}
where $\bb \omega$ collectively denotes the set of angles $\{\omega^\mu_i\}^{\mu=1,..,K}_{i=1,..,N}$ and the integral is supposed to be performed over the domain given by eq. \ref{eq:dom}.
\end{definition}

This assumes that the maps are statistically independent, which allows the expectation over the place fields to factorize over the sites \(i = 1, ..., N\) and the maps \(\mu = 1, ..., K\).\\
In practice, for our purposes, in the mean field approximation we need to compute the quenched average over functions that depend only on the scalar product of $\bb{\eta}$ and some vector $\bb a$, \emph{i.e.} $g(\bb \eta \cdot \bb a)$ and $\bb \eta \cdot \bb a \: g(\bb \eta \cdot \bb a)$, which is the case for our self-consistency eqs. \ref{eq:self1}-\ref{eq:self3}. The quenched averages read
\begin{align}\label{eq:general_qav}
    &\avg{g(\bb \eta \cdot \bb a)}_{\bb \eta} = \int_0^\pi \frac{d\theta}{\pi} \ g(|\bb a|\cos\theta),\\ \label{eq:x_qav}
    &\avg{\bb \eta \cdot \bb a \:g(\bb \eta \cdot \bb a)}_{\bb \eta} = |\bb a| \int_0^\pi \frac{d\theta}{\pi} \ \cos\theta\  g(|\bb a|\cos\theta)    
\end{align}
where we used $|\bb \eta|=1$. Notice that the important orthogonality condition holds
\begin{align}
    &\avg{\bb\eta^\mu_i \cdot \bb \eta^\nu_j}_{\bb\eta}=\delta_{ij}\delta^{\mu\nu}.
\end{align}

\bigskip

In order to derive the statistical properties of the model of grid cells, we introduce the dense generalization of the Battaglia-Treves Cost Function involving the neurons $\{s_i\}_{i=1,..,N}$ in $d-$dimensions and with general interaction order $K$ by analogy with the pairwise case:  following \cite{BattagliaTreves1998}, keeping in mind that the kernel has to be a function of a distance among place field cores on the manifold and that the latter is the unitary circle in two dimensions,  to take advantage from the {\em Hebbian experience} \cite{Amit1989}, the interacting strength between neurons will be written as
\begin{align}\label{eq:kernel1}
    J_{ij} =  \frac{1}{N}\sum_{\mu=1}^K \eta^\mu_i \cdot \eta^\mu_j \ .
\end{align}
\begin{figure}[!t] 
    \centering
    \includegraphics[width=0.45\linewidth]{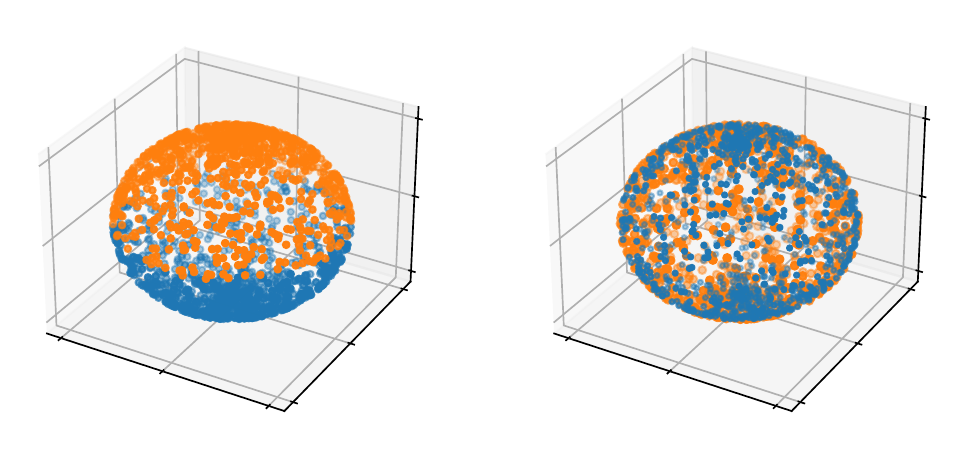}
    \includegraphics[width=0.35\linewidth]{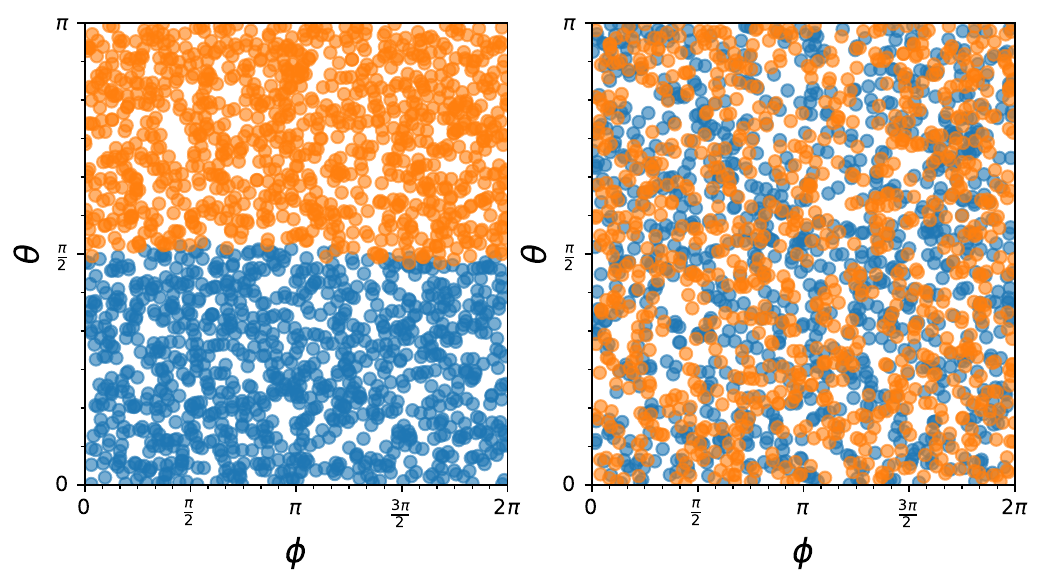}
    \caption{Neuronal activity in the first two maps $\mu=1,2$ at $\lambda=1$.
    The dense Hamiltonian  -see e.g. its sharp definition \eqref{eq:ham} from the next Appendix-  allows to define a MCMC neural update procedure that, in the retrieval regime, produces stable coherent states where half of the neurons are active and the other half quiescent. ({\em left panels}): the neural activity is spatially coherent in one of the maps ($\mu=1$, first panel on the left), while it looks random in the other maps ($\mu=2$, second panel from the left).  ({\em right panels}) The same neural activity in the $\mu=1$ and $\mu=2$ maps is shown in the third and fourth panels from the left in the spherical and angular coordinates $(\theta,\phi)$  respectively.
    }
    \label{fig:bump}  
\end{figure}
\\
Notice that the Hebbian kernel \eqref{eq:kernel1} is a function of the relative Euclidean distance of the $i,j$ neuron's coordinates $\theta_i,\theta_j$ in each map $\mu$: to show this, one can simply compute the dot product as $\eta^\mu_i\cdot\eta^\mu_j=\cos(\theta_i^\mu)\cos(\theta_j^\mu)+\sin(\theta_i^\mu)\sin(\theta_j^\mu)=\cos(\theta_i^\mu-\theta_j^\mu).$
\\
We can now define the  Cost function  of the original Battaglia-Treves model as given by the next
\begin{definition}[pairwise Battaglia-Treves Hamiltonian]\label{def:cost_function}
Given $N$ McCulloch $\&$ Pitts neurons $\bb s = (s_1, ..., s_N) \in \{ 0,1 \}^N$, $K$ charts $\bb \eta = (\eta^1,...,\eta^K)$ with  $\eta^{\mu} \in S_D$ for $\mu \in (1,...,K)$ encoded with the specific kernel \eqref{eq:kernel1}, and a free parameter $\lambda \in \mathbb{R}^+$ to tune the global inhibition within the network, the Battaglia-Treves Hamiltonian for chart reconstruction reads as\footnote{The symbol `$\approx$' in eq. \eqref{eq:H_Hop} becomes an exact equality in the thermodynamic limit, $N\to \infty$, where, splitting the summation as $\sum_{i<j} = 1/2\sum_{i,j}^N + \sum_{i=1}^N$, the last term, being sub-linear in $N$, can be neglected.}
\begin{align}\label{eq:H_Hop}
    H_N(\boldsymbol s | \boldsymbol \eta) = - \sum_{i<j}^{N,N}J_{ij} s_i s_j + \frac{\lambda-1}{N}  \sum_{i<j}^{N,N} s_i s_j \approx -\frac{1}{2N} \sum_{\mu=1}^K \sum_{i,j=1}^{N,N} (\eta^\mu_i \cdot \eta^\mu_j) s_i s_j + \frac{\lambda -1}{2N}  \sum_{i,j=1}^{N,N} s_i s_j.
\end{align}
\end{definition}
Notice that the factor $N^{-1}$ in front of the sums ensures that the Hamiltonian is extensive in the thermodynamic limit $N\to \infty$ and the factor $1/2$ is inserted in order to count only once the contribution of each couple: these pre-factors have to be suitably generalized when moving from two-body to many-body interactions ({\em vide infra}). 
Also, as we are working with McCulloch$\&$Pitts neurons (namely Boolean variables rather than Ising spins) the hyper-parameter $\lambda$ tunes a source of inhibition acting homogeneously among all pairs of neurons and prevents the network from collapsing onto a fully firing state $\bb s = (1, ..., 1)$. In fact, for $\lambda \gg 1$ the last term at the r.h.s. of eq.~ \eqref{eq:H_Hop} prevails,  global inhibition dominates over local excitation and the most energetically-favorable configuration is the fully inhibited one (where all the neurons are quiescent $\bb s=(0,..,0)$); in the opposite limit, for $\lambda \ll 1$, the most energetically-favorable configuration is  the totally excitatory one (where all the neurons are firing $\bb s = (1, ...,1)$).

\bigskip

By comparison with the above Hamiltonian, the generalization toward many-body is straightforward, as can be seen in the next definition, yet it is important to realize that -in the main text- we did not assume this dense generalization out of the blue, rather, we started from the Cost function coded in eq.\eqref{eq:lyap}, that stems from the biological evidence of existing dialogues among place and grid cells and that, in order to account for (at least) position and direction, each place cell must have access to (at least) the one- and two-point correlation functions, that is, each place cell  have to interact  with couples of grid cells (in a mean field manner in the present manuscript, for mathematical convenience): this results in a dense Battaglia-Treves model for grid cells where, remarkably, collectively the place cells dialogue with the grid cells in a grandmother setting such that they can be able to fire if and only if the animal enters the related place field.
\newline
Let us now introduce the Cost Function that we use in the present dense generalization. 

\begin{definition}[Dense Battaglia-Treves Hamiltonian]
Let \( a \in \mathbb{N} \) and \( p \in \mathbb{N} \). Consider a system of \( N \) binary neurons \( \bb{s} = (s_1, \ldots, s_N) \in \glr{0, 1}^N \), and \( K = \frac{2d^{p/2}}{p!} \alpha N^{a}\) charts \( \bb\eta = (\bb\eta^1, \ldots, \bb\eta^K) \), where \( \bb\eta^\mu = (\bb\eta^\mu_1, \ldots, \bb\eta^\mu_N) \), and each \( \bb\eta^\mu_i \in \mathbb{R}^d \) is a random unit vector independently drawn from the uniform distribution on the unit hypersphere \( \mathcal{S}^{d-1} \).

The Hamiltonian for the reconstruction of charts is expressed as\footnote{In the thermodynamic limit, the sum over ordered indices \( \sum_{i_1<\ldots<i_p} \) can be replaced by \( \frac{1}{p!} \sum_{i_1,\ldots,i_p} \). This cancels out the pre-factor \( p! \) in the original Hamiltonian.
}

\begin{equation}\label{eq:ham}
\begin{aligned}
    H_N^{\lr{p}}\lr{\bb{s}|\bb{\eta}} &= -\frac{p!}{N^{p-1}} \sum_{\mu=1}^K \sum_{i_1<\ldots<i_p=1}^{N} \lr{\bb\eta^\mu_{i_1} \cdot \bb\eta^\mu_{i_2}} \cdots \lr{\bb\eta^\mu_{i_{p-1}} \cdot \bb\eta^\mu_{i_p}} s_{i_1} \cdots s_{i_p} + \frac{p!\lr{\lambda-1}}{N^{p-1}} \sum_{i_1<\ldots<i_p=1}^{N} s_{i_1} \cdots s_{i_p} \\
    &\approx
    -\frac{1}{N^{p-1}} \sum_{\mu=1}^K \sum_{i_1, \ldots, i_p=1}^N \lr{\bb\eta^\mu_{i_1} \cdot \bb\eta^\mu_{i_2}} \cdots \lr{\bb\eta^\mu_{i_{p-1}} \cdot \bb\eta^\mu_{i_p} } s_{i_1} \cdots s_{i_p} + \frac{\lambda-1}{N^{p-1}} \sum_{i_1,\ldots,i_p=1}^{N} s_{i_1} \cdots s_{i_p}  \\
\end{aligned}
\end{equation}
\end{definition}

Now we are able to define the
\begin{definition}[Boltzmann and Quenched Averages]
Let \(f(\bb{s})\) be a function depending on the neuronal configuration \(\bb{s}\). The Boltzmann average, which represents the average over the Boltzmann-Gibbs distribution, is denoted as \(\omega(f(\bb{s}))\) and is defined as:
\begin{equation}
\omega(f(\bb{s})) = \frac{\sum_{\{\bb{s}\}} f(\bb{s}) e^{-\beta H_N(\bb{s}|\bb{\eta})}}{\sum_{\{\bb{s}\}} e^{-\beta H_N(\bb{s}|\bb{\eta})}},
\end{equation}
where \(H_N(\bb{s}|\bb{\eta})\) is the Hamiltonian of the system, and \(\beta = 1/T\) is the inverse temperature.
\newline
We use the notation \(\avg{\cdot}\) to indicate the average over both the Boltzmann-Gibbs distribution and the (quenched) realizations of the maps. This combined average is expressed as:
\begin{equation*}
\avg{\cdot} = \expect_{\bb{\eta}}[\omega(\cdot)].
\end{equation*} 
\end{definition}

\refstepcounter{section}
\section*{Appendix One: Interpolation Technique for dense networks of place cells}\label{Appendic:interpolazione}

In this appendix we adapt the celebrated Guerra's interpolation technique to the class of dense neural networks of the type coded by eq. \eqref{eq:ham}. The network is fully connected and features higher-order interactions: instead of simple pairwise couplings as in standard models with a synaptic matrix \(J_{ij}\), neurons interact in \(p\)-plets. These \(p\)-spin interactions are described by a tensorial structure involving \(p\) indices, constructed from the scalar products between the spatial positions of the neurons -- thereby encoding the geometry of the place fields -- and modulated according to the synaptic learning rules of the model. \\

The network is capable of storing \(K\) spatial maps, denoted by \(\{\bb{\eta}^\mu\}_{\mu=1}^{K}\), where each map \(\bb{\eta}^\mu\) is defined by a set of position vectors \(\bb{\eta}^\mu = \left( \bb{\eta}_{1}^\mu, \dots, \bb{\eta}_{N}^\mu \right)\), with \(\bb{\eta}_{i}^\mu \in \mathbb{R}^d\). These vectors represent the spatial coordinates of the place fields associated with the neurons, for \(i = 1,\dots, N\). \\

%

In particular, we focus on the high-storage regime, where the number of stored maps \(K\) grows extensively with the system size \(N\): to inspect analytically this regime, we adopt the one-body interpolation method adapting the original Guerra’s interpolation scheme \cite{FraDenso,Barra2017}. This technique provides a mathematically controlled framework for computing the free energy of the model and investigating the emergent  behavior of the network as a whole.

A central assumption of our approach is the \textit{replica symmetric hypothesis}, which posits that the relevant order parameters self-average and concentrate around their mean values in the thermodynamic limit. This assumption significantly simplifies the analysis and allows us to derive closed-form self-consistency equations for the evolution of the order parameters in the space of the control parameters.

Thus, these self-consistency relations are instrumental in determining the phase diagram of the model and identifying distinct operational regimes, such as the paramagnetic phase (where no memory is retrieved), the ferromagnetic phase (where a stored map is successfully retrieved), and the spin glass phase (where retrieval is hindered by a too-strong interference from multiple maps). \\

As standard in high-storage analyses, we assume that only one of the stored maps -- labeled \(\mu = 1\) -- is actively retrieved, while the remaining \(K-1\) maps act as quenched noise. This decomposition enables a clear separation between the signal and the noise components in the free energy computation and provides a tractable route to characterizing retrieval performance under heavy memory load.




Separating the signal term (\( \mu = 1 \)) from the noise contribution (\( \mu > 1 \)) in eq.~\eqref{eq:ham}, and using the definition \eqref{x_mu}, we write:

\begin{equation}
    \begin{aligned}
        H_N^{\lr{p}}\lr{\bb{s}|\bb{\eta}} =& -N\|\bb{x}_1\|^p + \frac{\lambda-1}{N^{p-1}} \sum_{i_1,\ldots,i_p=1}^{N} s_{i_1} \cdots s_{i_p} + \\
        &- \frac{1}{N^{p-1}} \sum_{\mu=2}^K \sum_{i_1,\ldots,i_p} \left( \bb\eta_{i_1}^\mu \cdot \bb\eta_{i_2}^\mu \right) \cdots \left( \bb\eta_{i_{p-1}}^\mu \cdot \bb\eta_{i_p}^\mu \right) s_{i_1} \cdots s_{i_p},
    \end{aligned}
\end{equation}

Each scalar product in the noise term is given by
$
\bb\eta^\mu_i \cdot \bb\eta^\mu_j = \sum_{t=1}^d \eta^{\mu, t}_i \eta^{\mu, t}_j,
$
where the components \( \eta^{\mu, t}_i \) are i.i.d. with zero mean and variance \( 1/d \). 

Following the reasoning of previous investigations on dense neural networks with interpolating tools (see e.g. \cite{LindaRSB,LindaSuper,FraDenso}),  to simplify the treatment of the noise term -- in particular, to allow for a Hubbard-Stratonovich (HS) transformation (that, in turn, is in order to lower the effective degree of interaction) -- we approximate the product of \( p/2 \) scalar products as the product of two independent Gaussian variables, each corresponding to a multilinear combination of \( p/2 \) vectors:
\[
\left( \bb\eta_{i_1}^\mu \cdot \bb\eta_{i_2}^\mu \right) \cdots \left( \bb\eta_{i_{p/2 -1}}^\mu \cdot \bb\eta_{i_{p/2}}^\mu \right) \approx \eta^\mu_{i_1,\ldots,i_{p/2}} \eta^\mu_{i_{p/2 +1},\ldots,i_p},
\]
where \( \eta^\mu_{i_1,\ldots,i_{p/2}} \) and \( \eta^\mu_{i_{p/2 +1},\ldots,i_p} \) are standard Gaussian variables with mean zero and variance \( 1/d^{p/4} \).

We can now rewrite the noise term approximately as 
\begin{equation*}
    \begin{aligned}
        &- \frac{1}{N^{p-1}} \sum_{\mu=2}^K \sum_{i_1,\ldots,i_p} \left( \bb\eta_{i_1}^\mu \cdot \bb\eta_{i_2}^\mu \right) \cdots \left( \bb\eta_{i_{p-1}}^\mu \cdot \bb\eta_{i_p}^\mu \right) s_{i_1} \cdots s_{i_p} \approx \\
       & \approx - \frac{A}{N^{p-1}} \sum_{\mu=2}^K \lr{ \sum_{i_1,\ldots,i_{p/2}} \eta^\mu_{i_1,...,i_{p/2}} s_{i_1} \cdots s_{i_{p/2}} }^2,
    \end{aligned}
\end{equation*}

where \( A \) is a pre-factor to be determined.

The transition from a full \( p \)-wise summation to a squared form involves a change in combinatorics. Specifically, the original sum includes all \( p! \) permutations of \( p \) indices, while the squared form symmetrically counts each unordered pair of \( p/2 \)-tuples twice. Therefore, to match the scale of the two expressions, we must correct for this over-counting by introducing a suitable normalization factor, namely $A = \sqrt{\frac{p!}{2}}$.
\newline
The factor \( p! \) accounts for all permutations of the \( p \) indices in the original term, while the factor \( \frac{1}{2} \) arises from the symmetric square, which double-counts each pair of index groups.

Note that expressing the term as a perfect square also introduces diagonal terms -- i.e., those with repeated index tuples:
\begin{equation}
\left( \sum_{i_1,...,i_{p/2}} \, \eta^\mu_{i_1,...,i_{p/2}} \, s_{i_1} \cdots s_{i_{p/2}} \right)^2
= \sum_{\substack{i_1, \dots, i_{p/2} \\ j_1, \dots, j_{p/2}}} \eta^\mu_{i_1, \dots, i_{p/2}} \eta^\mu_{j_1, \dots, j_{p/2}} s_{i_1} \cdots s_{i_{p/2}} s_{j_1} \cdots s_{j_{p/2}}.
\end{equation}

Diagonal contributions \( \lr{i_1, \dots, i_{p/2}} = \lr{j_1, \dots, j_{p/2}} \) are counted twice, whereas the original Hamiltonian counts them at most once. This overcounting introduces a systematic bias that must be corrected.

To address this, we subtract the expected value of the spurious diagonal contributions. Since each \( \eta^\mu_{i_1, \dots, i_{p/2}} \) is a zero-mean Gaussian variable with variance \( \mathbb{E}[(\eta^\mu_{i_1,\ldots,i_{p/2}})^2] = 1/d^{p/4} \), the correction term becomes:

\begin{equation}
\sqrt{\frac{p!}{2}} \sum_{i_1,...,i_{p/2}}\mathbb{E}\left[ (\eta^\mu_{i_1,...,i_{p/2}})^2 \right] s_{i_1}^2 \cdots s_{i_{p/2}}^2 = \sqrt{\frac{p!}{2}} \frac{1}{d^{p/4}} \sum_{i_1,...,i_{p/2}} s_{i_1}^2 \cdots s_{i_{p/2}}^2.
\end{equation}

\noindent Therefore, incorporating the diagonal terms and employing the definition of the order parameter~\eqref{m}, the Hamiltonian is expressed as
\begin{equation}
\begin{aligned}
H_N^{\lr{p}}\lr{\bb{s}|\bb{\eta}} =
& - N \|\bb{x}_1\|^p + N\lr{\lambda-1}m^p - \frac{\lambda-1}{N^{p-1}} \sum_{i_1,\ldots,i_{p/2}=1}^{N} s_{i_1}^2 \cdots s_{i_{p/2}}^2 + \\
&- \frac{1}{N^{p-1}} \sqrt{ \frac{p!}{2}} \sum_{\mu=2}^K \left( \sum_{i_1,\dots,i_{p/2}} \eta^\mu_{i_1,\dots,i_{p/2}} s_{i_1} \cdots s_{i_{p/2}} \right)^2 + \frac{1}{N^{p-1} d^{p/4}} \sqrt{ \frac{p!}{2}} \sum_{\mu=2}^K \sum_{i_1,\dots,i_{p/2}} s_{i_1}^2 \cdots s_{i_{p/2}}^2.
\end{aligned}
\end{equation}

We now introduce the partition function $Z_N \lr{\beta} = \sum_{\bb{s}} \exp \lr{-\beta H_N} $. By substituting \( H_N^{\lr{p}}\lr{\bb{s}|\bb{\bb{\eta}}} \) into the definition of the partition function \( Z_N \lr{\beta} \), we obtain:
\begin{equation}
    \begin{aligned}
        Z_N \lr{\beta, \bb{\eta}} = & \sum_{\bb{s}} \exp \bigg[ \beta N \|\bb{x}_1\|^p - \beta N\lr{\lambda-1}m^p + \beta \frac{\lambda-1}{N^{p-1}} \sum_{i_1,\ldots,i_{p/2}=1}^{N} s_{i_1}^2 \cdots s_{i_{p/2}}^2 + \\ 
        &+ \frac{\beta}{N^{p-1}} \sqrt{ \frac{p!}{2}} \sum_{\mu=2}^K \left( \sum_{i_1,\dots,i_{p/2}} \eta^\mu_{i_1,\dots,i_{p/2}} s_{i_1} \cdots s_{i_{p/2}} \right)^2 - \frac{\beta}{N^{p-1} d^{p/4}} \sqrt{ \frac{p!}{2}} \sum_{\mu=2}^K \sum_{i_1,\dots,i_{p/2}} s_{i_1}^2 \cdots s_{i_{p/2}}^2 \bigg].
    \end{aligned}
\end{equation}

Applying the HS transformation \footnote{$
\exp\slr{\beta Q^2} = \int \frac{dz}{\sqrt{2 \pi}} \exp \slr{-\frac{1}{2} z^2 + \sqrt{2 \beta}Qz}
$} to the quadratic term tacitly introduces the place cells as hidden variables and gives:
\begin{equation}
    \begin{aligned} \label{Z1}
        Z_N \lr{\beta, \bb{\eta}} = & \sum_{\bb{s}} \int D\boldsymbol z\: \exp \Bigg[ \beta N \|\bb{x}_1\|^p - \beta N\lr{\lambda-1}m^p + \beta \frac{\lambda-1}{N^{p-1}} \sum_{i_1,\ldots,i_{p/2}=1}^{N} s_{i_1}^2 \cdots s_{i_{p/2}}^2 + \\ 
        &+\sqrt{ \frac{2\beta}{N^{p-1}} \sqrt{\frac{p!}{2}} } \sum_{\mu=2}^K \sum_{i_1,\dots,i_{p/2}} \eta^\mu_{i_1,\dots,i_{p/2}}  s_{i_1} \cdots s_{i_{p/2}} z_\mu - \frac{\beta}{N^{p-1} d^{p/4}} \sqrt{ \frac{p!}{2}} \sum_{\mu=2}^K \sum_{i_1,\dots,i_{p/2}} s_{i_1}^2 \cdots s_{i_{p/2}}^2 \Bigg],
    \end{aligned}
\end{equation}
where $z_\mu \sim \mathcal{N}(0,1)$ and the Gaussian measure is defined as $D\boldsymbol z = \prod_{\mu=2}^K \frac{dz_\mu}{\sqrt{2 \pi}} \exp \lr{-\frac{z_\mu^2}{2}}$.

\begin{definition} 
Given the auxiliary Gaussian neurons \( z_\mu \sim \mathcal{N}(0,1) \) introduced via the HS transformation, we define the \emph{place-cell overlap} between replicas \( a \) and \( b \) as
\begin{equation}
    p_{ab} = \frac{1}{K} \sum_{\mu=1}^K z_\mu^a z_\mu^b.
\end{equation}
This order parameter measures the similarity between the hidden representations in two replicas.
\end{definition}


The form of the partition function in eq. \eqref{Z1} enables us to define a suitable interpolating Hamiltonian \( \mathcal{H}\lr{t} \), depending on a interpolation parameter \( t \in [0,1] \), which continuously connects the original model at \( t = 1 \) with a simplified one-body system at \( t = 0 \), where neurons interact only with independent Gaussian fields.

The free energy of the original model is then obtained via the fundamental theorem of calculus:
\begin{equation} \label{interp}
A(\alpha, \beta) = \mathcal{A}(1) = \mathcal{A}(0) + \int_0^1 ds\, \slr{\frac{d}{dt} \mathcal{A}(t)}_{t=s}.
\end{equation}

where the interpolating free energy \( \mathcal{A}(t) \) is defined as:

\begin{equation}
\mathcal{A}(t) = \lim_{N \to \infty} \frac{1}{N} \mathbb{E}_\eta \ln \mathcal{Z}(t),
\end{equation}

and \( \mathcal{Z}(t) \) is the interpolating partition function, defined as follows:

\begin{definition}[Interpolating partition function]
Let \( t \in [0, 1] \) be the interpolating parameter, and let $A$, $B$, $C$, $D$, $\psi_1$, $\psi_2$ in $\mathbb{R}$. Assume that \( J_i \sim \mathcal{N}(0, 1) \) for \( i = 1, \ldots, N \) and \( J_\mu \sim \mathcal{N}(0, 1) \) for \(\mu = 1, \ldots, K\), are independent and identically distributed standard Gaussian variables. The interpolating partition function \( \mathcal{Z}(t) \) is given by:

\begin{equation} \label{Z}
\begin{aligned}
    \mathcal{Z}\lr{t} = & \sum_{\bb{s}} \int D\boldsymbol z\: \exp \Bigg[ t \beta N \|\bb{x}_1\|^p + \lr{1-t}N\sum_{a=1}^d \psi_1^{\lr{a}}x_1^{\lr{a}} - t \beta N\lr{\lambda-1}m^p - \lr{1-t} N \psi_2 \lr{\lambda-1}m + \\
    &- t \beta \frac{\lambda-1}{N^{p-1}} \sum_{i_1,\ldots,i_{p/2}=1}^{N} s_{i_1}^2 \cdots s_{i_{p/2}}^2 +\sqrt{t} \sqrt{ \frac{2\beta}{N^{p-1}} \sqrt{\frac{p!}{2}} } \sum_{\mu=2}^K \sum_{i_1,\dots,i_{p/2}} \eta^\mu_{i_1,\dots,i_{p/2}} s_{i_1} \cdots s_{i_{p/2}} z_\mu + \\
    & - t \frac{\beta}{N^{p-1} d^{p/4}} \sqrt{ \frac{p!}{2}} \sum_{\mu=2}^K \sum_{i_1,\dots,i_{p/2}} s_{i_1}^2 \cdots s_{i_{p/2}}^2 + \\
    & + \sqrt{1-t} \lr{A\sum_{i=1}^N J_i s_i + B \sum_{\mu=2}^K J_\mu z_\mu} + \frac{1-t}{2}\lr{C \sum_{\mu=2}^K z_\mu^2 + D \sum_{i=1}^N s_i^2}
    \Bigg].
\end{aligned}
\end{equation}

\end{definition}
From now on, for the sake of clearness, we write explicitly the replica symmetric assumption at work on the order parameters.

\begin{definition}[Replica symmetry] \label{RS}
Under the replica-symmetry assumption, in the thermodynamic limit the order parameters self-average around their mean values (denoted with a bar), i.e., their distributions get delta-peaked, independently of the replica considered, namely
\begin{align}
    \lim_{N \to \infty} \left\langle \left( \| \bb{x}_1\| - \| \bb{\overline{x}} \| \right)^2 \right\rangle = 0 
    &\quad \Rightarrow \quad \lim_{N \to \infty} \left\langle \| \bb{x}_1 \| \right\rangle = \| \bb{\overline{x}} \| \\
    \lim_{N \to \infty} \left\langle \left( q_{11} - \overline{q}_1 \right)^2 \right\rangle = 0 
    &\quad \Rightarrow \quad \lim_{N \to \infty} \left\langle q_{11} \right\rangle = \overline{q}_1 \\
    \lim_{N \to \infty} \left\langle \left( m - \overline{m} \right)^2 \right\rangle = 0 
    &\quad \Rightarrow \quad \lim_{N \to \infty} \left\langle m \right\rangle = \overline{m} \\
    \lim_{N \to \infty} \left\langle \left( q_{12} - \overline{q}_2 \right)^2 \right\rangle = 0 \label{q_1 bar}
    &\quad \Rightarrow \quad \lim_{N \to \infty} \left\langle q_{12} \right\rangle = \overline{q}_2 \\
    \lim_{N \to \infty} \left\langle \left( p_{11} - \overline{p}_1 \right)^2 \right\rangle = 0 
    &\quad \Rightarrow \quad \lim_{N \to \infty} \left\langle p_{11} \right\rangle = \overline{p}_1 \\
    \lim_{N \to \infty} \left\langle \left( p_{12} - \overline{p}_2 \right)^2 \right\rangle = 0 
    &\quad \Rightarrow \quad \lim_{N \to \infty} \left\langle p_{12} \right\rangle = \overline{p}_2
\end{align}

Note that, for the generic order parameter \( X \), the above concentration can be rewritten as 
\[
\left\langle \left( \Delta X \right)^2 \right\rangle \xrightarrow[N \to \infty]{} 0, \quad \text{where} \quad \Delta X := X - \overline{X},
\]
and, clearly, the RS approximation also implies that, in the thermodynamic limit, 
\[
\left\langle \Delta X \Delta Y \right\rangle = 0
\quad \text{for any generic pair of order parameters } X, Y,
\quad \text{as well as} \quad \left\langle \left( \Delta X \right)^k \right\rangle \to 0 \quad \text{for } k \geq 2.
\]
\end{definition}

\begin{lemma}
    The t derivative of interpolating free energy is given by

\begin{equation} \label{der_1}
    \begin{aligned}
        \der[]{t}\mathcal{A}(t) =& \beta \avg{\|\bb{x}_1\|^p} - \sum_{a=1}^d \psi_1^{(a)} \avg{x_1^{(a)}} - \beta \lr{\lambda-1}\avg{m^p} + \psi_2 \lr{\lambda-1}\avg{m} + \\
        & + \frac{\beta K}{N^{p/2} d^{p/4}} \sqrt{\frac{p!}{2}}  \avg{p_{11}q_{11}^{p/2}} - \lr{\frac{A^2}{2} + \frac{D}{2}} \avg{q_{11}}  - \lr{\frac{B^2 K}{2N} + \frac{CK}{2N}} \avg{p_{11}} +  \\
        & - \frac{\beta K}{N^{p/2} d^{p/4}} \sqrt{\frac{p!}{2}} \avg{p_{12}q_{12}^{p/2}} + \frac{A^2}{2} \avg{q_{12}}  + \frac{B^2 K}{2N} \avg{p_{12}} +  \\
        & - \frac{\beta K}{ N^{p/2} d^{p/4}} \sqrt{\frac{p!}{2}} \avg{q_{11}^{p/2}} - \beta \frac{\lambda-1}{N^{p-1}} \avg{q_{11}^{p/2}}.
    \end{aligned}
\end{equation}

\end{lemma}

\begin{proof}
We differentiate $\mathcal{A}(t)$ with respect to $t$:
    \begin{equation}
    \begin{aligned}
    \der[]{t}\mathcal{A}\lr{t} =& 
    \frac{1}{N} \expect \frac{1}{\mathcal{Z}(t)} \sum_{\bb{s}} \int D\boldsymbol z\: \mathcal{B}(\bb{s}, \bb{z}; t) \Bigg[
    \beta N \|\bb{x}_1\|^p -N\sum_{a=1}^d \psi_1^{\lr{a}}x_1^{\lr{a}} + \\
     &- \beta N\lr{\lambda-1}m^p + N \psi_2 \lr{\lambda-1}m  - \beta \frac{\lambda-1}{N^{p-1}} \sum_{i_1,\ldots,i_{p/2}=1}^{N} s_{i_1}^2 \cdots s_{i_{p/2}}^2 + \\
    & + \frac{1}{2 \sqrt{t}} \sqrt{ \frac{2\beta}{N^{p-1}} \sqrt{\frac{p!}{2}} } \sum_{\mu=2}^K \sum_{i_1,\dots,i_{p/2}} \eta^\mu_{i_1,\dots,i_{p/2}} s_{i_1} \cdots s_{i_{p/2}} z_\mu + \\
    & - \frac{\beta}{N^{p-1} d^{p/4}} \sqrt{ \frac{p!}{2}} \sum_{\mu=2}^K \sum_{i_1,\dots,i_{p/2}} s_{i_1}^2 \cdots s_{i_{p/2}}^2 + \\
    & -\frac{1}{2\sqrt{1-t}} \lr{A\sum_{i=1}^N J_i s_i + B \sum_{\mu=2}^K J_\mu z_\mu} - \frac{1}{2}\lr{C \sum_{\mu=2}^K z_\mu^2 + D \sum_{i=1}^N s_i^2}
    \Bigg],
    \end{aligned}
\end{equation}
where \( \mathcal{B}(\bb{s}, \bb{z}; t) = \exp\left( -\beta \mathcal{H}(t) \right) \) is the Boltzmann weight associated with the interpolating Hamiltonian \( \mathcal{H}(t) \).

We now evaluate each term separately.

\begin{equation}
    (i) = \frac{1}{N} \expect \slr{\omega \lr{\beta N\|\bb{x}_1\|^p} } + \frac{1}{N} \expect \slr{\omega \lr{-N\sum_{a=1}^d \psi_1^{\lr{a}}x_1^{\lr{a}}} } = \beta \avg{\|\bb{x}_1\|^p} -\sum_{a=1}^d \psi_1^{\lr{a}} \avg{x_1^{\lr{a}}}.
\end{equation}

\begin{equation}
    (ii) = \frac{1}{N} \expect \slr{\omega \lr{- \beta N \lr{\lambda -1} m^p }} + \frac{1}{N} \expect \slr{ \omega \lr{ N \psi_2 \lr{\lambda -1}m } } = - \beta \lr{\lambda -1} \avg{m^p} + \psi_2 \lr{\lambda -1 } \avg{m} .
\end{equation}

\begin{equation}
    (iii) = \frac{1}{N} \expect \slr{ - \beta \frac{\lambda-1}{N^p} \sum_{i_1,\dots,i_{p/2}} \omega \lr{s_{i_1}^2 \cdots s_{i_{p/2}}^2 } } = - \beta \frac{\lambda -1}{N^{p/2}} \avg{q_{11}^{p/2}}.
\end{equation}

We aim to compute the contribution
\begin{equation*}
\begin{aligned}
(iv) &= \frac{1}{N} \expect \slr{\frac{1}{2 \sqrt{t}} \sqrt{ \frac{2\beta}{N^{p-1}} \sqrt{\frac{p!}{2}} } \sum_{\mu=2}^K \sum_{i_1,\dots,i_{p/2}} \eta^\mu_{i_1,\dots,i_{p/2}} \omega \lr{ s_{i_1} \cdots s_{i_{p/2}} z_\mu}} \\
& = \frac{1}{2 N \sqrt{t}} \sqrt{ \frac{2\beta}{N^{p-1}} \sqrt{\frac{p!}{2}} } \sum_{\mu=2}^K \sum_{i_1,\dots,i_{p/2}} \expect \slr{\eta^\mu_{i_1,\dots,i_{p/2}} \omega \lr{ s_{i_1} \cdots s_{i_{p/2}} z_\mu}}.
\end{aligned}
\end{equation*}

We apply Stein's lemma \footnote{%
\emph{Stein's lemma}.
Let \( X \sim \mathcal{N}(0, \sigma^2) \) and let \( f \colon \mathbb{R} \to \mathbb{R} \) be a differentiable function such that \( \mathbb{E}[|f'(X)|] < \infty \). 
Then, the following identity holds:
$
\expect[X f(X)] = \sigma^2 \expect[f'(X)].
$
} to obtain:
\begin{equation*}
    \begin{aligned}
        \expect \slr{ \eta^\mu_{i_1,\dots,i_{p/2}} \omega \lr{ s_{i_1} \cdots s_{i_{p/2}} z_\mu}}
        &= \expect \slr{ \lr{\eta^\mu_{i_1,\dots,i_{p/2}}}^2 } \expect \slr{ \pder[]{\eta^\mu_{i_1,\dots,i_{p/2}}} \omega \lr{s_{i_1} \cdots s_{i_{p/2}} z_\mu}} \\
        & = \frac{1}{d^{p/4}} \sqrt{t} \sqrt{\frac{2 \beta}{N^{p-1}} \sqrt{\frac{p!}{2}}} \slr{\omega\lr{\lr{s_{i_1} \cdots s_{i_{p/2}} z_\mu}^2} - \omega^2\lr{s_{i_1} \cdots s_{i_{p/2}} z_\mu}}.
    \end{aligned}
\end{equation*}

The expression above corresponds to a difference of overlaps under the interpolating measure. Recognizing the definitions of the replica overlaps (\ref{qab}), (\ref{pab}) we conclude:
\begin{equation}
(iii) = \frac{\beta K}{N^{p/2} d^{p/4}} \sqrt{\frac{p!}{2}} \lr{ \avg{q_{11}^{p/2} p_{11}} - \avg{q_{12}^{p/2} p_{12}}}.
\end{equation}

We now evaluate the contribution of the diagonal correction term:
\begin{equation}
\begin{aligned}
(v) &= \frac{1}{N} \mathbb{E} \slr{ - \frac{\beta}{N^{p-1} d^{p/4}} \sqrt{ \frac{p!}{2}} \sum_{\mu=2}^K \sum_{i_1,\dots,i_{p/2}} \omega \lr{ s_{i_1}^2 \cdots s_{i_{p/2}}^2}} \\
&=- \frac{\beta K}{ N^{p/2} d^{p/4}} \sqrt{\frac{p!}{2}} \avg{q_{11}^{p/2}}.
\end{aligned}
\end{equation}

We continue with the remaining terms:
\begin{equation*}
    \begin{aligned}
        (vi) &= \frac{1}{N} \expect \slr{ -\frac{A}{2 \sqrt{1-t}} \sum_{i=1}^N J_i \omega\lr{s_i} - \frac{B}{2\sqrt{1-t}} \sum_{\mu=2}^K J_\mu \omega \lr{z_\mu} } \\
        & = -\frac{A}{2 N \sqrt{1-t}} \sum_{i=1}^N \expect \slr{J_i \omega\lr{s_i}} -\frac{B}{2 N \sqrt{1-t}} \sum_{\mu=2}^K \expect \slr{J_\mu \omega\lr{z_\mu}}.
    \end{aligned}
\end{equation*}

We apply Stein's lemma to obtain:
\begin{equation*}
    \begin{aligned}
        & \expect \slr{J_i \omega\lr{s_i}} = \expect \slr{J_i^2} \expect \slr{ \pder[]{J_i} \omega \lr{s_i}} = A \sqrt{1-t} \slr{\omega \lr{s_i^2} - \omega^2\lr{s_i}}, \\
        & \expect \slr{J_\mu \omega\lr{z_\mu}} = \expect \slr{J_\mu^2} \expect \slr{ \pder[]{J_\mu} \omega \lr{z_\mu}} = A \sqrt{1-t} \slr{\omega \lr{z_\mu^2} - \omega^2\lr{z_\mu}},
    \end{aligned}
\end{equation*}
noting that $\expect \slr{J_i^2} = 1$, $\expect \slr{J_\mu^2} = 1$.
Again, this corresponds to a difference of overlaps under the interpolating measure. Recognizing the definitions of the replica overlaps, (\ref{qab}) and (\ref{pab}) we conclude:
\begin{equation}
    (vi) = -\frac{A^2}{2} \lr{\avg{q_{11}} - \avg{q_{12}}} - \frac{B^2 K}{2N} \lr{\avg{p_{11}} - \avg{p_{12}}}.
\end{equation}

Finally,
\begin{equation}
    \begin{aligned}
        (vii) &= \frac{1}{N} \expect \slr{ -\frac{C}{2} \sum_{\mu=2}^K \omega
         \lr{z_\mu^2} -\frac{D}{2} \sum_{i=1}^N \omega \lr{s_i^2} } \\
         &= - \frac{CK}{2N} \avg{p_{11}} - \frac{D}{2} \avg{q_{11}}.
    \end{aligned}
\end{equation}

\noindent
Collecting all the contributions $(i),...,(vii)$, we recover the expression for the derivative of the interpolating free energy as stated in the lemma, thus completing the proof.

\end{proof}

\begin{proposition}
Assuming replica symmetry, we define the following constants:
\begin{align}
    &\psi_1^{(a)} =  \beta p \| \bb{\overline{x}} \|^{p-2} \overline{x}^{a} \label{psi} , \\
    &\psi_2 = \beta p \overline{m}^{p-1} \label{psi2} , \\
    &A^2 = \frac{\beta K p}{N^{p/2} d^{p/4}} \sqrt{\frac{p!}{2}} \overline{p}_2 \overline{q}_2^{p/2-1}, \label{A^2}\\
    &B^2 = \frac{2\beta}{N^{p/2-1} d^{p/4}} \sqrt{\frac{p!}{2}}\overline{q}_2^{p/2}, \label{B^2} \\
    &C = \frac{2\beta}{N^{p/2-1} d^{p/4}} \sqrt{\frac{p!}{2}} \lr{\overline{q}_1^{p/2} - \overline{q}_2^{p/2}}, \label{C} \\
    &D = \frac{\beta K p}{N^{p/2} d^{p/4}} \sqrt{\frac{p!}{2}} \lr{\overline{p}_1 \overline{q}_1^{p/2-1} - \overline{p}_2 \overline{q}_2^{p/2-1}} \label{D}.
\end{align}

Then, the derivative of the interpolating  free energy simplifies to:
\begin{equation} \label{dA}
    \begin{aligned}
            \der[]{t}\mathcal{A}(t) =& \lr{1-p} \beta \| \bb{\overline{x}} \|^p - \beta \lr{\lambda - 1} \lr{1-p} \overline{m}^p + \\
            &- \frac{\beta K p}{2N^{p/2} d^{p/4}} \sqrt{\frac{p!}{2}} \lr{\overline{p}_1 \overline{q}_1^{p/2} - \overline{p}_2 \overline{q}_2^{p/2}} + \\
            & - \frac{\beta K}{N^{p/2} d^{p/4}} \sqrt{\frac{p!}{2}} \overline{q}_1^{p/2} - \beta \frac{\lambda -1}{N^{p/2}} \overline{q}_1^{p/2} .
    \end{aligned}
\end{equation}

\end{proposition}

\begin{proof}
We apply the replica-symmetry (RS) assumption \eqref{RS} to each term in the derivative of the interpolating free energy \eqref{der_1}.

Let \( \| \bb{\overline{x}} \| \) denote the mean of \( \| \bb{x}_1 \| \) under RS. Applying a binomial expansion \footnote{Newton’s formula: $\lr{a+b}^n = \sum_{k=0}^n \binom{n}{k}a^k b^{n-k}$ } around the mean, we obtain:
\begin{equation*}
    \begin{aligned}
        \avg{\|\bb{x_1}\|^p} = \avg{\lr{\bb{x_1} \cdot \bb{x_1}}}^{p/2} &= \sum_{k=0}^{p/2} \binom{p/2}{k} \avg{\lr{\bb{x_1} \cdot \bb{x_1} - \bb{\overline{x}} \cdot \bb{\overline{x}}}^k} \lr{\bb{\overline{x}} \cdot \bb{\overline{x}}}^{p/2 - k} \\
        &= \lr{\bb{\overline{x}} \cdot \bb{\overline{x}}}^{p/2} + \frac{p}{2}\lr{\bb{\overline{x}} \cdot \bb{\overline{x}}}^{p/2 -1} \avg{\bb{x_1} \cdot \bb{x_1}} - \frac{p}{2} \lr{\bb{\overline{x}} \cdot \bb{\overline{x}}}^{p/2} + V_N\lr{\bb{\overline{x}}} \\
        &= \lr{1- \frac{p}{2}}\lr{\bb{\overline{x}} \cdot \bb{\overline{x}}}^{p/2} + \frac{p}{2}\lr{\bb{\overline{x}} \cdot \bb{\overline{x}}}^{p/2-1} \avg{ \sum_{a=1}^d \lr{x_1^{\lr{a}}}^2 } + V_N\lr{\bb{\overline{x}}}.
    \end{aligned}
\end{equation*}

Observe that, as \( N \to \infty \), that is, in the thermodynamic limit, the term \\
$V_N\lr{\bb{\overline{x}}} = \sum_{k=2}^{p/2} \binom{p/2}{k} \avg{\lr{\bb{x_1} \cdot \bb{x_1} - \bb{\overline{x}} \cdot \bb{\overline{x}}}^k} \lr{\bb{\overline{x}} \cdot \bb{\overline{x}}}^{p/2 - k} \to 0$.
Moreover, since $\bb{\overline{x}} = \lr{\overline{x}^{\lr{a}}}_{a=1}^d$, each $\overline{x}^{\lr{a}}$ represents the mean of $x_1^{\lr{a}}$ aunder the RS assumption; thus, in the thermodynamic limit, we have $\avg{\lr{x_1^{\lr{a}} - \overline{x}^{\lr{a}}}^2} \to 0$. It follows that $ \avg{ \sum_{a=1}^d \lr{x_1^{\lr{a}}}^2 } \to \sum_{a=1}^d \slr{- \lr{\overline{x}^{\lr{a}}}^2 + 2 \bb{\overline{x}}^{a} \avg{x_1^{\lr{a}}} } $.

Hence,
\begin{equation*}
    \begin{aligned}
        \avg{\|\bb{x_1}\|^p} = \avg{\lr{\bb{x_1} \cdot \bb{x_1}}}^{p/2} &= \lr{1- \frac{p}{2}}\lr{\bb{\overline{x}} \cdot \bb{\overline{x}}}^{p/2} + \frac{p}{2}\lr{\bb{\overline{x}} \cdot \bb{\overline{x}}}^{p/2-1} \sum_{a=1}^d \slr{- \lr{\overline{x}^{\lr{a}}}^2 + 2 \bb{\overline{x}}^{a} \avg{x_1^{\lr{a}}} } \\
        &= \lr{1- \frac{p}{2}}\lr{\bb{\overline{x}} \cdot \bb{\overline{x}}}^{p/2} - \frac{p}{2}\lr{\bb{\overline{x}} \cdot \bb{\overline{x}}}^{p/2} + p \lr{\bb{\overline{x}} \cdot \bb{\overline{x}}}^{p/2-1} \sum_{a=1}^d \overline{x}^{a} \avg{x_1^{a}} \\
        &= \lr{1-p} \| \bb{\overline{x}} \|^p + p \| \bb{\overline{x}} \|^{p-2} \sum_{a=1}^d \overline{x}^{a} \avg{x_1^{a}} .
    \end{aligned}
\end{equation*}

Therefore,
\begin{equation*}
    \beta \avg{\|\bb{x_1}\|^p} - \beta p \| \bb{\overline{x}} \|^{p-2} \sum_{a=1}^d \overline{x}^{a} \avg{x_1^{a}} = \beta \lr{1-p} \| \bb{\overline{x}} \|^p.
\end{equation*}

We thus define:
\begin{equation*}
    \psi_1^{(a)} =  \beta p \| \bb{\overline{x}} \|^{p-2} \overline{x}^{a}.
\end{equation*}

Let \( \overline{m}\) denote the RS mean of \(m\). Applying a binomial expansion, we get:

\begin{equation*}
    \begin{aligned}
        \avg{m^p} &= \sum_{k=0}^p \binom{p}{k} \avg{\lr{m-\overline{m}}^k} \overline{m}^{p - k} \\
        &= \lr{1-p} \overline{m}^p + p \overline{m}^{p-1} \avg{m} + V_N\lr{\overline{m}}.
    \end{aligned}
\end{equation*}

Observe that, as \( N \to \infty \), the term
$V_N\lr{\overline{m}} = \sum_{k=1}^{p} \binom{p}{k} \avg{\lr{m-\overline{m}}^k} \overline{m}^{p - k} \to 0$.

Therefore, 

\begin{equation*}
    - \beta \lr{\lambda-1}\avg{m^p} + \psi_2 \lr{\lambda-1}\avg{m} = - \beta \lr{\lambda-1} \lr{1-p} \overline{m}^p.
\end{equation*}

We define: 
\begin{equation*}
    \psi_2 = \beta p \overline{m}^{p-1}.
\end{equation*}

Let \( \overline{q}_2 \) and \( \overline{p}_2 \) denote the RS mean of \( q_{12} \) and \( p_{12} \), respectively. Applying a binomial expansion yields:
\begin{equation*}
    \begin{aligned}
        \avg{p_{12}q_{12}^{p/2}} &= \avg{\lr{p_{12}-\overline{p}_2 + \overline{p}_2} \lr{q_{12}-\overline{q}_2 + \overline{q}_2}^{p/2}} \\
        &= \avg{\lr{p_{12}-\overline{p}_2} \lr{q_{12}-\overline{q}_2 + \overline{q}_2}^{p/2}} + \overline{p}_2\avg{\lr{q_{12}-\overline{q}_2 + \overline{q}_2}^{p/2}} \\
        &= \sum_{k=0}^{p/2} \binom{p/2}{k} \avg{ \lr{p_{12}-\overline{p}_2} \lr{q_{12}-\overline{q}_2}^2 }\overline{q}_2^{p/2 - k} + \overline{p}_2 \sum_{k=0}^{p/2} \binom{p/2}{k} \avg{\lr{q_{12}-\overline{q}_2}^k}\overline{q}_2^{p/2-k} \\
        &= \overline{q}_2^{p/2} \avg{p_{12}} + \frac{p}{2} \overline{p}_2 \overline{q}_2^{p/2-1} \avg{q_{12}} - \frac{p}{2} \overline{p}_2 \overline{q}_2^{p/2} + V_N^{\lr{1}} \lr{\overline{p}_2,\overline{q}_2} + V_N^{\lr{2}}\lr{\overline{p}_2,\overline{q}_2}.
    \end{aligned}
\end{equation*}

In the thermodynamic limit, the term \( V_N^{(1)}(\overline{p}_2, \overline{q}_2) + V_N^{(2)}(\overline{p}_2, \overline{q}_2) \to 0 \). 
Therefore,
\begin{equation*}
    \begin{aligned}
        - \frac{\beta K}{N^{p/2} d^{p/4}}\sqrt{\frac{p!}{2}} \avg{p_{12}q_{12}^{p/2}} + \frac{\beta K}{N^{p/2} d^{p/4}}\sqrt{\frac{p!}{2}} \overline{q}_2^{p/2} \avg{p_{11}} + \frac{\beta Kp}{2N^{p/2} d^{p/4}}\sqrt{\frac{p!}{2}} \overline{p}_2 \overline{q}_2^{p/2-1} \avg{q_{12}} = \frac{\beta Kp}{2N^{p/2} d^{p/4}}\sqrt{\frac{p!}{2}} \overline{p}_2 \overline{q}_2^{p/2},
    \end{aligned}
\end{equation*}
Thus, we define:
\begin{equation*}
    \begin{aligned}
        &A^2 = \frac{\beta Kp}{N^{p/2} d^{p/4}}\sqrt{\frac{p!}{2}} \overline{p}_2 \overline{q}_2^{p/2-1}, \\
        &B^2 = \frac{2\beta}{N^{p/2-1} d^{p/4}}\sqrt{\frac{p!}{2}} \overline{q}_2^{p/2}.
    \end{aligned}
\end{equation*}

Similarly, let \( \overline{q}_1 \) and \( \overline{p}_1 \) denote the RS mean of \( q_{11} \) and \( p_{11} \), respectively. Expanding via the binomial formula gives:
\begin{equation*}
    \begin{aligned}
        \avg{p_{11}q_{11}^{p/2}} &= \avg{\lr{p_{11}-\overline{p}_1 + \overline{p}_1} \lr{q_{11}-\overline{q}_1 + \overline{q}_1}^{p/2}} \\
        &= \avg{\lr{p_{11}-\overline{p}_1} \lr{q_{11}-\overline{q}_1 + \overline{q}_1}^{p/2}} + \overline{p}_1\avg{\lr{q_{11}-\overline{q}_1 + \overline{q}_1}^{p/2}} \\
        &= \sum_{k=0}^{p/2} \binom{p/2}{k} \avg{ \lr{p_{11}-\overline{p}_1} \lr{q_{11}-\overline{q}_1}^2 }\overline{q}_1^{p/2 - k} + \overline{p}_1 \sum_{k=0}^{p/2} \binom{p/2}{k} \avg{\lr{q_{11}-\overline{q}_1}^k}\overline{q}_1^{p/2-k} \\
        &= \overline{q}_1^{p/2} \avg{p_{11}} + \frac{p}{2} \overline{p}_1 \overline{q}_1^{p/2-1} \avg{q_{11}} - \frac{p}{2} \overline{p}_1 \overline{q}_1^{p/2} + V_N^{\lr{1}} \lr{\overline{p}_1,\overline{q}_1} + V_N^{\lr{2}}\lr{\overline{p}_1,\overline{q}_1}.
    \end{aligned}
\end{equation*}

Again, in the thermodynamic limit, \( V_N^{(1)}(\overline{p}_1, \overline{q}_1) + V_N^{(2)}(\overline{p}_1, \overline{q}_1) \to 0 \). Therefore,
\begin{equation*}
    \begin{aligned}
        \frac{\beta K}{N^{p/2} d^{p/4}}\sqrt{\frac{p!}{2}} \avg{p_{11}q_{11}^{p/2}} - \frac{\beta K}{N^{p/2} d^{p/4}}\sqrt{\frac{p!}{2}} \overline{q}_1^{p/2} \avg{p_{11}} - \frac{\beta Kp}{2N^{p/2} d^{p/4}}\sqrt{\frac{p!}{2}} \overline{p}_1 \overline{q}_1^{p/2-1} \avg{q_{11}} = - \frac{\beta Kp}{2N^{p/2} d^{p/4}}\sqrt{\frac{p!}{2}} \overline{p}_1 \overline{q}_1^{p/2},
    \end{aligned}
\end{equation*}
We define:
\begin{equation*}
    \begin{aligned}
        &\frac{A^2}{2} + \frac{D}{2} = \frac{\beta Kp}{2N^{p/2} d^{p/4}}\sqrt{\frac{p!}{2}} \overline{p}_1 \overline{q}_1^{p/2-1}, \\
        &\frac{B^2 K}{2N} + \frac{CK}{2N} = \frac{\beta K}{N^{p/2} d^{p/4}}\sqrt{\frac{p!}{2}} \overline{q}_1^{p/2}.
    \end{aligned}
\end{equation*}

Recalling the definitions of \( A^2 \) \eqref{A^2} and \( B^2 \) \eqref{B^2}, we obtain:
\begin{equation*}
    \begin{aligned}
        &C = \frac{2\beta}{N^{p/2-1} d^{p/4}} \sqrt{\frac{p!}{2}} \lr{\overline{q}_1^{p/2} - \overline{q}_2^{p/2}}, \\
        &D = \frac{\beta K p}{N^{p/2} d^{p/4}} \sqrt{\frac{p!}{2}} \lr{\overline{p}_1 \overline{q}_1^{p/2-1} - \overline{p}_2 \overline{q}_2^{p/2-1}}.
    \end{aligned}
\end{equation*}

Finally, the term $- \frac{\beta K}{N^{p/2} d^{p/4}} \sqrt{\frac{p!}{2}} \overline{q}_1^{p/2} - \beta \frac{\lambda -1}{N^{p/2}} \overline{q}_1^{p/2} $ arises directly from the RS identity \eqref{q_1 bar}.

\end{proof}

We must now evaluate the one-body contribution \( \mathcal{A}(t=0)\). 

\begin{proposition}
The Cauchy condition $\mathcal{A}(t = 0)$ in the thermodynamic limit reads as:
\begin{equation} \label{A0}
    \begin{aligned}
        \mathcal{A}\lr{t=0} &= \expect_{\bb{\eta}} \int Dz \ln \bigg[1 + \exp \bigg(\beta p \| \bb{\overline{x}} \|^{p-2} \lr{ \bb{\overline{x}} \cdot \bb{\eta}} - \beta p \lr{\lambda-1}\overline{m}^{p-1} + \\
        & \qquad +\sqrt{p \overline{p}_2 \overline{q}_2^{p/2 -1}}z + \frac{p}{2} \lr{\overline{p}_1 \overline{q}_1^{p/2-1} - \overline{p}_2 \overline{q}_2^{p/2-1} } \bigg) \bigg]+ \\
        &\quad+ \frac{\beta K}{N^{p/2} d^{p/4}} \sqrt{\frac{p!}{2}} \overline{q}_1^{p/2} + \frac{\beta^2 K p!}{2 N^{p-1} d^{p/2}} \lr{\overline{q}_1^p - \overline{q}_2^p}.
    \end{aligned}
\end{equation}
\end{proposition}

\begin{proof}
This follows from directly setting \( t=0 \) in equation \eqref{Z}. We obtain:
\begin{equation*}
    \begin{aligned}
            \mathcal{A}\lr{t=0} =& \frac{1}{N} \expect \ln \sum_{\bb{s}} \int D\boldsymbol z\: \exp \bigg( N\sum_{a=1}^d \psi_1^{\lr{a}} x_1^{\lr{a}} - N \psi_2 \lr{\lambda-1}m + \\
            & \quad + A \sum_{i=1}^N J_i s_i + B \sum_{\mu=2}^K J_\mu z_\mu + \frac{C}{2} \sum_{\mu=2}^K z_\mu^2 + \frac{D}{2} \sum_{i=1}^N s_i^2 \bigg).
    \end{aligned}
\end{equation*}

We separate the terms that depend on \( z_\mu \) from those that do not, leading to:
\begin{equation*}
    \begin{aligned}
        &\mathcal{A}_1 = \frac{1}{N} \expect \ln \sum_{\bb{s}} \exp \lr{ N\sum_{a=1}^d \psi_1^{\lr{a}} x_1^{\lr{a}} - N \psi_2 \lr{\lambda-1}m + A \sum_{i=1}^N J_i s_i + \frac{D}{2} \sum_{i=1}^N s_i^2 }, \\
        & \mathcal{A}_2 = \frac{1}{N} \expect \ln \int D\boldsymbol z\: \exp \lr{ + B \sum_{\mu=2}^K J_\mu z_\mu + \frac{C}{2} \sum_{\mu=2}^K z_\mu^2 }.
    \end{aligned}
\end{equation*}
Let us first analyze \( \mathcal{A}_1 \). Using the definitions \( x_1^{(a)} = \frac{1}{N} \sum_{i=1}^N \eta_i^{1, (a)} s_i \) and \( m = \frac{1}{N} \sum_{i=1}^N s_i \), we get:
\begin{equation*}
    \begin{aligned}
        \mathcal{A}_1 &= \frac{1}{N} \expect_{\bb{\eta}} \ln \sum_{\bb{s}} \exp \lr{\sum_{a=1}^d \psi_1^{\lr{a}} \sum_{i=1}^N \eta_i^{1, \lr{a}} s_i - \psi_2 \lr{\lambda-1} \sum_{i=1}^N s_i + A \sum_{i=1}^N J_i s_i + \frac{D}{2} \sum_{i=1}^N s_i^2 } \\
        &= \expect_{\bb{\eta}} \int Dz \ln \slr{1 + \exp \lr{ \sum_{a=1}^d \psi_1^{\lr{a}} \eta^{1,\lr{a}} - \psi_2 \lr{\lambda-1} +Az + \frac{D}{2} }} ,
    \end{aligned}
\end{equation*}
with \( z \sim \mathcal{N}(0, 1) \).

Substituting the definitions of \( \psi_1^{(a)} \), \(\psi_2\), \( A \), and \( D \), we obtain:
\begin{equation*}
    \begin{aligned}
        \mathcal{A}_1 =&  \expect_{\bb{\eta}} \int Dz \ln \Bigg[ 1 + \exp \bigg( \beta p \| \bb{\overline{x}} \|^{p-2} \lr{\bb{\overline{x}} \cdot \bb{\eta}} - \beta p \lr{\lambda-1} \overline{m}^{p-1} + \\
        & \quad +\sqrt{ \frac{\beta K p}{N^{p/2} d^{p/4}} \sqrt{\frac{p!}{2}}  \overline{p}_2 \overline{q}_2^{p/2 -1}}z  +\frac{\beta K p}{2N^{p/2} d^{p/4}} \sqrt{\frac{p!}{2}} \lr{\overline{p}_1 \overline{q}_1^{p/2-1} - \overline{p}_2 \overline{q}_2^{p/2-1} } \bigg) \Bigg].
    \end{aligned}
\end{equation*}
Now rescaling: 
\begin{equation} \label{risc}
    \frac{\beta K }{N^{p/2} d^{p/4}} \sqrt{\frac{p!}{2}} \overline{p}_1 \to \overline{p}_1 \qquad \text{and} \qquad \frac{\beta K}{N^{p/2} d^{p/4}} \sqrt{\frac{p!}{2}} \overline{p}_2 \to \overline{p}_2,
\end{equation} we obtain:
\begin{equation*}
    \mathcal{A}_1 = \expect_{\bb{\eta}} \int Dz \ln \slr{1 + \exp \lr{\beta p \| \bb{\overline{x}} \|^{p-2} \lr{\bb{\overline{x}} \cdot \bb{\eta}} - \beta p \lr{\lambda-1} \overline{m}^{p-1} + \sqrt{p \overline{p}_2 \overline{q}_2^{p/2 -1}}z + \frac{p}{2} \lr{\overline{p}_1 \overline{q}_1^{p/2-1} - \overline{p}_2 \overline{q}_2^{p/2-1} } } }.
\end{equation*}
We now consider \( \mathcal{A}_2 \). Using the Gaussian integral and recalling that \( J_\mu \sim \mathcal{N}(0,1) \), we find:

\begin{equation*}
    \begin{aligned}
        \mathcal{A}_2 &= \frac{1}{N} \expect \slr{ \ln \prod_{\mu>1} \int \exp \lr{-\frac{1-C}{2} z_\mu^2 + B J_\mu z_\mu} } \\
        &= - \frac{K}{2N} \ln \lr{1-C} + \frac{B^2 K}{2N\lr{1-C}} \expect_J\slr{J_\mu^2} \\
        &= - \frac{K}{2N} \slr{\ln \lr{1-C} + \frac{B^2}{\lr{1-C}}}.
    \end{aligned}
\end{equation*}

Expanding \( \ln (1-C) \) and \( \frac{1}{1-C} \) in Taylor series\footnote{
Taylor expansions:
\( \ln(1-C) = -C - \frac{C^2}{2} + \mathcal{O}(C^3), \quad \frac{1}{1-C} = 1 + C + C^2 + \mathcal{O}(C^3) \).
}, and inserting the definitions of \( C \) and \( B^2 \), we obtain:

\begin{equation*}
    \begin{aligned}
        \mathcal{A}_2 &= \frac{K}{2N} \Bigg[ \frac{2\beta}{N^{p/2-1} d^{p/4}} \sqrt{\frac{p!}{2}} \lr{ \overline{q}_1^{p/2} - \overline{q}_2^{p/2} } + \frac{2\beta^2}{N^{p-2} d^{p/2}} \frac{p!}{2} \lr{\overline{q}_1^p + \overline{q}_2^p -2\overline{q}_1^{p/2}\overline{q}_2^{p/2} } + \\
        &\quad \quad + \frac{2\beta}{N^{p/2-1} d^{p/4}} \sqrt{\frac{p!}{2}}\overline{q}_2^{p/2} + 
        \frac{2\beta}{N^{p/2-1} d^{p/4}} \sqrt{\frac{p!}{2}}\overline{q}_2^{p/2} \lr{ \frac{2\beta}{N^{p/2-1} d^{p/4}} \sqrt{\frac{p!}{2}} \lr{\overline{q}_1^{p/2} - \overline{q}_2^{p/2}} } \Bigg] \\
        &= \frac{K}{2N} \Bigg[ \frac{2\beta}{N^{p/2-1} d^{p/4}} \sqrt{\frac{p!}{2}} \overline{q}_1^{p/2} + \frac{2\beta^2}{N^{p-2} d^{p/2}} \frac{p!}{2} \overline{q}_1^p + \frac{2\beta^2}{N^{p-2} d^{p/2}} \frac{p!}{2} \overline{q}_2^p + \\ 
        &\quad \quad - \frac{4\beta^2}{N^{p-2} d^{p/2}} \frac{p!}{2}\overline{q}_1^{p/2}\overline{q}_2^{p/2}  +
        \frac{4\beta^2}{N^{p-2} d^{p/2}} \frac{p!}{2} \overline{q}_1^{p/2} \overline{q}_2^{p/2} - \frac{4\beta^2}{N^{p-2} d^{p/2}} \frac{p!}{2} \overline{q}_2^{p} \Bigg] \\
        &= \frac{\beta K}{N^{p/2} d^{p/4}} \sqrt{\frac{p!}{2}} \overline{q}_1^{p/2} + \frac{\beta^2 K}{N^{p-1} d^{p/2}} \frac{p!}{2} \lr{ \overline{q}_1^p -\overline{q}_2^p }.
    \end{aligned}
\end{equation*}

Therefore, we obtain Eq. \eqref{A0}.

\end{proof}

Applying eq. \eqref{interp}, we obtain the following result.

\begin{theorem}
In the thermodynamic limit, the replica-symmetric quenched free energy of the dense Battaglia-Treves model, which includes McCulloch-Pitts neurons as described in eq. \eqref{eq:ham}, for the $S^{d-1}$ embedding space, can be expressed in terms of the (mean values of the) order parameters $\| \bb{\overline{x}} \|, \overline{q}_1, \overline{q}_2, \overline{p}_1, \overline{p}_2$, and the control parameters $\alpha, \beta$, as follows:

\begin{equation} \label{ARS}
    \begin{aligned}
        A\lr{\alpha,\beta} &= \lr{1-p} \beta \| \bb{\overline{x}} \|^p - \beta \lr{\lambda-1}\lr{1-p}\overline{m}^p + \alpha \beta^2 \lr{ \overline{q}_1^p -\overline{q}_2^p } - \frac{p}{2} \lr{\overline{p}_1 \overline{q}_1^{p/2} - \overline{p}_2 \overline{q}_2^{p/2}} + \\
        & \quad +\expect_{\bb{\eta}} \int Dz \ln \bigg[1 + \exp \bigg(\beta p \| \bb{\overline{x}} \|^{p-2} \lr{\bb{\overline{x}} \cdot \bb{\eta}} - \beta p \lr{\lambda-1}\overline{m}^{p-1} + \\
        & \qquad + \sqrt{p \overline{p}_2 \overline{q}_2^{p/2 -1}}z + \frac{p}{2} \lr{\overline{p}_1 \overline{q}_1^{p/2-1} - \overline{p}_2 \overline{q}_2^{p/2-1} } \bigg) \bigg].
    \end{aligned}
\end{equation}

\end{theorem}

\begin{proof}
The result follows directly by substituting eq.~\eqref{dA} and eq.~\eqref{A0} into eq.~\eqref{interp}, namely:
\begin{equation*}
    \begin{aligned}
        A\lr{\alpha,\beta} &= \beta \lr{1-p} \| \bb{\overline{x}} \|^p - \beta \lr{\lambda-1}\lr{1-p}\overline{m}^p - \beta \frac{\lambda-1}{N^{p/2}} \overline{q}_1^{p/2} + \\
        & - \frac{\beta K}{N^{p/2} d^{p/4}} \sqrt{\frac{p!}{2}} \overline{q}_1^{p/2} - \frac{\beta K p}{2N^{p/2} d^{p/4}} \sqrt{\frac{p!}{2}} \lr{\overline{p}_1 \overline{q}_1^{p/2} - \overline{p}_2 \overline{q}_2^{p/2}} + \\
        & \quad + \expect_{\bb{\eta}} \int Dz \ln \bigg[1 + \exp \bigg(\beta p \| \bb{\overline{x}} \|^{p-2} \lr{\bb{\overline{x}} \cdot \bb{\eta}} - \beta p \lr{\lambda-1} \overline{m^{p-1}} + \\
        & \qquad + \sqrt{p \overline{p}_2 \overline{q}_2^{p/2 -1}}z + \frac{p}{2} \lr{\overline{p}_1 \overline{q}_1^{p/2-1} - \overline{p}_2 \overline{q}_2^{p/2-1} } \bigg) \bigg] + \\ 
        &\quad+ \frac{\beta K}{N^{p/2} d^{p/4}} \sqrt{\frac{p!}{2}} \overline{q}_1^{p/2} + \frac{\beta^2 K p!}{2 N^{p-1} d^{p/2}} \lr{\overline{q}_1^p - \overline{q}_2^p}.
    \end{aligned}
\end{equation*}

\noindent Observe that, as \( N \to \infty \), the term \( - \beta \frac{\lambda-1}{N^{p/2}} \overline{q}_1^{p/2}\) vanishes. Moreover, applying the rescaling defined in eq.~\eqref{risc}, and substituting \( K = \frac{2 \alpha N^a d^{p/2}}{p!} \), we obtain:

\begin{equation*}
    \begin{aligned}
        A\lr{\alpha,\beta} &= \beta \lr{1-p} \| \bb{\overline{x}} \|^p - \beta \lr{\lambda-1} \lr{1-p} \overline{m}^p - \frac{p}{2} \lr{\overline{p}_1 \overline{q}_1^{p/2} - \overline{p}_2 \overline{q}_2^{p/2}} + \alpha N^{a-p+1} \beta^2 \lr{\overline{q}_1^p - \overline{q}_2^p} + \\
        & \quad + \expect_{\bb{\eta}} \int Dz \ln \bigg[1 + \exp \bigg(\beta p \| \bb{\overline{x}} \|^{p-2} \lr{\bb{\overline{x}} \cdot \bb{\eta}} - \beta p \lr{\lambda-1} \overline{m}^{p-1} + \\
        & \qquad + \sqrt{p \overline{p}_2 \overline{q}_2^{p/2 -1}}z + \frac{p}{2} \lr{\overline{p}_1 \overline{q}_1^{p/2-1} - \overline{p}_2 \overline{q}_2^{p/2-1} } \bigg) \bigg].
    \end{aligned}
\end{equation*}

To ensure that the free energy remains finite and well-defined in the limit \( N \to \infty \), it is necessary that \( a \leq p - 1 \). Therefore, in the thermodynamic limit, by setting \( a = p - 1 \) with even \( p \geq 4 \), we recover the desired expression.

\end{proof}

We now derive the self-consistency equations for all the order parameters involved in the replica-symmetric solution of the dense Battaglia-Treves model, including both the primary ones \( \| \bb{\overline{x}} \| \), \(\overline{m}\), \(\overline{q}_1\), \(\overline{q}_2\), and the auxiliary parameters \(\overline{p}_1\), \(\overline{p}_2\). Notably, the equations for \(\overline{p}_1\) and \(\overline{p}_2\) depend explicitly on \(\overline{q}_1\) and \(\overline{q}_2\). This allows us to eliminate the auxiliary variables by substituting their expressions back into the replica-symmetric free energy \(A(\alpha, \beta)\), given in eq.~\eqref{ARS}, so that the free energy depends solely on the primary order parameters \( \| \bb{\overline{x}} \| \), \(\overline{m}\), \(\overline{q}_1\), and \(\overline{q}_2\).

This reformulation simplifies the theoretical framework and allows us to construct the phase diagram in the space of control parameters \((\alpha, \beta)\). The diagram reveals the regions corresponding to different dynamical phases and identifies the critical thresholds for memory retrieval and storage.

In this way, the phase diagram offers a clear understanding of how the interplay between \(\alpha\) and \(\beta\) governs the collective dynamics and memory performance of the network in the high-storage regime.

\begin{theorem} 
In the thermodynamic limit, the replica-symmetric quenched free energy of the dense Battaglia-Treves model, equipped with McCulloch-Pitts neurons as described in eq. \eqref{eq:ham}, for the $S^{d-1}$ embedding space, can be expressed in terms of the (mean values of the) order parameters $\| \bb{\overline{x}}\|, \overline{m},  \overline{q}_1, \overline{q}_2$, and the control parameters $\alpha, \beta$, as follows:

\begin{equation} \label{Amean}
    \begin{aligned}
        A\lr{\alpha,\beta} &= \lr{1-p} \beta \| \bb{\overline{x}} \|^p - \beta \lr{\lambda-1} \lr{1-p} \overline{m}^p + \lr{1-p} \alpha \beta^2 \lr{ \overline{q}_1^p -\overline{q}_2^p } + \\
        & \quad +\expect_{\bb{\eta}} \int Dz \ln \bigg[1 + \exp \bigg(\beta p \| \bb{\overline{x}} \|^{p-2} \lr{\bb{\overline{x}} \cdot \bb{\eta}} - \beta p \lr{\lambda-1}\overline{m}^{p-1} + \\
        & \qquad + \alpha \beta^2 p \lr{ \overline{q}_1^{p-1} - \overline{q}_2^{p-1}} + \beta \sqrt{2 \alpha p \overline{q}_2^{p-1}} z \bigg) \bigg].
    \end{aligned}
\end{equation}

By extremizing the replica-symmetric free energy \( A(\alpha, \beta) \) with respect to the order parameters, we obtain the following self-consistency equations:

\begin{align}
    & \| \bb{\overline{x}} \| ^2 = \int D\boldsymbol z\:\avg{\lr{\bb{\overline{x}} \cdot \bb{\eta}} \sigma \lr{ \beta h\lr{z}}}_{\bb{\eta}}, \label{xbar} \\
    & \overline{m} = \overline{q}_1 = \int D\boldsymbol z\: \avg{\sigma \lr{\beta h\lr{z}}}_{\bb{\eta}}, \label{q1bar} \\
    & \overline{q}_2 = \int D\boldsymbol z\: \avg{\sigma^2 \lr{\beta h\lr{z}}}_{\bb{\eta}}, \label{q2bar}
\end{align}

where \(\sigma(t) = \frac{1}{1 + e^{-t}}\) is the sigmoid activation function, \(D\boldsymbol z\:\) represents the Gaussian measure for \(z \sim \mathcal{N}(0,1)\), and \(h(z)\) is the internal field acting on the neurons and reads as
\begin{equation}
    h\lr{z} = p \| \bb{\overline{x}} \|^{p-2} \lr{\bb{\overline{x}} \cdot \bb{\eta}} - p \lr{\lambda-1}\overline{m}^{p-1} + \alpha \beta p \lr{ \overline{q}_1^{p-1} - \overline{q}_2^{p-1}} + \sqrt{2 \alpha p \overline{q}_2^{p-1}} z.
\end{equation}
These equations describe the self-consistent relationships between the order parameters and the control parameters of the system (\(\alpha, \beta\)).
\end{theorem}

\begin{proof}
Let \( f(z) \) denote the argument of the exponential in Eq.~\eqref{ARS}, namely:
\begin{equation} \label{f}
    f(z) = \beta p \| \bb{\overline{x}} \|^{p-2} \lr{\bb{\overline{x}} \cdot \bb{\eta}} - \beta p \lr{\lambda-1}\overline{m}^{p-1} + \frac{p}{2} \lr{\overline{p}_1 \overline{q}_1^{p/2-1} - \overline{p}_2 \overline{q}_2^{p/2-1} } + \sqrt{p \overline{p}_2 \overline{q}_2^{p/2 -1}}z.
\end{equation}

We begin by extremizing Eq.~\eqref{ARS} with respect to \( \bb{\overline{x}} \):

\begin{equation} \label{x}
    \nabla_{ \bb{\overline{x}} } A \lr{\alpha, \beta} = 0 \Leftrightarrow \beta \lr{p-1} \nabla_{\bb{\overline{x}}} \| \bb{\overline{x}} \|^p = \expect_{\bb{\eta}} \int Dz \frac{\nabla_{\bb{\overline{x}}} \exp \slr{ f\lr{z}} }{1 + \exp \slr{f\lr{z}}}
\end{equation}

For the left-hand side we have:
\begin{equation} \label{x_M1}
    \beta \lr{p-1} \nabla_{\bb{\overline{x}}} \| \bb{\overline{x}} \|^p = \beta p \lr{p-1} \| \bb{\overline{x}} \|^{p-2} \bb{\overline{x}}.
\end{equation}

On the right-hand side:
\begin{equation} \label{x_M2}
    \nabla_{\bb{\overline{x}}} \exp \slr{ f\lr{z}} = \nabla_{\bb{\overline{x}}} f\lr{z} \cdot \exp \slr{ f\lr{z}} = \beta p \nabla_{\bb{\overline{x}}} \slr{ \| \bb{\overline{x}} \|^{p-2} \lr{\bb{\overline{x}} \cdot \bb{\eta} } } \cdot \exp \slr{ f\lr{z}},
\end{equation}
with
\begin{equation*}
    \begin{aligned}
            \nabla_{\bb{\overline{x}}} \slr{ \| \bb{\overline{x}} \|^{p-2} \lr{\bb{\overline{x}} \cdot \bb{\eta} } } &= \nabla_{\bb{\overline{x}}} \slr{ \| \bb{\overline{x}} \|^{p-2} } \lr{\bb{\overline{x}} \cdot \bb{\eta} } + \| \bb{\overline{x}} \|^{p-2} \nabla_{\bb{\overline{x}}} \slr{ \lr{\bb{\overline{x}} \cdot \bb{\eta} } } \\
            &= \lr{p-2} \| \bb{\overline{x}} \|^{p-4} \bb{\overline{x}} \lr{\bb{\overline{x}} \cdot \bb{\eta}} + \| \bb{\overline{x}} \|^{p-2} \bb{\eta}.
    \end{aligned}
\end{equation*}

Substituting into Eq. \eqref{x_M2}, we obtain:
\begin{equation*}
    \begin{aligned}
        \nabla_{\bb{\overline{x}}} \exp \slr{ f\lr{z}} = \beta p \slr{\lr{p-2} \| \bb{\overline{x}} \|^{p-4} \bb{\overline{x}} \lr{\bb{\overline{x}} \cdot \bb{\eta}} + \| \bb{\overline{x}} \|^{p-2} \bb{\eta}} \exp \slr{ f\lr{z}}.
    \end{aligned}
\end{equation*}

It follows that:
\begin{equation} \label{sigma_x}
    \frac{\nabla_{\bb{\overline{x}}} \exp \slr{ f\lr{z}} }{1 + \exp \slr{f\lr{z}}} = \beta p \slr{\lr{p-2} \| \bb{\overline{x}} \|^{p-4} \bb{\overline{x}} \lr{\bb{\overline{x}} \cdot \bb{\eta}} + \| \bb{\overline{x}} \|^{p-2} \bb{\eta}} \sigma \slr{ f\lr{z}}.
\end{equation}

Substituting Eqs. \eqref{x_M1} and \eqref{sigma_x} into Eq. \eqref{x}, we get:
\begin{equation*}
    \begin{aligned}
        \nabla_{ \bb{\overline{x}} } A \lr{\alpha, \beta} = 0 & \Leftrightarrow \beta p \lr{p-1} \| \bb{\overline{x}} \|^{p-2} \bb{\overline{x}} = \expect_{\bb{\eta}} \int Dz \beta p \slr{\lr{p-2} \| \bb{\overline{x}} \|^{p-4} \bb{\overline{x}} \lr{\bb{\overline{x}} \cdot \bb{\eta}} + \| \bb{\overline{x}} \|^{p-2} \bb{\eta}} \sigma \slr{ f\lr{z}} \\
        & \Leftrightarrow \bb{\overline{x}} = \frac{1}{p-1} \expect_{\bb{\eta}} \int Dz \slr{\lr{p-2} \frac{\bb{\overline{x}} \lr{\bb{\overline{x}} \cdot \bb{\eta}}}{\| \bb{\overline{x}} \|^2 } + \bb{\eta}} \sigma \slr{ f\lr{z}}.
    \end{aligned}
\end{equation*}

Multiplying both sides by \( \bb{\overline{x}} \), we find:
\begin{equation*}
    \| \bb{\overline{x}} \|^2 = \frac{1}{p-1} \expect_{\bb{\eta}} \int Dz \slr{\lr{p-2} \lr{\bb{\overline{x}} \cdot \bb{\eta}} + \bb{\overline{x}} \cdot \bb{\eta}} \sigma \slr{ f\lr{z}}.
\end{equation*}

Therefore:
\begin{equation}
    \| \bb{\overline{x}} \|^2 = \expect_{\bb{\eta}} \int Dz \lr{\bb{\overline{x}} \cdot \bb{\eta}} \sigma \slr{ f\lr{z}}.
\end{equation}

Next, we extremize Eq.~\eqref{ARS} with respect to \( \overline{m} \):
\begin{equation} \label{mbar}
    \begin{aligned}
        \pder[]{\overline{m}} A \lr{\alpha, \beta} = 0 & \Leftrightarrow \beta \lr{\lambda-1}\lr{1-p}p\,\overline{m}^{p-1} = \expect_{\bb{\eta}} \int Dz \frac{ \pder[]{\overline{m}} \exp \slr{ f\lr{z}} }{1 + \exp \slr{f\lr{z}}} \\
        & \Leftrightarrow \beta \lr{\lambda-1}\lr{1-p}p\,\overline{m}^{p-1} = \expect_{\bb{\eta}} \int Dz \, \beta \lr{\lambda-1} \lr{1-p} p\,\overline{m}^{p-2} \sigma \slr{f\lr{z}} \\
        & \Leftrightarrow \overline{m} = \expect_{\bb{\eta}} \int Dz \, \sigma \slr{f\lr{z}}.
    \end{aligned}
\end{equation}

Now, we extremize Eq.~\eqref{ARS} with respect to \( \overline{p}_1 \):
\begin{equation} \label{p1}
    \begin{aligned}
        \pder[]{\overline{p}_1} A \lr{\alpha, \beta} = 0 & \Leftrightarrow \frac{p}{2} \overline{q}_1^{p/2} \pder[]{\overline{p}_1} \overline{p}_1 = \expect_{\bb{\eta}} \int Dz \frac{ \pder[]{\overline{p}_1} \exp \slr{ f\lr{z}} }{1 + \exp \slr{f\lr{z}}} \\
        & \Leftrightarrow \frac{p}{2} \overline{q}_1^{p/2} = \expect_{\bb{\eta}} \int Dz \, \frac{p}{2} \overline{q}_1^{p/2 - 1} \sigma \slr{f\lr{z}} \\
        & \Leftrightarrow \overline{q}_1 = \expect_{\bb{\eta}} \int Dz \, \sigma \slr{f\lr{z}}.
    \end{aligned}
\end{equation}

Similarly, extremizing with respect to \( \overline{p}_2 \), we obtain:
\begin{equation} \label{p2}
    \begin{aligned}
        \pder[]{\overline{p}_2} A \lr{\alpha, \beta} = 0 & \Leftrightarrow \frac{p}{2} \overline{q}_2^{p/2} \pder[]{\overline{p}_2} \overline{p}_2 = - \expect_{\bb{\eta}} \int Dz \frac{ \pder[]{\overline{p}_2} \exp \slr{ f\lr{z}} }{1 + \exp \slr{f\lr{z}}} \\
        & \Leftrightarrow \frac{p}{2} \overline{q}_2^{p/2} = \expect_{\bb{\eta}} \int Dz \, \slr{ \frac{p}{2} \overline{q}_2^{p/2 - 1} - \frac{p \overline{q}_2^{p/2 -1} z}{2 \sqrt{p \overline{p}_2 \overline{q}_2^{p/2 -1}} } } \sigma \slr{f\lr{z}} \\
        & \Leftrightarrow \overline{q}_2 = \expect_{\bb{\eta}} \int Dz \, \sigma \slr{f\lr{z}}  - \frac{1}{\sqrt{p \overline{p}_2 \overline{q}_2^{p/2 -1}}} \, \expect_{\bb{\eta}} \int Dz \,  z  \sigma \slr{f\lr{z}}.
    \end{aligned}
\end{equation}

Using Wick’s theorem and the identity \( \mathbb{E}[z^2] = 1 \), we have\footnote{ \( \der[]{x} \sigma \lr{f\lr{x}} = f'\lr{x} \slr{1 - \sigma \lr{f\lr{x}}} \sigma \lr{f\lr{x}} \) }:
\begin{equation} \label{wick}
    \begin{aligned}
        \expect_{\bb{\eta}} \int Dz \,  z  \sigma \slr{f\lr{z}} &= \expect_{\bb{\eta}, z} \slr{z^2} \expect_{\bb{\eta}, z} \slr{ \pder[]{z} \sigma \slr{f\lr{z}} } \\
        &= \expect_{\bb{\eta}, z} \sqrt{p \overline{p}_2 \overline{q}_2^{p/2 -1}} \slr{1 - \sigma \slr{f\lr{z}} } \sigma \slr{f\lr{z}}.
    \end{aligned}
\end{equation}

Substituting into Eq.~\eqref{p2}, we get:
\begin{equation}
    \overline{q}_2 = \expect_{\bb{\eta}} \int Dz \, \sigma^2 \slr{h\lr{z}}.
\end{equation}

Extremizing \( A(\alpha, \beta) \) with respect to \( \overline{q}_1 \), we obtain:
\begin{equation} \label{q1}
    \begin{aligned}
        \pder[]{\overline{q}_1} A \lr{\alpha, \beta} = 0 & \Leftrightarrow \alpha \beta^2 \pder[]{\overline{q}_1} \lr{\overline{q}_1^p - \overline{q}_2^p} -  \frac{p}{2} \pder[]{\overline{q}_1} \lr{\overline{p}_1\overline{q}_1^{p/2} - \overline{p}_2\overline{q}_2^{p/2}} + \expect_{\bb{\eta}} \int Dz \, \frac{ \pder[]{\overline{q}_1} \exp \slr{ f\lr{z}} }{1 + \exp \slr{f\lr{z}}} = 0 \\
        & \Leftrightarrow \alpha \beta^2 p \overline{q}_1^{p-1} - \frac{p^2}{4} \overline{q}_1^{p/2 -1} + \expect_{\bb{\eta}} \int Dz \, \slr{ \frac{p}{2} \overline{p}_1 \lr{\frac{p}{2} -1} \overline{q}_1^{p/2 - 2} } \sigma \slr{f\lr{z}} = 0 \\
        & \Leftrightarrow 2 \alpha \beta^2 \overline{q}_1^{p-1} - \frac{p}{2} \overline{q}_1^{p/2 -1} + \slr{ \overline{p}_1 \lr{\frac{p}{2} -1} \overline{q}_1^{p/2 - 2} } \expect_{\bb{\eta}} \int Dz \, \sigma \slr{f\lr{z}} = 0 \\
        & \Leftrightarrow 2 \alpha \beta^2 \overline{q}_1^{p-1} - \frac{p}{2} \overline{q}_1^{p/2 -1} + \overline{p}_1 \lr{\frac{p}{2} -1} \overline{q}_1^{p/2 - 1} = 0 \\
        & \Leftrightarrow \overline{p}_1 = 2 \alpha \beta^2 \overline{q}_1^{p/2}.
    \end{aligned}
\end{equation}

Extremizing with respect to \( \overline{q}_2 \), similarly:
\begin{equation} \label{q2}
    \begin{aligned}
        \pder[]{\overline{q}_2} A \lr{\alpha, \beta} = 0 &\Leftrightarrow \alpha \beta^2 \pder[]{\overline{q}_2} \lr{\overline{q}_1^p - \overline{q}_2^p} -  \frac{p}{2} \pder[]{\overline{q}_2} \lr{\overline{p}_1\overline{q}_1^{p/2} - \overline{p}_2\overline{q}_2^{p/2}} + \expect_{\bb{\eta}} \int Dz \, \frac{ \pder[]{\overline{q}_2} \exp \slr{ f\lr{z}} }{1 + \exp \slr{f\lr{z}}} = 0 \\
        & \Leftrightarrow - \alpha \beta^2 p \overline{q}_2^{p-1} + \frac{p^2}{4} \overline{q}_2^{p/2 -1} + \\
        & \quad + \expect_{\bb{\eta}} \int Dz \, \slr{ \frac{p}{2} \frac{\overline{p}_2 \lr{\frac{p}{2} -1} \overline{q}_2^{p/2 - 2} }{\sqrt{p \overline{p}_2 \overline{q}_2^{p/2 - 1}}}z - \frac{p}{2} \overline{p}_2 \lr{\frac{p}{2} -1} \overline{q}_2^{p/2 - 2} } \sigma \slr{f\lr{z}} = 0 \\
        & \Leftrightarrow - 2 \alpha \beta^2 \overline{q}_2^{p-1} + \frac{p}{2} \overline{q}_2^{p/2 -1} + \frac{\overline{p}_2 \lr{\frac{p}{2} -1} \overline{q}_2^{p/2 - 2} }{\sqrt{p \overline{p}_2 \overline{q}_2^{p/2 - 1}}} \expect_{\bb{\eta}} \int Dz \,  z \sigma \slr{f\lr{z}} + \\
        & \quad - \overline{p}_2 \lr{\frac{p}{2} -1} \overline{q}_2^{p/2 - 2} \, \expect_{\bb{\eta}} \int Dz \,  \sigma \slr{f\lr{z}} = 0
    \end{aligned}
\end{equation}

Using Stein lemma (i.e. Wick’s theorem), the identity \( \mathbb{E}[z^2] = 1 \) and substituting \eqref{wick} into Eq. \eqref{q2}, we get:

\begin{equation}
    \begin{aligned}
        \pder[]{\overline{q}_2} A \lr{\alpha, \beta} = 0 &\Leftrightarrow - 2 \alpha \beta^2 \overline{q}_2^{p-1} + \frac{p}{2} \overline{q}_2^{p/2 -1} + \overline{p}_2 \lr{\frac{p}{2} -1} \overline{q}_2^{p/2 - 2} \expect_{\bb{\eta}} \int Dz \,  \slr{1 - \sigma \slr{f\lr{z}} } \sigma \slr{f\lr{z}} + \\
        & \quad - \overline{p}_2 \lr{\frac{p}{2} -1} \overline{q}_2^{p/2 - 2} \, \expect_{\bb{\eta}} \int Dz \,  \sigma \slr{f\lr{z}} = 0 \\
        &\Leftrightarrow - 2 \alpha \beta^2 \overline{q}_2^{p-1} + \frac{p}{2} \overline{q}_2^{p/2 -1} - \overline{p}_2 \lr{\frac{p}{2} -1} \overline{q}_2^{p/2 - 2} \, \expect_{\bb{\eta}} \int Dz \,  \sigma^2 \slr{f\lr{z}} = 0 \\
        &\Leftrightarrow - 2 \alpha \beta^2 \overline{q}_2^{p-1} + \frac{p}{2} \overline{q}_2^{p/2 -1} - \overline{p}_2 \lr{\frac{p}{2} -1} \overline{q}_2^{p/2 - 1} = 0 \\
        &\Leftrightarrow \overline{p}_2 = 2 \alpha \beta^2 \overline{q}_2^{p/2}.
    \end{aligned}
\end{equation}

Substituting
\begin{equation*}
    \overline{p}_1 = 2 \alpha \beta^2 \overline{q}_1^{p/2}, \qquad
    \overline{p}_2 = 2 \alpha \beta^2 \overline{q}_2^{p/2},
\end{equation*}
into the expression \eqref{ARS} of the free energy , we recover Eq. \eqref{Amean}.

Substituting \( \overline{p}_1 \) and \( \overline{p}_2 \) into Eq. \eqref{f}, we get:
\begin{equation}\label{f_new}
        f(z) = \beta \slr{p \| \bb{\overline{x}} \|^{p-2} \lr{\bb{\overline{x}} \cdot \bb{\eta}} - p \lr{\lambda-1}\overline{m}^{p-1} + \alpha \beta p \lr{ \overline{q}_1^{p-1} - \overline{q}_2^{p-1} } + \sqrt{2 \alpha p \overline{q}_2^{p -1}}z},
\end{equation}
from which the internal field \( h(z) \) is identified.

Finally, inserting Eq.~\eqref{f_new} into the previously derived self-consistency conditions yields the Eqs. \eqref{xbar}, \eqref{q1bar} and \eqref{q2bar}.

\end{proof}

\begin{corollary}
To facilitate the numerical solution of the self-consistent equations \eqref{xbar}–\eqref{q2bar}, it is helpful to perform a change of variables. The equations can then be rewritten as:
\begin{align} 
&\overline{x} = \frac{1}{\pi} \int Dz \int_{0}^{\pi} d\theta \, \cos(\theta) \ \sigma(\beta h(z, \theta)), \label{xbar_new} \\ 
&\overline{m} = \overline{q}_1 = \frac{1}{\pi} \int Dz \int_{0}^{\pi} d\theta \ \sigma(\beta h(z, \theta)), \label{q1bar_new} \\ 
&\overline{q}_2 = \frac{1}{\pi} \int Dz \int_{0}^{\pi} d\theta \ \sigma^2(\beta h(z, \theta)), \label{q2bar_new}
\end{align}

where the function \(h(z, t)\) is defined as:
\begin{equation} \label{campo}
h(z,\theta) = p \overline{x}^{p-1}\cos\theta - p \lr{\lambda-1}\overline{m}^{p-1} + \alpha \beta p \lr{ \overline{q}_1^{p-1} - \overline{q}_2^{p-1}} + \sqrt{2 \alpha p \overline{q}_2^{p-1}} z.
\end{equation}
\end{corollary}

\begin{proof}
The equations \eqref{xbar_new}-\eqref{q2bar_new} can be explicitly formulated by applying the relations \eqref{eq:general_qav}-\eqref{eq:x_qav} to evaluate the expectations over the map realizations.

Let us focus on equation \eqref{xbar_new}. Starting from equation \eqref{xbar} and applying the identity \eqref{eq:x_qav}, we obtain:
\begin{equation}
     \| \bb{\overline{x}} \|^2 = \frac{1}{\pi} \int Dz \int_{0}^{\pi} d\theta \ \| \bb{\overline{x}}\| \cos\theta \  \sigma \lr{\beta h \lr{z,\theta}}.
\end{equation}

Dividing both sides by \( \| \bb{\overline{x}} \| \) and setting \( \| \bb{\overline{x}} \| = \overline{x} \), we recover equation \eqref{xbar_new}.
\end{proof}

After proving the corollary, we proceed with the numerical resolution of equations \eqref{xbar_new}--\eqref{q2bar_new}, using the internal field defined in \eqref{campo}, in order to construct the phase diagram(s) of the model reported in the main text.




\refstepcounter{section}
\section*{Appendix Two: Replica trick for dense networks of place cells}\label{sec:replica}

In a nutshell, the replica trick allows to compute the free energy by the formula
\begin{equation}\label{eq:RT}
   \mathcal{A}(\alpha,\beta)=\lim_{N\to \infty} N^{-1}\expect \ln Z= \lim_{N\to \infty} \lim_{n \to 0} \frac{\expect Z^n -1}{nN},
\end{equation}
that is, it allows to avoid the computation of the logarithm of the partition function by dealing with the quantity $\expect Z^n$, that  reads accordingly to the next
\begin{proposition}
As different replicas are coupled together trough the product space of their relative Boltzmann measures, the quenched expectation of the $n$ moments of the partition function $\expect Z^n$ can be written as  
\begin{align}
    \expect Z^n=\expect \lr{\prod_{a=1}^n\sum_{\bb s_a}}\exp{\slr{\frac{\beta p!}{N^{p-1}} \sum_{\mu=1}^K\sum_{a=1}^n\sum_{i_1<..<i_p=1}^N\lr{\boldsymbol\eta^\mu_{i_1}\cdot \boldsymbol\eta^\mu_{i_2}}\dots\lr{\boldsymbol\eta^\mu_{i_{p-1}}\cdot \boldsymbol\eta^\mu_{i_p}}\  s^a_{i_1}\dots s^a_{i_p}}}.
\end{align}
\end{proposition}
As stated, we now divide the contribution of the signal term (chosen by $\mu=1$ with no loss of generality) from the noise term given by all the remaining $\mu>1$ maps as stated by the next  
\begin{proposition}
The signal-to-noise distinction among the various contribution to the Cost function of the dense network results in a signal  that reads as
\begin{align}
    \mathcal S: \:\: \sum_{i_1<..<i_p} \lr{\boldsymbol\eta^1_{i_1}\cdot \boldsymbol\eta^1_{i_2}}\dots\lr{\boldsymbol\eta^1_{i_{p-1}}\cdot \boldsymbol\eta^1_{i_p}} s^a_{i_1}\dots s^a_{i_p} \sim \frac{1}{p!}\sum_{i_1,..,i_p} \lr{\boldsymbol\eta^1_{i_1}\cdot \boldsymbol\eta^1_{i_2}}\dots\lr{\boldsymbol\eta^1_{i_{p-1}}\cdot \boldsymbol\eta^1_{i_p}} s^a_{i_1}\dots s^a_{i_p}+\mathcal O(N^{p-1})
\end{align}
and where we approximated $\sum_{i_1<..<i_p}\sim\frac{1}{p!}\sum_{i_1,..,i_p}$ plus contributions that vanish in the thermodynamic limit. As a result, overall, the contribution of the signal to $\expect Z^n$ is 
\begin{align}
    \expect_{\boldsymbol \eta^1} \int \slr{\prod_{a=1}^n\frac{d^d x^1_a d^d \hat x^1_a}{2\pi/N}}\exp\lr{iN\sum_a \bb x^1_a\cdot \bb{\hat x}^1_a - i\sum_{a,i}\bb{\eta}^1_i\cdot \bb{\hat x}^1_a \:s_i^a + \beta N\sum_a\lr{\bb{x}^1_a}^p},
\end{align}
where $\lr{\bb{x}^1_a}^p=\bb{x}^1_a \cdot \bb{x}^1_a \dots \bb{x}^1_a \cdot \bb{x}^1_a$ is the dot product of $p$ terms (we remind that $p$ is even).

\bigskip

The signal-to-noise distinction  among the various contribution to the Cost function of the dense network gives rise to a slow-noise contribution $\mathcal N$ to $\expect Z^n$ that reads as
\begin{align}
    \nonumber
    \expect &\exp\lr{\frac{\beta p!}{N^{p-1}}\sum_{a,\mu>1} \sum_{i_1<..<i_p}\boldsymbol\eta^\mu_{i_1}\cdot \boldsymbol\eta^\mu_{i_2}\dots\boldsymbol\eta^\mu_{i_{p-1}}\cdot \boldsymbol\eta^\mu_{i_p}\:\:s^a_{i_1}\dots s^a_{i_p}}\sim\\\nonumber
    &\sim 1+\frac{\beta p!}{N^{p-1}}\sum_{a,\mu>1} \sum_{i_1<..<i_p}\expect\slr{\boldsymbol\eta^\mu_{i_1}\cdot \boldsymbol\eta^\mu_{i_2}\dots\boldsymbol\eta^\mu_{i_{p-1}}\cdot \boldsymbol\eta^\mu_{i_p}}\:\:s^a_{i_1}\dots s^a_{i_p} +\\
    &+\frac{1}{2}\lr{\frac{\beta p!}{N^{p-1}}}^2\expect\slr{\lr{\sum_{a,\mu>1} \sum_{i_1<..<i_p}\boldsymbol\eta^\mu_{i_1}\cdot \boldsymbol\eta^\mu_{i_2}\dots\boldsymbol\eta^\mu_{i_{p-1}}\cdot \boldsymbol\eta^\mu_{i_p}\:\:s^a_{i_1}\dots s^a_{i_p}}^2}.\label{eq:quadterm}
\end{align}
\end{proposition}
The linear term in the expansion vanishes since the maps have different indices (given the condition on the sum $i_1<..<i_p$). Let us now focus on the quadratic term. The expectation that appears in it can be expanded as follows
\begin{align*}
    \sum_{a,b}\sum_{\mu,\nu>1}\sum_{i_1<..<i_p}\sum_{j_1<..<j_p}\expect\slr{ \boldsymbol\eta^\mu_{i_1}\dots\boldsymbol\eta^\mu_{i_p} \:\boldsymbol\eta^\nu_{j_1}\dots\boldsymbol\eta^\nu_{j_p}}\:\:s^a_{i_1}\dots s^a_{i_p}\:\:s^b_{j_1}\dots s^b_{j_p}\sim\\
    \sim \frac{1}{(p!)^2} \sum_{a,b}\sum_{\mu,\nu>1}\sum_{i_1,..,i_p}\sum_{j_1,..,j_p}\expect\slr{ \boldsymbol\eta^\mu_{i_1}\dots\boldsymbol\eta^\mu_{i_p} \:\boldsymbol\eta^\nu_{j_1}\dots\boldsymbol\eta^\nu_{j_p}}\:\:s^a_{i_1}\dots s^a_{i_p}\:\:s^b_{j_1}\dots s^b_{j_p}.
\end{align*}
Now, let us compute the expectation $\expect\slr{ \boldsymbol\eta^\mu_{i_1}\dots\boldsymbol\eta^\mu_{i_p} \:\boldsymbol\eta^\nu_{j_1}\dots\boldsymbol\eta^\nu_{j_p}}$. We consider the case $p=4$ specifically now: we write
\begin{align}\label{eq:expect1}
    \expect\slr{ \boldsymbol\eta^\mu_{i_1}\cdot\boldsymbol\eta^\mu_{i_2} \: \boldsymbol\eta^\mu_{i_3}\cdot\boldsymbol\eta^\mu_{i_4} \boldsymbol\eta^\mu_{j_1}\cdot\boldsymbol\eta^\mu_{j_2} \: \boldsymbol\eta^\mu_{j_3}\cdot\boldsymbol\eta^\mu_{j_4}}=\sum_{t_1,t_2,t_3,t_4=1}^d \expect\slr{ \eta^{\mu,t_1}_{i_1} \eta^{\mu,t_1}_{i_2} \eta^{\mu,t_2}_{i_3} \eta^{\mu,t_2}_{i_4} 
    \eta^{\nu,t_3}_{j_1} \eta^{\nu,t_3}_{j_2} \eta^{\nu,t_4}_{j_3} \eta^{\nu,t_4}_{j_4} },
\end{align}
where we have expanded the dot products by writing them explicitly by means of the indices $t_1,t_2,t_3,t_4=1,..,d$ that run over the components of the vectors. The leading terms in $N$ are those for $i_1\neq i_2 \neq i_3 \neq i_4$ and $j_1\neq j_2 \neq j_3 \neq j_4$. Focusing on these terms, we can rewrite eq. \eqref{eq:expect1} as
\begin{align}\label{eq:expect2}
    4!\sum_{t_1,t_2,t_3,t_4=1}^d \expect\slr{ \eta^{\mu,t_1}_{i_1} \eta^{\nu,t_3}_{j_1}}\expect\slr{ \eta^{\mu,t_1}_{i_2} \eta^{\nu,t_3}_{j_2}}\expect\slr{ 
    \eta^{\mu,t_2}_{i_3} \eta^{\nu,t_4}_{j_3}}\expect\slr{ \eta^{\mu,t_2}_{i_4} \eta^{\nu,t_4}_{j_4} },
\end{align}
where the $4!$ term accounts for the number of possible pairings among the two set of indices $\{i_1,..i_4\}$ and $\{j_1,..,j_4\}$. Each expectation appearing in the latter expression is non-zero only when $\mu=\nu$ and the site-indices $i_1,..,i_4, j_1,..,j_4$ are equal in pairs: $i_1=j_1, i_2=j_2, i_3=j_3, i_4=j_4$. We can write (neglecting the irrelevant index $\mu$ for clarity):
\begin{align}\label{eq:expect3}
    4!\delta_{\mu,\nu}\delta_{i_1,j_1}\delta_{i_2,j_2}\delta_{i_3,j_3}\delta_{i_4,j_4}\sum_{t_1,t_2,t_3,t_4=1}^d \expect\slr{ \eta^{t_1}_{i_1} \eta^{t_3}_{i_1}}\expect\slr{ \eta^{t_1}_{i_2} \eta^{t_3}_{i_2}}\expect\slr{ 
    \eta^{t_2}_{i_3} \eta^{t_4}_{i_3}}\expect\slr{ \eta^{t_2}_{i_4} \eta^{t_4}_{i_4} }.
\end{align}
Now each expectation appearing in the latter expression is just the covariance of the vectors $\bb \eta_i$ over the space of their components $t=1,..,d$. In the thermodynamic limit we can approximate such covariance with the $d-$identity matrix $\frac{1}{d}\bb 1_d$, or in other words we can assume that, at leading order in $N$ we have $\expect\slr{ \eta^{t_1}_{i} \eta^{t_2}_{i}}=\frac{1}{d}\delta_{t_1,t_2}$, where the factor $d^{-1}$ in front is there to account for the normalization of the maps $\bb \eta$, \emph{i.e.} $\bb \eta_i \cdot \bb \eta_i=1$. The only non-zero contributions (in the thermodynamic limit) in eq. \ref{eq:expect3} are those for which $t_1=t_2=t_3=t_4$ and $t_1=t_3, t_2=t_4$, hence the latter equation has two distinct terms that read (apart from the delta's in front which we omit for clarity):
\begin{align}\label{eq:expect4}
    \nonumber
    4!\sum_{t=1}^d \expect\slr{ \eta^{t}_{i_1} \eta^{t}_{i_1}}\expect\slr{ \eta^{t}_{i_2} \eta^{t}_{i_2}}\expect\slr{ 
    \eta^{t}_{i_3} \eta^{t}_{i_3}}\expect\slr{ \eta^{t}_{i_4} \eta^{t}_{i_4} }+&4!
    \sum_{t_1\neq t_2}\expect\slr{ \eta^{t_1}_{i_1} \eta^{t_1}_{i_1}}\expect\slr{ \eta^{t_1}_{i_2} \eta^{t_1}_{i_2}}\expect\slr{ 
    \eta^{t_2}_{i_3} \eta^{t_2}_{i_3}}\expect\slr{ \eta^{t_2}_{i_4} \eta^{t_2}_{i_4} }=\\
    &=4!\lr{\frac{d}{d^4}+\frac{d(d-1)}{d^4}}=\frac{4!}{d^2}
\end{align}
It can be shown that, for a generic even $p$, the magnitude of the expectation in eq. \ref{eq:expect1} generalizes to $\frac{p!}{d^{p/2}}$. Hence, we have that the noise contribution (eq. \ref{eq:quadterm}) can be rewritten as
\begin{align}
    \mathcal N: \:\:1+\frac{p!K}{2d^{p/2}}\lr{\frac{\beta}{N^{p-1}}}^2\sum_{a,b}\lr{\sum_{i} s^a_{i}s^b_{i}}^p + \mathcal{O}\lr{N^{(3-p)/2}}\sim \exp\slr{\frac{p!K}{2d^{p/2}}\lr{\frac{\beta}{N^{p-1}}}^2\sum_{a,b}\lr{\sum_{i} s^a_{i}s^b_{i}}^p}.
\end{align}
The overall contribution of the noise to $\expect Z^n$ is then
\begin{align}
    \int \slr{\prod_{a,b=1}^{n,n}\frac{dq_{ab}d\hat{q}_{ab}}{2\pi/N}}\exp\slr{iN\sum_{ab}q_{ab}\hat{q}_{ab}-i\sum_{ab}\hat{q}_{ab}\sum_i s^a_i s^b_i + N\beta^2 \alpha \sum_{ab}\lr{q_{ab}}^p},
\end{align}
where 
\begin{equation}\label{AppendiStorage}
    \alpha=\frac{p!}{2d^{p/2}}\frac{K}{N^{p-1}}
\end{equation}
is the load of the model.
\begin{figure}
    \centering
    \includegraphics[width=0.35\linewidth]{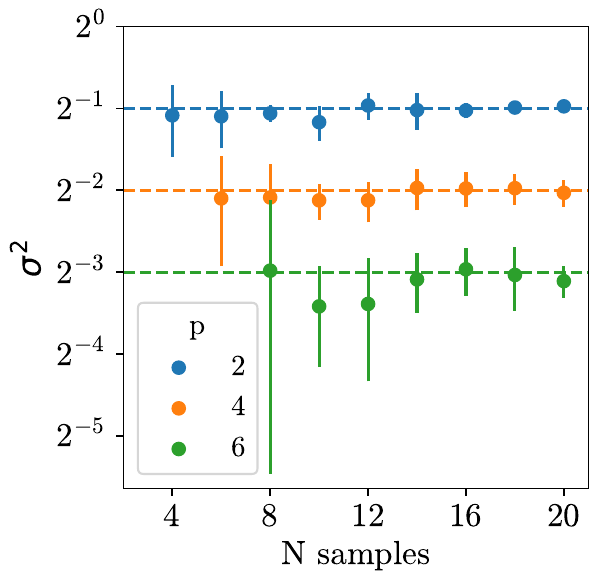}
    \includegraphics[width=0.35\linewidth]{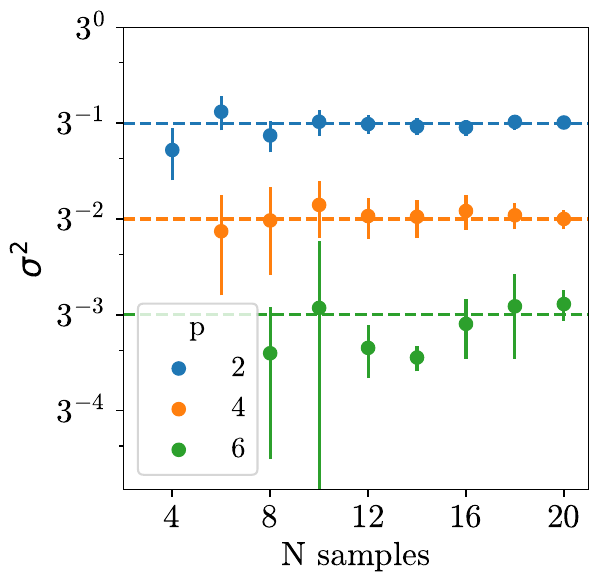}
    \caption{The variance of the product $\sigma^2=\text Var \lr{ \boldsymbol\eta_{i_1}\dots\boldsymbol\eta_{i_p}}$ for $i_1<..<i_p$ is confronted numerically with the expected value of $\frac{1}{d^{p/2}}$ for different sizes of $N$ (x-axis). In particular, we tested the cases $d=2$ (\emph{left}) and $d=3$ (\emph{right}) for $p=2,4,6$.}
    \label{fig:var}
\end{figure}
\begin{remark}
We stress that the storage of patterns in the synaptic coupling is the maximal allowed, even for such a dense network, supporting a supra-linear scaling $K \propto N^{p-1}$.    
\end{remark}
Now we are able to write the quantity $\expect Z^n$, that, after some manipulations, reads
\begin{align}\nonumber
    \expect Z^n=\int &\glr{d\bb x^1 d\bb{\hat x^1}d\bb q d\bb{\hat q}}\exp\Bigg[iN\sum_a \bb x^1_a\cdot \bb{\hat x}^1_a + \beta N\sum_a\lr{\bb{x}^1_a}^p + iN\sum_{ab}q_{ab}\hat{q}_{ab}+ N\beta^2 \alpha \sum_{ab}\lr{q_{ab}}^p +\\
    &+N\ln \sum_{\{\bb s_a\}}\expect_{\boldsymbol \eta^1} \lr{ - i\sum_{a}\bb{\eta}^1\cdot \bb{\hat x}^1_a \:s^a-i\sum_{ab}\hat{q}_{ab}s^a s^b}\Bigg]\equiv \int \glr{d\bb x^1 d\bb{\hat x^1}d\bb q d\bb{\hat q}}e^{N\phi(\bb x^1,\bb{\hat x^1},\bb q,\bb{\hat q})}
\end{align}
where $$\glr{d\bb x^1 d\bb{\hat x^1}d\bb q d\bb{\hat q}}\equiv \slr{\prod_{a=1}^n\frac{d^d x^1_a d^d \hat x^1_a}{2\pi/N}}\slr{\prod_{a,b=1}^{n,n}\frac{dq_{ab}d\hat{q}_{ab}}{2\pi/N}}$$ is the integral measure in short form and $$\sum_{\{\bb s_a\}}\equiv \prod_{a=1}^n \sum_{s_a=\{0,1\}}$$ is the partition sum over the replicated neurons $s_a$.\\
As standard in replica calculations, we exchange the two limits appearing in eq. \ref{eq:RT} and perform the thermodynamic limit first. The saddle point conditions $\pder[\phi]{x^1_a}=0$ and $\pder[\phi]{q_{ab}}=0$ give:
\begin{align}
    &\hat x^1_a=i\beta \:p\:(x^1_a)^{p-1}\\
    &\hat q_{ab}=i\alpha\beta^2 \:p\:(q_{ab})^{p-1}
\end{align}
which allow us to write
\begin{align}\nonumber
    \phi(\bb x^1,\bb q)=(1-p)\beta \sum_a(\bb x^1_a)^p+&(1-p)\alpha\beta^2\sum_{ab}(q_{ab})^p+\\
    &+\ln \expect_{\bb{\eta}}\sum_{\{\bb s_a\}} \exp\lr{\beta p\sum_a \bb{\eta}\cdot (\bb x^1_a)^{p-1}\:s_a + \alpha\beta^2p\sum_{ab}(q_{ab})^{p-1}s_a s_b}.
\end{align}
Under RS assumption we have
\begin{align}
    \bb x^1_a=\overline{\bb x},\qquad q_{ab}=\overline{q}_1\delta_{ab}+\overline{q}_2(1-\delta_{ab}),
\end{align}
thus we are now ready to state the main proposition of this Appendix, namely
\begin{proposition}
In the thermodynamic limit, the replica-symmetric quenched free energy of the dense Battaglia-Treves model, equipped with McCulloch $\&$ Pitts neurons as described by eq. \eqref{eq:ham}, for the $S^{d-1}$ embedding space, can be expressed in terms of the (mean values of the) order parameters $\overline{\bb x}$, $\overline{q}_1$,  $\overline{q}_2$ and the control parameters $\alpha$, $\beta$ (keeping $\lambda=1$), as follows:
\begin{align}\nonumber
    \mathcal{A}(\alpha,\beta,\lambda=1)=&(1-p)\beta \overline{\bb x}^p+(1-p)\alpha\beta^2 (\overline{q}_1^p-\overline{q}_2^p)+\\
    &+\expect_{\bb{\eta}}\int Dz\ln\lr{1+ \exp\lr{\beta p \:\bb{\eta}\cdot \overline{\bb x}^{p-1} + \alpha\beta^2p \:(\overline{q}_1^{p-1}-\overline{q}_2^{p-1})+\beta\sqrt{2\alpha p\:\overline{q}_2^{p-1}}\:z }},
\end{align}
where the order parameters must assume values that extremize the above expression to ensure Thermodynamics prescriptions to hold.
\end{proposition}


\begin{remark}
In the replica derivation we adopt a slight abuse of notation; accordingly, powers such as $\overline{\bb x}^{\,p-1}$ are understood as $\|\overline{\bb x}\|^{\,p-2}\overline{\bb x}$, so that 
$\bb{\eta}\!\cdot\!\overline{\bb x}^{\,p-1}=\|\overline{\bb x}\|^{\,p-2}(\overline{\bb x} \!\cdot\!\bb{\eta})$ and $\overline{\bb x}^{\,p}=\|\overline{\bb x}\|^{\,p}$.
Under this convention, the replica-symmetric free energy reported above exactly coincides with the one previously obtained via the interpolation method (Eq.~\eqref{ARS} for $\lambda=1$). 
Therefore, it is not necessary to re-derive the self-consistency equations, as they have already been obtained within the interpolation framework and hold unchanged here.

\end{remark}



\section{Acknowledgments}
The authors are grateful to the PRIN 2022 grants (a) {\em Statistical Mechanics of Learning Machines: from algorithmic and information theoretical limits to new biologically inspired paradigms} n. 20229T9EAT funded by European Union - Next Generation EU and (b)  {\em “Stochastic Modeling of Compound Events (SLIDE)”} n. P2022KZJTZ funded by the Italian Ministry of University and Research (MUR) in the framework of European Union - Next Generation EU.
\newline
A.B. acknowledges fundings also by Sapienza University of Rome via the grant {\em Statistical learning theory for generalized Hopfield models}.
\newline
A.B. and D.T. are members of the INdAM's group GNFM which is acknowledged too.

\bibliographystyle{elsarticle-num}
\bibliography{biblio}

\end{document}